\documentclass{article}

\usepackage{graphicx,tikz,amsmath,amsfonts,amsthm,amssymb, paralist,xcolor,slashed,accents,esint,mathrsfs}%
\usepackage[utf8]{inputenc}%
\usepackage[center,it]{titlesec}%
\usepackage[affil-it]{authblk}%
\usepackage[a4paper, total={6.7in, 9in}]{geometry}%
\usepackage[title]{appendix}
\usepackage{hyperref}
\hypersetup{
    colorlinks=true,
    linkcolor=blue,
    filecolor=magenta,      
    urlcolor=blue,
    citecolor=blue,
    citebordercolor=0 1 0, 
    pdftitle={Bernhardt (2025)},
}

\newcommand{\RR}{\mathbb{R}}
\newcommand{\ZZ}{\mathbb{Z}}

\newcommand{\NN}{\mathbb{N}}
\newcommand{\TT}{\mathbb{T}}

\newcommand{\Ric}{\mathrm{Ric}}
\newcommand{\osc}{\mathrm{osc}}
\newcommand{\av}{\mathrm{av}}
\newcommand{\de}{\mathrm{d}}

\newcommand{\mc}{\mathcal}
\newcommand{\mf}{\mathfrak}
\newcommand{\ol}{\overline}
\newcommand{\wh}{\widehat}
\newcommand{\wt}{\widetilde}

\newcommand{\lesa}{\lesssim}
\newcommand{\ve}{\varepsilon}
\newcommand{\pa}{\partial}
\newcommand{\bs}{\backslash}
\newcommand{\ra}{\rightarrow}

\title{Future stability of solutions of the Einstein-nonlinear scalar field system with decelerated expansion}
\author{Louie Bernhardt}
\affil{\small University of Melbourne, School of Mathematics and Statistics, Parkville~VIC~3010, Australia}
\date{}

\begin{document}
\numberwithin{equation}{section}
\newtheorem{theorem}{Theorem}[section]
\newtheorem{proposition}[theorem]{Proposition}
\newtheorem{corollary}[theorem]{Corollary}
\newtheorem{lemma}[theorem]{Lemma}
\theoremstyle{definition}
\newtheorem{definition}[theorem]{Definition}
\theoremstyle{remark}
\newtheorem{remark}[theorem]{Remark}
\maketitle
\begin{abstract}
    We study solutions to the Einstein equations coupled to a nonlinear scalar field with exponential potential. This system admits Friedmann-Lema\^itre-Robertson-Walker solutions undergoing decelerated expansion, with $\TT^3$ spatial topology and scale factor $a(t) = t^p$ for $1/3 < p < 1$. For each $p \in (2/3,1)$, we prove that the corresponding FLRW spacetime is future-stable as a solution to the Einstein-nonlinear scalar field system. Given initial data on a spacelike hypersurface that is sufficiently close to the FLRW data, we show the resulting solution is future-causal geodesically complete, and remains close to the FLRW solution for all time. Moreover, we show the perturbed metric components and scalar field converge to spatially homogeneous functions as $t \ra \infty$. A key feature of our analysis is the decomposition of the metric and scalar field perturbations into their spatial averages and oscillatory remainders with zero average.
\end{abstract}
\small\tableofcontents
\normalsize\section{Introduction}
The Friedmann-Lema\^itre-Robertson-Walker (FLRW) spacetimes are a family of Lorentzian manifolds that are fundamental objects in the study of cosmology \cite{Fri:FLRW:22,Rin:cosmologyreview:25}. They are spatially homogeneous and isotropic solutions to the Einstein equations, typically coupled to some matter model. FLRW spacetimes with flat spatial geometry have a metric of the form
\begin{equation}\label{sec:intro.eq:FLRWmetric}
    \wt{g} = -\de t^2 + a^2(t) \delta_{ij}\de x^i \de x^j,
\end{equation}
where the scale factor $a:[t_0,\infty) \ra \RR^+$ is a function of time coordinate $t$. If $a$ is strictly increasing then we say the spacetime is \emph{expanding}. In this paper we consider solutions to Einstein's equations that are undergoing \emph{decelerated} expansion, which for the metric \eqref{sec:intro.eq:FLRWmetric} corresponds to the case $a'(t) > 0$, $a''(t) < 0$. A simple model for spacetimes undergoing decelerated expansion are FLRW spacetimes with spatial factor $a(t) = t^p$, where the expansion exponent $p \in (0,1)$.

One mechanism by which decelerated expansion occurs is through the coupling of the Einstein equations to a scalar field with an exponential potential. The resulting Einstein-nonlinear scalar field (ENSF) system are geometric equations on a Lorentzian manifold $(\mc{M},\wt{g})$ which read
\begin{subequations}\label{sec:intro.eq:ensf}
    \begin{align}
        \Ric_{\mu\nu}[\wt{g}] - \frac{1}{2}R{}{\wt{g}}_{\mu\nu} &= T_{\mu\nu},\label{sec:intro.eq:einstein}\\
        \square_{\wt{g}}\phi - V'(\phi) &= 0\label{sec:intro.eq:scalarfield},
    \end{align}
\end{subequations}
where $\Ric_{\mu\nu}$ is the Ricci curvature tensor of $\wt{g}$, $R = \wt{g}^{\alpha\beta}\Ric_{\alpha\beta}$ is the Ricci scalar, 
\begin{equation}
    T_{\mu\nu} = \pa_\mu \phi \pa_\nu \phi - \frac{1}{2}\wt{g}_{\mu\nu}(\wt{g}^{\alpha\beta}\pa_\alpha \phi \pa_\beta \phi + 2V(\phi))
\end{equation}
is the stress-energy momentum tensor of the scalar field $\phi$, and $\square_{\wt{g}} := \wt{g}^{\alpha\beta}\pa_\alpha\pa_\beta - \wt{g}^{\alpha\beta}\Gamma_{\alpha\beta}^\lambda(\wt{g})$ is the scalar wave operator with respect to $\wt{g}$. The function $V(\phi)$ is an exponential potential of the form
\begin{equation}\label{sec:intro.eq:potential}
    V(\phi) = V_0 \exp(-(2/p)^{1/2}\,\phi),
\end{equation}
where $V_0$, $p$ are constants such that $V_0 > 0$, $p > 1/3$.
The equations \eqref{sec:intro.eq:ensf} describe the evolution of the metric $\wt{g}$ and scalar field $\phi$ on $\mc{M}$. We consider these equations in the physically relevant $3+1$ spacetime dimensions. Our analysis does not rely in a fundamental way on this assumption, however, and we expect a suitable extension of the results in this paper to hold for general $n+1$ dimensions with $n \geq 3$.

For fixed $p > 1/3$, the Einstein-nonlinear scalar field system with potential \eqref{sec:intro.eq:potential} admits the expanding FLRW solution $((0,\infty)\times \TT^3,\wt{g}_b,\phi_b)$ given by
\begin{subequations}\label{sec:intro.eq:background}
    \begin{align}
        \wt{g}_b &= -\de t^2 + t^{2p} \delta_{ij}\,\de x^i \de x^j,\label{sec:intro.eq:backgroundmetric}\\
        \phi_b &= (2p)^{1/2} \log t + (p/2)^{1/2}\log\Big|\frac{V_0}{p(3p-1)}\Big|.\label{sec:intro.eq:backgroundscalarfield}
    \end{align}
\end{subequations}
\begin{remark}[Expansion exponents $p \leq 1/3$]
    One can also consider FLRW spacetimes with expansion exponent $p \in (0,1/3]$. The FLRW spacetime with $0 < p < 1/3$ solves the ENSF system with potential \eqref{sec:intro.eq:potential}, however the constant $V_0$ must be negative. Meanwhile the spacetime with $p = 1/3$ is a solution to the Einstein-(linear) scalar field system.
\end{remark}
\begin{remark}[Spatial topology]
    One can take the FLRW solution \eqref{sec:intro.eq:background} to have either $\TT^3$ or $\RR^3$ spatial topology. In this paper we consider the spatially compact $\TT^3$ case, and we fix the torus to be $\TT^3 := [-\pi,\pi]^3$ with the ends identified.
\end{remark}

The Einstein-nonlinear scalar field system with exponential potential has thus far been studied in the regime of \emph{accelerated expansion}, with spacetime exponent $p > 1$. Originally of interest as a model for power-law inflation in the early universe (see \cite{GroPetRin:quescientbigban:23} for example), much work has also been carried out in the future dynamics of solutions as $t \ra \infty$. The evolution of spatially homogeneous, isotropic solutions to the ENSF system was studied by Halliwell in \cite{Hal:ENSFFLRW:87}. In \cite{HeiRen:powerlawstab:07} Heinzle and Rendall investigated the \emph{stability} of FLRW solutions to the Einstein-nonlinear scalar field system outside of symmetry. For a discrete set of expansion exponents $p_n >1$ such that $p_n \searrow 1$, they showed that initially small perturbations of the corresponding FLRW solution are future geodesically complete, and remain close the the FLRW solution for all time. This result was extended by Ringstr\"om in \cite{Rin:powerlaw:09} to show stability of the FLRW solution for all $p > 1$. See also \cite{KitMae:cosmicnohair:92,LuoIse:powerlawmag:13,GroPetRin:quescientbigban:23}. 

A key difference between the regimes of accelerated and decelerated expansion is that their causal structures are very different. For spacetimes undergoing accelerated expansion, two late-time observers can eventually become causally disconnected, with their future domains of communication becoming disjoint after sufficient time.\footnote{This can also be interpreted as the existence of a \emph{cosmological horizon}, a boundary around any single observer, beyond which a large part of the spacetime is permanently hidden.} These properties assist in the analysis of such spacetimes, and in some cases turn global problems into entirely local ones \cite{Fri:deSitterstab:86}.

In contrast, such localisation does not occur in spacetimes with decelerated expansion. In the above FLRW solutions with $p \leq 1$, all observers ``see'' all other observers for all time. More concretely, their future domains of communication never become disjoint, as is depicted in Figure~\ref{sec:intro.fig:FLRWpicture}. This leads to more complicated behaviour of perturbations of these spacetimes. Additionally, it has been shown that the expansion rate of a given spacetime is closely related to the asymptotics of perturbations of that spacetime \cite{Rin:powerlaw:09}. In particular, more expansion generally yields more decay for perturbations and/or their derivatives.

This all raises the question of whether decelerated FLRW spacetimes are stable. In this paper, we show that these decelerated spacetimes are indeed stable in the range of expansion exponents $2/3 < p < 1$. We now give an abridged version of our main result; see \textbf{Theorem~\ref{sec:globex.thm:globalexistence}} for a full statement.
\begin{theorem}[Future nonlinear stability of decelerated solutions of the Einstein-nonlinear scalar field system]\label{sec:intro.thm:main}
Fix the spacetime exponent $2/3 < p < 1$, and let $((0,\infty)\times \TT^3,\wt{g}_b,\phi_b)$ be the FLRW spacetime \eqref{sec:intro.eq:background} that solves the Einstein-nonlinear scalar field system \eqref{sec:intro.eq:ensf} with potential \eqref{sec:intro.eq:potential}. This solution is globally future-stable under small perturbations. More specifically, each solution $(\mc{M},\wt{g},\phi)$ to the ENSF system with initial data that is sufficiently close to that of $(\wt{g}_b,\phi_b)$ has a maximal globally hyperbolic development which is future-causal geodesically complete. 

The metric components $\wt{g}_{tt}$, $\wt{g}_{it}:i=1,2,3$, and scalar field $\phi$ satisfy the asymptotic bounds
\begin{equation}\label{sec:intro.eq:asymptotics1}
    |\wt{g}_{tt} + 1|\lesa \frac{1}{t^{3p-2}}, \qquad \frac{1}{t^p}|\wt{g}_{it}| \lesa \frac{1}{t^{2p-1-\delta}},\qquad \frac{1}{\log t}|\phi - \phi_b| \lesa \frac{1}{t^{3p-2+\delta}},
\end{equation}
where $\delta > 0$ is some parameter that can in principle be made arbitrarily small. For the horizontal metric components $\wt{g}_{ij}:i,j=1,2,3$,
there exists a symmetric, positive-definite $3\times 3$ matrix $(g_\infty)_{ij}$ that is close to the identity matrix, and satisfies
\begin{equation}\label{sec:intro.eq:asymptotics2}
    |t^{-2p}\wt{g}_{ij} - (g_{\infty})_{ij}| \lesa \frac{1}{t^{2p-1-\delta}}.
\end{equation}
\end{theorem}
\begin{remark}
    Theorem~\ref{sec:intro.thm:main} is a proof of \emph{future} stability, meaning we only prove stability in the expanding direction as $t \ra \infty$. The direction $t \ra 0$ corresponds to the study of power-law big bang singularities, we refer the interested reader to \cite{GroPetRin:quescientbigban:23}.
\end{remark}
\begin{remark}[Leading-order asymptotics]
    In the asymptotic estimates \eqref{sec:intro.eq:asymptotics1}, \eqref{sec:intro.eq:asymptotics2}, the various rescalings of the metric components $\wt{g}_{it}$, $\wt{g}_{ij}$ and the scalar field $\phi$ normalise the leading-order behaviour for the perturbations.\footnote{This is not immediately clear for the metric shift components $\wt{g}_{it}$, as their counterparts for the background solution $(\wt{g}_b)_{it}$ are identically zero. The rescaling comes about from a change to a conformal time coordinate $\tau = t^{1-p}$, which rescales all metric components to have the same leading-order asymptotic behaviour. Boundedness of the metric components in the conformal time coordinate translates back to boundedness of the quantities $t^{-p}\wt{g}_{it}$. See Remark~\ref{sec:globex.rmk:asymptotics} for more details.} It follows that $\wt{g}_{tt}$, $\wt{g}_{it}$, and $\phi$ all converge at leading order to their background FLRW counterparts. Meanwhile Theorem~\ref{sec:intro.thm:main} shows that the rescaled horizontal metric components $t^{-2p}\wt{g}_{ij}$ converge not to their background counterparts $\delta_{ij}$, but to a nearby metric $(g_\infty)_{ij}$ that is spatially homogeneous. This means that to leading order, the perturbed metric converges to a spatially homogeneous (but not isotropic) metric that is close to the background FLRW solution.
\end{remark}
\begin{remark}[Higher-order asymptotics]
    Theorem~\ref{sec:intro.thm:main} does not tell us whether there exists an asymptotic expansion for the perturbed metric or scalar field (past leading-order). In fact, we do not expect such an expansion to exists, as an aspect of the decelerated setting is higher-order oscillatory behaviour. This behaviour plays an important role in our analysis which we will explain further in Section~\ref{sec:intro.subsec:proof}.
\end{remark}

We discuss Theorem~\ref{sec:intro.thm:main} and its broader context in more detail in the remainder of this introduction. In Section~\ref{sec:intro.subsec:related}, we discuss previous work on related problems, and in Section~\ref{sec:intro.subsec:proof} we comment on the analysis carried out in this paper to prove Theorem~\ref{sec:intro.thm:main}. Finally we explain the breakdown of our proof for expansion exponents $p \leq 2/3$ in Section~\ref{sec:intro.subsec:breakdown}.
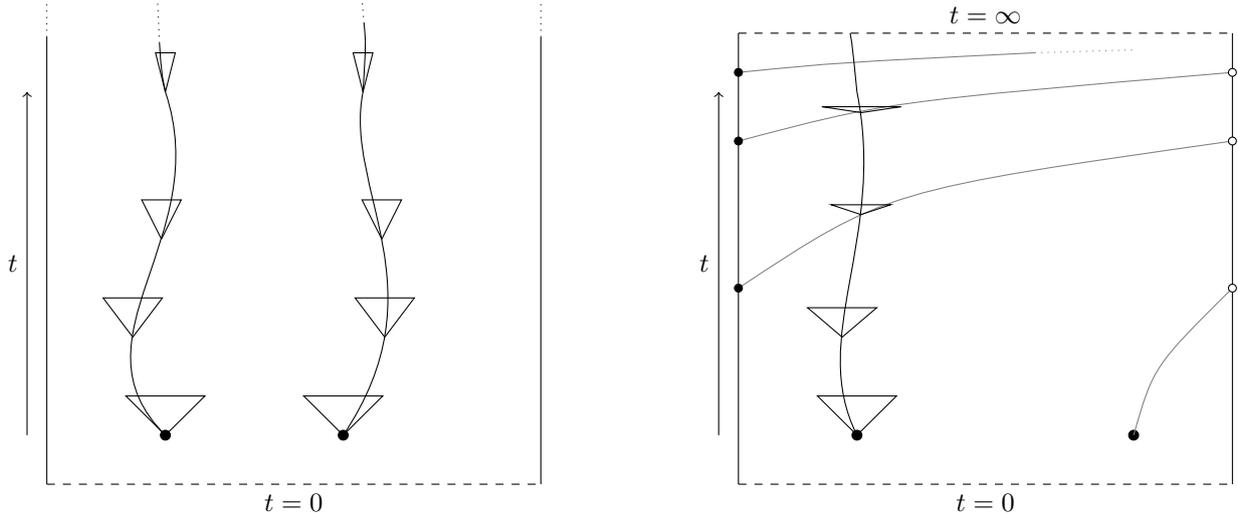
\begin{figure}\label{sec:intro.fig:FLRWpicture}
    \centering
    \begin{tikzpicture}[scale=1.3]
            \draw[dashed] (0,0) -- node[below]{$t = 0$}(5,0);
            \draw (0,0) --  (0,4.56);
            \draw[->] (-0.2,0.5) -- node[left]{$t$} (-0.2,4);
            \draw (5,0) -- (5,4.56);
            \draw[dotted] (0,4.56) -- (0,4.95);
            \draw[dotted] (5,4.56) -- (5,4.95);
            
            \filldraw (1.2,0.5) circle (0.05cm);
            \draw (1.2,0.5) .. controls (0.2,1.5) and (1.7,2.5) .. (1.2,4);
            \draw (1.2,4) .. controls (1.16,4.2) .. (1.14,4.5);
            \draw[dotted] (1.14,4.5) -- (1.12,4.95);
            \draw (1.2,0.5) -- (0.8,0.9) -- (1.6,0.9) -- (1.2,0.5);
            \draw (0.87,1.5) -- (0.57,1.9) -- (1.17,1.9) -- (0.87,1.5);
            \draw (1.16,2.5) -- (0.96,2.9) -- (1.36,2.9) -- (1.16,2.5);
            \draw (1.2,4) -- (1.1,4.4) -- (1.3,4.4) -- (1.2,4);

            \filldraw (3,0.5) circle (0.05cm);
            \draw (3,0.5) .. controls (4,2) and (3,3) .. (3.2,4);
            \draw (3.2,4) .. controls (3.23,4.5) .. (3.21,4.7);
            \draw (3,0.5) -- (2.6,0.9) -- (3.4,0.9) -- (3,0.5);
            \draw (3.42,1.5) -- (3.12,1.9) -- (3.72,1.9) -- (3.42,1.5);
            \draw (3.39,2.5) -- (3.19,2.9) -- (3.59,2.9) -- (3.39,2.5);
            \draw (3.2,4) -- (3.1,4.4) -- (3.3,4.4) -- (3.2,4);
            \draw[dotted] (3.21,4.7) -- (3.19,4.95);
            
            \draw[dashed] (7,0) -- node[below]{$t = 0$}(12,0);
            \draw (7,0) --  (7,4.6);
            \draw[->] (6.8,0.5) -- node[left]{$t$} (6.8,4);
            \draw (12,0) -- (12,4.6);
            \draw[dashed] (7,4.6) -- node[above]{$t = \infty$} (12,4.6);

            \filldraw (11,0.5) circle (0.05cm);
            \draw[gray] (11,0.5) .. controls (11.2,1.2) .. (12,2);
            \filldraw[fill=white] (12,2) circle (0.04cm);
            \draw[gray] (7,2) .. controls (8.5,3) .. (12,3.5);
            \filldraw (7,2) circle (0.04cm);
            \filldraw[fill=white] (12,3.5) circle (0.04cm);
            \draw[gray] (7,3.5) .. controls (8.5,3.9) .. (12,4.2);
            \filldraw (7,3.5) circle (0.04cm);
            \filldraw[fill=white] (12,4.2) circle (0.04cm);
            \draw[gray] (7,4.2) .. controls (8,4.3) .. (10,4.4);
            \filldraw (7,4.2) circle (0.04cm);
            \draw[gray,dotted] (10,4.4) --(11,4.43);

            \filldraw (8.2,0.5) circle (0.05cm);
            \draw (8.2,0.5) .. controls (7.7,1.5) and (8.5,2.5) .. (8.2,4);
            \draw (8.2,4) .. controls (8.15,4.5) .. (8.13,4.6);
            \draw (8.2,0.5) -- (7.8,0.9) -- (8.6,0.9) -- (8.2,0.5);
            \draw (8.05,1.5) -- (7.7,1.8) -- (8.4,1.8) -- (8.05,1.5);
            \draw (8.24,2.75) -- (7.94,2.85) -- (8.54,2.85) -- (8.24,2.75);
            \draw (8.245,3.79) -- (7.845,3.85) -- (8.645,3.85) -- (8.245,3.79);
        \end{tikzpicture}
    \caption{\small{The left diagram depicts (a portion of) a decelerated FLRW spacetime in $(x,t)$ coordinates, along with two causal observers. Due to the $\TT^3$ spatial topology, the left and right boundaries are identified. The relative velocity of the two observers is proportional to $t^{-p}$, which leads to the shrinking of the light cones. However, as shown in the compactified diagram on the right, this relative velocity does not decrease at a fast enough rate to localise the observers, and so so the two observers never become causally disconnected. In fact, observers can potentially traverse the entire spatial topology \emph{infinitely many times} along certain causal paths, as depicted by the grey observer. This matches the behaviour of geodesics in flat Minkowski space with $\TT^3$ spatial topology.}}
\end{figure}
\subsection{Related work}\label{sec:intro.subsec:related}
In the following we discuss some of the previous results in the literature, and how they relate to the present work.
\paragraph{Spacetimes undergoing decelerated expansion.}
Theorem~\ref{sec:intro.thm:main} is the first stability statement for decelerated spacetimes without additional symmetry assumptions. The recent work \cite{Tay:decelFLRWstab:24} studies the decelerated solutions to the Einstein-massless Vlasov system in spherical symmetry. The Einstein-massless Vlasov system describes the dynamics of an ensemble of massless, gravitationally attractive particles, and admits spatially homogeneous FLRW solutions with scale factor $a(t) = t^{1/2}$. In \cite{Tay:decelFLRWstab:24} spherically symmetric perturbations of these FLRW spacetimes (with $\RR^3$ spatial topology) are shown to be future-stable. We note that the noncompact spatial topology leads to additional decay via dispersion, and so we expect the spatially non-compact setting to have improved stability in comparison to the spatially compact setting. 

The stability of spacetimes undergoing decelerated expansion has also been investigated numerically, in the context of the Einstein-Euler system \cite{Mar:instabilityFLRW:25}. This system models a self-gravitating perfect fluid, and admits an FLRW solution with $\TT^3$ spatial topology undergoing decelerated expansion, provided the fluid satisfies a linear equation of state $P = c_s^2\rho$, where $P$ is the fluid pressure, $\rho$ is the fluid density, and $c_s$ is the fluid sound speed. For a given sound speed $c_s \in [0,1]$, the corresponding FLRW solution has a fixed expansion exponent in the range $1/3 \leq p \leq 2/3$.\footnote{The precise relation between the sound speed $c_s$ and expansion exponent $p$ is $p = \frac{2}{3}(1+c_s^2)^{-1}$, from which one can see that if $c_s^2 \in [0,1]$ then $p \in [1/3,2/3]$. Some important equations of state are (i) the \emph{dust} model with $c_s = 0$, $p = 2/3$, the \emph{radiation} model with $c_s=1/\sqrt{3}$, $p=1/2$, and the \emph{stiff fluid} model with $c_s=1$, $p=1/3$.} Numerical simulations carried out in \cite{Mar:instabilityFLRW:25} indicate that for all sound speeds $c_s \in [0,1]$, the FLRW solutions with $\TT^3$-spatial topology are \emph{unstable} as solutions to the Einstein-Euler system, with shocks forming in finite time for arbitrarily small perturbations of the FLRW solutions. We discuss this result and how it relates to the present paper in more detail in Section~\ref{sec:intro.subsec:breakdown}.

Related hyperbolic equations have been studied on \emph{fixed} decelerated FLRW backgrounds. The linear wave equation on FLRW spacetimes undergoing decelerated expansion has been studied in the works \cite{KlaSar:waveformulaFLRW:81,NatRos:waveformulaFLRW:23,Hag:waveFLRW:25}. The relativistic Euler equations on a fixed spacetime undergoing decelerated expansion have been considered in \cite{Spe:relEulerstab:13,Fajetal:decelEulerstab:25}.

\paragraph{Stability for linearly expanding spacetimes.}
Our result does not include the endpoint $p=1$ case corresponding to \emph{linear} expansion, meaning we do not prove the stability of linearly expanding FLRW solutions to the ENSF system. This specific problem is open, although linearly expanding solutions to the Einstein equations have been shown to be stable outside of symmetry. The first such result was \cite{AndMon:Milnestab:11}, in which Andersson and Moncrief proved that the \emph{Milne universe}, a linearly expanding spacetime with negative spatial Einstein geometry, is stable as a solution to the Einstein-vacuum equations. This has been extended in \cite{AndFaj:Milnestab:20} to include perturbations of Milne that contain Vlasov-type matter. We also mention \cite{FajOfnWya:linearduststab:24} which demonstrates stability of a class of linearly expanding FLRW solutions to the Einstein-dust system. Like the Milne model, these spacetimes also possess negative spatial curvature, and indeed this property is used in a fundamental way in all three papers. The spacetimes for which we prove stability do not have negative spatial geometry, and so the methods of proof here are different to those in \cite{AndMon:Milnestab:11,AndFaj:Milnestab:20,FajOfnWya:linearduststab:24}.
\paragraph{Stability in the accelerated expansion regime.}
In the setting of accelerated expansion, much more is understood regarding stability. We have already mentioned the works \cite{HeiRen:powerlawstab:07,Rin:powerlaw:09} which consider the stability of accelerated FLRW spacetimes to the ENSF system. We also mention Luo and Isenberg's \cite{LuoIse:powerlawmag:13}, in which Ringstr\"om's stability result is generalised by coupling the ENSF system to an electromagnetic field satisfying Maxwell's equations. Stability of the FLRW solution is then proved for the Einstein-nonlinear scalar field-Maxwell system.

Another mechanism for accelerated expansion is the inclusion of a positive cosmological constant in Einstein's vacuum equations, which produces spacetimes which expand at an \emph{exponential} rate. There is a rich and growing literature concerning the stability of solutions to the Einstein equations with positive cosmological constant, for an incomplete list see \cite{Fri:deSitterMaxwellstab:91,And:deSitteroddstab:04,LubKro:deSitterradstab11,RodSpe:irrEulerdeSitterstab:13, Spe:EulerdeSitterstab:12, HadSpe:deSitterduststab:15, Oli:EulerdeSitterstab:16,HinVas:KdSstationarystab:18,Fou:FLRWscalar:22,Fan:KdSstationarystab:25,FouMarOli:tiltedfluids:24,FouSch:KdSexpandingstab:24,HinVas:KdSexpandingstab:24,Cic:deSitterscat:24}. 

Study of the nonlinear stability of solutions to the Einstein vacuum equations with positive cosmological constant began with Friedrich's proof of stability of $3+1$-de Sitter in \cite{Fri:deSitterstab:86}. In that work, Friedrich introduced a formalism of Einstein's equations known as the conformal field equations, which reduces the global problem to an entirely local one. We also mention \cite{Rin:deSitterENSFstab08}, in which Ringstr\"om proved stability of de Sitter in general spacetime dimensions as a solution to the Einstein-nonlinear scalar field equations with a positive well-like potential. This work introduced the idea of fixing a choice of generalised wave coordinates with gauge source functions that capture the expansion of the spacetime; we will discuss this in more detail in Section~\ref{sec:intro.subsec:proof}.
\subsection{Outline of the proof of global existence }\label{sec:intro.subsec:proof}
The key components of the proof of Theorem~\ref{sec:intro.thm:main} are (i) a reduction of the Einstein-nonlinear scalar field system into a system of quasilinear wave equations by fixing a certain choice of wave coordinates, (ii) a decomposition of the metric components and scalar field into their spatial averages and oscillatory remainders, along with a derivation of decoupled systems for these quantities, (iii) ODE estimates for the spatial averages, and (iv) energy estimates for the oscillatory remainders. We discuss these steps in more detail below.

\paragraph{Wave coordinates.}
We use wave coordinates (also known as  \emph{harmonic gauge} or \emph{de Donder gauge}) in this paper \cite{Cho:Einsteinlwp:52}. This breaks the diffeomorphism invariance of the Einstein-equations and transforms them into a system of quasilinear wave equations for the metric components. The ``standard'' choice of wave coordinates is to fix coordinates $\{x^\mu\}_{\mu=0,1,2,3}$ such that $\square_{\wt{g}} x^\mu = 0$. This is equivalent to setting $\Gamma^\mu= 0$, $\mu = 0,1,2,3$, where $\Gamma^\mu = \wt{g}^{\alpha\beta}\Gamma_{\alpha\beta}^\mu(\wt{g})$ are contracted Christoffel symbols. Wave coordinates were used most prominently in Choquet-Bruhat's proof of the local well-posedness of the Einstein vacuum equations \cite{Cho:Einsteinlwp:52}. A key observation made by Choquet-Bruhat was that the wave coordinate condition $\Gamma^\mu = 0$ is propagated by the gauge-fixed Einstein vacuum equations, so that this condition is preserved provided it is satisfied by the initial data.

Wave coordinates have also proved useful in understanding the global theory of Einstein's equations. This was exemplified by Lindblad and Rodnianski's \cite{LinRod:Minkowskiwavestab10}, where they proved stability of Minkowski in wave coordinates, see also Christodoulou and Klainerman's original proof \cite{ChrKlai:minkowskistab:93}.

\paragraph{Gauge source functions.}
It has been shown that in the setting of spacetime expansion, a more general form of wave coordinates can be useful. These amount to adding ``gauge source functions'' and fixing $\Gamma^\mu = \mc{Y}^\mu$. The source functions $\mc{Y}^\mu$ can depend on the coordinate functions $x^\mu$, and the metric components $\wt{g}$, but \emph{not} on any derivatives of the metric. Gauge source functions can capture the spacetime's expansion, in a certain sense, and induce decay in the resulting system. Gauge source functions were used to this effect in Ringstr\"om's \cite{Rin:deSitterENSFstab08}, in which asymptotically de Sitter spacetimes were shown to be future-stable as solutions to the Einstein-nonlinear scalar field system with a positive well-type potential. Ringstr\"om's proof of stability of FLRW spacetimes undergoing accelerated expansion built on this, and also used a generalized form of wave coordinates \cite{Rin:powerlaw:09}. Numerous works since have applied this method in the cosmological setting; see for example \cite{RodSpe:irrEulerdeSitterstab:13,Spe:EulerdeSitterstab:12,HadSpe:deSitterduststab:15,Oli:EulerdeSitterstab:16,HinVas:KdSexpandingstab:24}.

\paragraph{The conformal metric \texorpdfstring{$g$}{g}.}
Before fixing wave coordinates, we reformulate the Einstein equations as equations for a conformally rescaled metric. To justify this, we consider a background FLRW spacetime $\wt{g}_b$ of the form given in \eqref{sec:intro.eq:backgroundmetric}, writing
\begin{equation}
    \wt{g}_b = t^{2p}\Big[-(t^{-p}\de t)^2 + \delta_{ij}\de x^i \de x^j\Big].
\end{equation}
Then we change to a conformal time coordinate $\tau$ so that $\de \tau = (1-p)t^{-p}\de t$. Integrating, it follows that
\begin{equation}
    \tau(t) = t^{1-p},
\end{equation}
and so $\tau \in (0,\infty)$. The background metric then takes the form
\begin{equation}
        \wt{g}_b = \tau^{2p/(1-p)}\big[-(1-p)^{-2}\de \tau^2 + \delta_{ij}\de x^i\de x^j\big].
\end{equation}
Thus the conformal metric $m = \tau^{-2p/(1-p)}\wt{g}_b$ is the flat Minkowski metric. We take advantage of this structure and derive equations for the conformal metric 
\begin{equation}
    g := \tau^{-2p/(1-p)}\wt{g}
\end{equation}
in $(\tau,x)$ coordinates. This rescales the leading-order asymptotics of the perturbed metric components so that they are small perturbations of their Minkowski counterparts.

\begin{remark}
    We point out that the conformal coordinate $\tau$ and associated conformal transformation do not compactify the spacetime. In particular we have $\tau(t=0) = 0$ and $\tau(t=\infty) = \infty$, so intervals of the form $[t_0,\infty)$ do not become compact intervals in $\tau$. This stands in contrast to spacetimes undergoing accelerated expansion, where the same process results in a \emph{conformal compactification} of the spacetime.\footnote{Since conformal transformations map null curves to null curves, this is another way one can show that causal observers on accelerated spacetimes can become causally disconnected, while observers on decelerated (or linearly expanding) spacetimes cannot.} 
    
    On a related note, while our setup reduces the ENSF system to equations for a conformal metric, we emphasise that our analysis does not appeal to the conformal field equations of Friedrich \cite{Fri:deSitterstab:86}. The conformal factor is a fixed function of the time coordinate, not an evolution variable. This approach can also be found in \cite{Oli:EulerdeSitterstab:16}.
\end{remark}
Additionally, we formulate the conformal system so that they constitute equations for the \emph{inverse} metric components $g^{\mu\nu}:=(g^{-1})^{\mu\nu}$. This is conceptually equivalent to analysing the lowered metric components, and simplifies the form of the resulting equations. We denote the metric and scalar field perturbations by 
\begin{equation}
    h^{\mu\nu} := g^{\mu\nu} - m^{\mu\nu},\qquad\psi := \phi - \phi_b.
\end{equation}

\paragraph{The gauge-fixed system.}
We derive the gauge-fixed system for the conformal metric and scalar field perturbations $(h^{\mu\nu},\psi)$ in \textbf{Section~\ref{sec:gaugedeqs}}. We fix wave coordinates with a particular choice of gauge source functions $\mc{Y}^\mu$ (see \eqref{sec:gaugedeqs.eq:gaugesourcechoice} for the specific choice). The resulting equations for the $(h^{\mu\nu},\psi)$ are quasilinear wave equations, schematically of the form
\begin{equation}\label{sec:intro.eq:waveeqschem}
    \square_{g(u)} u = -\frac{a}{\tau}g^{00}\pa_\tau u - \frac{b}{\tau^2}g^{00}u + \mc{N},
\end{equation}
for constants $a,b \in \RR$ and semilinear error term $\mc{N}$. The first and second term on the right hand side of the above equation are dissipative terms. Depending on their signs, they can have a damping effect, inducing decay in solutions and/or their derivatives. We emphasise though that the signs and sizes of $a,b$ are different for $h^{00}$, $h^{0i}$, $h^{ij}$ and $\psi$. See \eqref{sec:gaugedeqs.eq:wave} for the full system. Moreover, the values of the coefficients in the full system depend on the choice of gauge source functions, and the motivation for our particular choice is discussed in Remark~\ref{sec:gaugedeqs.rmk:gaugemotivation}.

For the purpose of exposition, in this section we will relate our analysis of the gauge-fixed system to the ``toy model''
\begin{equation}\label{sec:intro.eq:waveeqtoy}
    \square u = \frac{a}{\tau}\pa_\tau u + \frac{b}{\tau^2}u + F,
\end{equation}
where $\square = -\pa_\tau^2 + \Delta$ is the flat wave operator and $F$ is an inhomogeneous source term. The dynamics of the two equations \eqref{sec:intro.eq:waveeqschem}, \eqref{sec:intro.eq:waveeqtoy} are essentially the same, particularly because the perturbed metric $g$ remains globally close to the Minkowski metric for all time.
\begin{remark}[The linear wave equation on a fixed FLRW background]
    A related problem in this setting is the dynamics of solutions to the linear wave equation
    \begin{equation}\label{sec:intro.eq:linwave}
        \square_{\wt{g}_b}\phi = 0
    \end{equation}
    on a fixed decelerated FLRW background. A short computation reveals that \eqref{sec:intro.eq:linwave} is equivalent to the equation \eqref{sec:intro.eq:waveeqtoy} with vanishing $F$ and coefficients of the dissipative terms taking the values
    \begin{equation}
        (a,b) = \Big(\frac{2p}{1-p},0\Big).
    \end{equation}
    In \cite{Hag:waveFLRW:25}, the linear wave equation is studied on a fixed decelerated FLRW spacetime with topology $(0,\infty)\times \RR^3$, and all expansion exponents $p \in (0,1)$. In that paper boundedness and decay of solutions is shown using energy methods which capture the \emph{dispersion} of waves on non-compact manifolds. Wave equations with dissipative terms like \eqref{sec:intro.eq:waveeqtoy} have also been studied on $\RR^{3+1}$-Minkowski, see for example \cite{Mat:dissipativewave:77,Ues:dissipativewave:80}. We emphasise that for spatially compact spacetimes, this decay via dispersion does not occur, and so the only source of decay present is from from spacetime expansion.
\end{remark}
\paragraph{Estimating the spatial averages.}
Our analysis of the perturbed quantities $(h^{\mu\nu},\psi)$ is based on a decomposition of these functions into their spatial averages, which are obtained by integrating these scalar functions over the torus, and oscillatory remainders which have zero average. We write this decomposition like $u = u_{\av} + u_{\osc}$, where
\begin{equation}
    u_{\av}(\tau) = \frac{1}{(2\pi)^3}\int_{\TT^3} u(\tau,x)\de^3 x
\end{equation}
is the spatial average of $u$ over $\TT^3$. A key observation in our analysis is that the gauge-fixed system \eqref{sec:gaugedeqs.eq:wave} decouples (up to some decaying error) into an ODE system for the averaged quantities $(h_{\av}^{\mu\nu},\psi_{\av})$ and a system for the remainders $(h_{\osc}^{\mu\nu},\psi_{\osc})$. This allows us to analyse the behaviour of the averages and the remainders separately.
\begin{remark}[Spatial average decomposition for the relativistic Euler equations]\label{sec:intro.rmk:linwaveFLRW}
    A similar decomposition of the evolution variables into average contributions and oscillations was used in \cite{Fajetal:decelEulerstab:25} in their analysis of the relativistic Euler equations on a fixed slowly expanding background. In that paper the authors decomposed the fluid density and velocity into their averaged values and an oscillating remainder, and estimated the two separately.
\end{remark}
\begin{remark}[Relation of spatial averages to mode decomposition]
    The spatial average $u_{\av}$ of a function is precisely the zeroth mode of $u$ with respect to the Laplacian $\Delta$ on $\TT^3$:
    \begin{equation}
        u_{\av} = \langle u,e_0\rangle_{L^2(\TT^3)}e_0,
    \end{equation}
    where $e_0$ satisfies $\Delta e_0 = 0$, $\|e_0\|_{L^2(\TT^3)} = 1$. Moreover, the remainder $u_{\osc}$ is the projection of $u$ onto the eigen-functions of $\Delta$ with positive eigenvalues, which are of course $e_k = e^{ix\cdot k}$, $|k| >1$.
    
    One can introduce more general spatial means that respect the geometry of the spacetime by projecting onto the eigen-functions of the Laplace-Beltrami operator $\Delta_{\ol{g}_\tau}$, where $\ol{g}_\tau$ is the induced metric on the level sets of $\tau$. Here we stick to the spatial means on the flat torus, as they are simpler to analyse, and we are still able to close the argument with them because all spatial derivatives of the perturbations decay sufficiently fast.
\end{remark}

We now explain our strategy for estimating the averaged quantities $(h_{\av}^{\mu\nu},\psi_{\av})$, which is carried out in \textbf{Section~\ref{sec:avgsys}}. One can derive equations for the averaged metric and scalar field perturbations $(h_{\av}^{\mu\nu},\psi_{\av})$ by simply integrating the gauge-fixed system \eqref{sec:gaugedeqs.eq:wave} over the torus. In analogy we integrate the toy model \eqref{sec:intro.eq:waveeqtoy} over the torus and obtain the ODE
\begin{equation}
    -\pa_\tau^2 u_{\av}(\tau) = \frac{a}{\tau}\pa_\tau u_{\av}(\tau) + \frac{b}{\tau^2}u_{\av}(\tau) + F_\av.
\end{equation}
Ignoring the inhomogeneous term in the above equation for a moment, the fundamental solutions of this ODE are\footnote{There is also the case where $(a-1)^2 = 4b$, so that $\lambda_- = \lambda_+ = (a-1)/2$. In this case the fundamental solutions of the ODE are $\{\tau^{-\lambda_-},\tau^{-\lambda_-}\log \tau\}$. This does not occur in the equations for $\{h_{\av}^{\mu\nu},\psi_{\av}\}$, and so we ignore this case.} $\{\tau^{-\lambda_-},\tau^{-\lambda_+}\}$, where the exponents $\lambda_\pm$ are
\begin{equation}\label{sec:intro.eq:indicialroots}
    \lambda_{\pm} = \frac{1}{2}(a-1)\pm \frac{1}{2}\sqrt{(a-1)^2-4b}.
\end{equation}
If the $\lambda_{\pm}$ are real, then $\lambda_+ > \lambda_-$, and so $u_{\av} = O(\tau^{-\lambda_-})$. If $\lambda_\pm$ are complex, so that $\lambda_\pm = c \pm di$, then the functions $\tau^{-\lambda{\pm}}$ become linear combinations of $\tau^{-c}\cos(d\log \tau)$ and $\tau^{-c}\sin(d\log \tau)$, and so either way we have $u_\av = O(\tau^{-\mathrm{Re}(\lambda_-)})$.

Including the inhomogeneous term, one can subsequently derive estimates of the form
\begin{equation}
    |u_\av(\tau)| \leq C \tau^{-\mathrm{Re}(\lambda_-)} \Big(|u_\av(\tau_0)| + |\pa_\tau u_{\av}(\tau_0)|\ + \int_{\tau_0}^\tau s^{\mathrm{Re}(\lambda_-)}\|F(s)\|_{L^2}\de s\Big).
\end{equation}
It follows that if $\mathrm{Re}(\lambda_-) \geq 0$ and $F$ has sufficient decay in time, then $u_\av$ is bounded. It is clear from the identity \eqref{sec:intro.eq:indicialroots} that boundedness/decay for $u_{\av}$ occurs only for specific ranges of coefficients $a,b$. In particular, $u_{\av}(\tau)$ is bounded in general only if $a - 1 \geq 0$ and $b \geq 0$, so that $\mathrm{Re}(\lambda_-) \geq 0$.

\paragraph{Improved decay for the spatial averages from the gauge equations.}
The condition $a-1 \geq 0$ is satisfied by the equations for all the perturbed quantities $(h^{\mu\nu},\psi)$. It turns out, however, that while the condition $b \geq 0$ is fulfilled by the equations \eqref{sec:gaugedeqs.eq:wavehori}, \eqref{sec:gaugedeqs.eq:wavescalar} for the horizontal metric and scalar field perturbation $h^{ij}$, $\psi$, it is not fulfilled by the equations \eqref{sec:gaugedeqs.eq:wavelapse}, \eqref{sec:gaugedeqs.eq:waveshift} for the lapse and shift metric perturbations $h^{00}$, $h^{0i}$.

To overcome this obstruction, we derive estimates for the lapse and shift directly from the wave coordinate equations $\Gamma^\mu(g) = \mc{Y}^\mu(g)$. Expanding the contracted Christoffel symbols via the formula
\begin{equation}
    \Gamma^\mu(g) = -\pa_\lambda g^{\mu\lambda} + \frac{1}{2}g^{\mu\lambda}g_{\alpha\beta}\pa_\lambda g^{\alpha\beta},
\end{equation}
the wave coordinate equations become transport equations relating the various metric components. These take the form 
\begin{equation}
    -\pa_\tau g^{0\mu} -\pa_a g^{a\mu} + \dots = \mc{Y}^\mu(g),\quad \mu=0,1,2,3.\\
\end{equation}
One can integrate these equations over the torus too, which leads to first order ODE for the $h_{\av}^{00}$, $h_{\av}^{0i}$. These are schematically of the form
\begin{equation}
    \pa_\tau u_\av(\tau) + \frac{\eta}{\tau}u_\av(\tau) = F'_\av,
\end{equation}
where the $\eta\tau^{-1}u_\av$ arises from our particular choice \eqref{sec:gaugedeqs.eq:gaugesourcechoice} of gauge source functions, and $F'$ is an error term; see \eqref{sec:avgsys.eq:odeavglapseshiftproof} for the actual equations. This implies the estimate
\begin{equation}
    |u_{\av}(\tau)|\leq C\tau^{-\eta}\Big(|u_{\av}(\tau_0)| +\int_{\tau_0}^\tau s^\eta \|F'\|_{L^2(\TT^3)}\de s\Big).
\end{equation}
If $\eta > 0$ (and $\|F'\|_{L^2}$ decays sufficiently fast), then $u_{\av}$ decays. This is the case for the equations for $h^{00}$, $h^{0i}$. In particular, our choice \eqref{sec:gaugedeqs.eq:gaugesourcechoice} of source functions $\mc{Y}^\mu(g)$ implies that $\eta > 0$.
\begin{remark}[Obstructions to decay of the spatial averages]\label{sec:intro.rmk:bottleneck}
    The decay one can obtain for all the averaged quantities $(h_{\av}^{\mu\nu},\psi_{\av})$ obviously depend on the decay of the various nonlinear terms in the equations. This error depend on all of the functions $(h^{\mu\nu},\psi)$, including the oscillatory remainders $(h_{\osc}^{\mu\nu},\psi_{\osc})$. It turns out that the principle equations for the averages imply more decay than the principle equations for the remainders. Due to the nonlinear coupling of the systems, the decay of the remainders ``bottleneck'' the decay of the averages, reducing the amount of decay of the averaged quantities. We discuss this more in Remark~\ref{sec:avgsys.rmk:bottleneck}. An additional complication (one that is also true of the equations for the oscillatory quantities $(h_{\osc}^{\mu\nu},\psi_{\osc})$) is that the system for the averages does not diagonalise entirely. There are mixed linear terms in the equations for $h_{\av}^{00}$ $h_{\av}^{ij}$, see \eqref{sec:avgsys.eq:odeavg2}.
\end{remark}
\paragraph{Estimates for the oscillatory remainders.}
The functions $h_{\osc}^{\mu\nu} = h^{\mu\nu} - h_{\av}^{\mu\nu}$, $\psi_{\osc} = \psi - \psi_{\av}$ capture the higher-order oscillations of the perturbations, and have zero spatial average. In line with these observations, we estimate these terms using energy estimates and the  Poincar\'e inequality\footnote{See Section~\ref{sec:prelims.subsec:notation} for a description of the notation we use for differential operators and Sobolev spaces in this paper.}
\begin{equation}\label{sec:intro.sec:avgsys.eq:poincare}
    \|u_{\osc}\|_{L^2(\TT^3)}\leq C\|\ol{\pa} u_{\osc}\|_{L^2(\TT^3)}.
\end{equation}
In brief, we use energy estimates to control the top-order derivatives in the system, while the Poincar\'e inequality allows us to control lower-order terms.

The oscillatory remainders $(h_{\osc}^{\mu\nu},\psi_{\osc})$ satisfy the same wave equation as the full functions $(h^{\mu\nu},\psi)$, at least up to a nonlinear error. One can simply subtract off the equations satisfied by the spatial averages, which are estimated separately. For the toy model \eqref{sec:intro.eq:waveeqtoy}, this yields an equation of the form
\begin{equation}
    \square u_\osc = \frac{a}{\tau}\pa_\tau u_\osc + \frac{b}{\tau^2}u_\osc + F_\osc.
\end{equation}
As previously mentioned, the first two terms on the right hand side of the above equation have a damping effect. The signs already play a role in determining the behaviour of the spatial averages. This is also true for the oscillating remainders, although it turns out only the sign of $a$ plays a role. If $a > 0$, a short calculation shows that the rescaled function $v = \tau^{a/2}u_\osc$ satisfies the equation
\begin{equation}\label{sec:intro.eq:oscexplain1}
    \square v = \frac{c}{\tau^2}v + \tau^{a/2}F_\osc,
\end{equation}
where $c = b-a(a-2)/4$. We let $E_0[v]$ denote the standard wave equation energy on the torus 
\begin{equation}
    E_0[v](\tau) = \frac{1}{2}(\|\pa_{\tau} v(\tau)\|_{L^2(\TT^3)}^2 + \|\ol{\pa} v\|_{L^2(\TT^3)}^2).
\end{equation}
Then $E_0[v]$ satisfies the differential inequality
\begin{equation}\label{sec:intro.eq:oscexplain2}
    \pa_\tau E_0[v] \leq \Big(\frac{|c|}{\tau^2}\|v\|_{L^2} + \tau^{a/2}\|F\|_{L^2}\Big)(E_0[v])^{1/2}.
\end{equation}
The $\tau^{-2}\|v\|_{L^2}$ term could form an obstruction to this argument, since $c = b-a(a-2)/4$ may be nonzero,\footnote{This is the case for all equations in the system for $(h_{\osc}^{\mu\nu},\psi_\osc)$.} However, $v$ has zero spatial average, and so we use the aforementioned Poincar\'e inequality \eqref{sec:intro.sec:avgsys.eq:poincare}, and bound $\tau^{-2}\|v\|_{L^2} \leq \tau^{-2}(E_0[v])^{1/2}$. This turns the lower-order term into an error term in terms of decay, and can be controlled in a straightforward manner via a Gr\"onwall-type inequality.
 
The upshot is that the decay of the oscillatory remainders are dictated by the $\tau^{-1}\pa_\tau u_{\osc}$ term (as well as the error term $F$), and not by the $\tau^{-2} u_{\osc}$ term. For the actual system \eqref{sec:gaugedeqs.eq:wave}, we derive the relevant energy inequalities in \textbf{Section~\ref{sec:energyests}}, and use them to control the remainders in \textbf{Section~\ref{sec:oscsys}}.

\begin{remark}
    Our analysis of the system for $(h_{\osc}^{\mu\nu},\psi_{\osc})$ is less straight-forward than rescaling the quantities and deriving standard energy estimates for them. We instead define a family of weighted energies which have the same effect, but allow for more freedom in estimating mixed linear terms which appear in the system. See Section~\ref{sec:energyests} for more details.
\end{remark}
\begin{remark}[Analysis of linear wave equations]
    The decomposition of solutions that we apply to the toy model \eqref{sec:intro.eq:waveeqtoy} can also be used to analyse several linear wave equations of note.
    
    The flat wave equation on Minkowski with $\TT^3$ spatial topology has $a=b=0$, and the averaged equation is $\pa_\tau^2 u_\av = 0$, which admits solutions $u_\av(\tau) = c_1 + c_2 \tau$, for $c_1,c_2\in \RR$. It follows that generic solutions grow linearly in time.\footnote{In the setting of $\RR^3$-spatial topology, we exclude these solutions from consideration, as they are not finite-energy.} This destroys any chance of stability for nonlinear versions of this equation, an obvious example of which is the Einstein-vacuum equations with $\TT^3$ spatial topology.
    
    We compare this to the linear wave equation on a flat FLRW background (see Remark~\ref{sec:intro.rmk:linwaveFLRW}). The equation for the average is
    \begin{equation}
        \pa_\tau^2 u_\av = -\frac{2p}{1-p}\frac{1}{\tau}\pa_\tau u_\av,
    \end{equation}
    which has solutions
    \begin{equation}
        u_\av(\tau) = c_1 + \frac{c_2}{\tau^{2p/(1-p)-1}},\qquad c_1,c_2 \in \RR.
    \end{equation}
    If $p > 1/3$, then $2p/(1-p) > 1$, and so solutions are bounded in this case. Moreover, one can show using the methods discussed above that $\|u_\osc\|_{L^2} \leq C\tau^{-p/(1-p)}$, and so $u$ has the partial asymptotic expansion
    \begin{equation}
        u = c_1 + c_2\tau^{-(2p/(1-p)-1)} + \tau^{-p/(1-p)}v,
    \end{equation}
    where $v$ is an oscillatory function that does not have a limit as $\tau \ra \infty$. This means that to leading order, $u$ smooths out and becomes spatially homogeneous as $\tau \ra \infty$, but this homogenisation does not occur at higher order.
\end{remark}
\subsection{Breakdown of the proof for \texorpdfstring{$1/3 < p \leq 2/3$}{p between one third and two thirds}}\label{sec:intro.subsec:breakdown}
It is interesting to note that the numerical simulations in \cite{Mar:instabilityFLRW:25} indicate instability of the FLRW solutions of the Einstein-Euler system with expansion exponent $1/3 \leq p \leq 2/3$. We recall that the ENSF system admits FLRW solutions with expansion exponent $p > 1/3$. While we do not prove instability for the ENSF system in the range $1/3 < p \leq 2/3$ of expansion exponents, the proof of Theorem~\ref{sec:intro.thm:main} does break down in a fundamental way for this range. In brief, the restricted range is due to insufficient decay from the principle part of the gauge-fixed system, which leads specifically to insufficient decay of $h_{\osc}^{00}$. In this section we discuss how this breakdown of the stability proof occurs.
\begin{remark}[Issues with gauge]\label{sec:intro.rmk:proofbreakdown}
     Given that the lapse and shift equations are intrinsically tied to the choice of wave coordinates (more so than the equations for the horizontal components and scalar field), there is the possibility that this is an artefact of the gauge we have set. We do not know whether other choices of gauge source functions, for example, could resolve this issue.
\end{remark}

A model problem which captures the obstruction in the gauge-fixed system is the system of semilinear wave equations
\begin{subequations}\label{sec:intro.eq:bdownmodel}
    \begin{align}
        \square u_1 &= \Big[\frac{2p}{1-p}-2\Big]\frac{1}{\tau}\pa_\tau u_1 + \mc{N}(\pa u_1,\pa u_1),\label{sec:intro.eq:bdownmodel1}\\
        \square u_2 &= \frac{2p}{1-p}\frac{1}{\tau}\pa_\tau u_2 + \mc{N}(\pa u_2,\pa u_2),\label{sec:intro.eq:bdownmodel2}
    \end{align}
\end{subequations}
where $\mc{N}$ is arbitrary quadratic form.\footnote{A more realistic model system would include mixed terms in the nonlinearity, so that $\mc{N} = \mc{N}(\pa u_1,\pa u_2)$. With this system one would see that not only does the equation \eqref{sec:intro.eq:bdownmodel1} for $u_1$ form the primary obstruction to stability if $p \leq 2/3$, it also ``bottlenecks'' the decay rate of $u_2$.} The equation \eqref{sec:intro.eq:bdownmodel1} for $u_1$ models the equation for the remainder quantities of the lapse perturbation $h_{\osc}^{00}$, while the equation \eqref{sec:intro.eq:bdownmodel2} for $u_2$ models the equations for $h_{\osc}^{0i}$, $h_{\osc}^{ij}$, $\psi_{\osc}$. Since the equations \eqref{sec:intro.eq:bdownmodel} model remainder functions, we assume that $u_1$ and $u_2$ have zero spatial average. The absence of terms of the form $\tau^{-2}u$, $\tau^{-2}v$ in either equations are because these are essentially error terms which can be controlled via the Poincar\'e inequality \eqref{sec:intro.sec:avgsys.eq:poincare}, hence we ignore them in this discussion. The spatial averages $(h_{\av}^{\mu\nu},\psi_{\av})$ do not play a role in this breakdown of stability, as the ODE governing their evolution are somewhat better behaved, cf. Remark~\ref{sec:intro.rmk:bottleneck}. 

Decay of the functions $(u_1,u_2)$ stem from the dissipative terms $\tau^{-1}\pa_\tau u_1$, $\tau^{-1}\pa_\tau u_2$ in the equations \eqref{sec:intro.eq:bdownmodel} respectively. We introduce the rescaled quantities
\begin{equation}
    v_1 = \tau^{p/(1-p)-1}u_1,\qquad v_2 = \tau^{p/(1-p)}u_2.
\end{equation}
A brief computation shows that $(v,V)$ satisfy the system
\begin{align}
    \square v_1 &= \frac{c_1}{\tau^2}v_1 + \tau^{-(p/(1-p)-1)}\mc{N}(\pa v_1,\pa v_1),\\
    \square v_2 &= \frac{c_2}{\tau^2}v_2 + \tau^{-p/(1-p)}\mc{N}(\pa v_2,\pa v_2),
\end{align}
for constants $c_1,c_2$. Again, the functions $(v_1,v_2)$ have zero spatial average, so we can safely ignore the $\tau^{-2}v_1$, $\tau^{-2}v_2$ terms in the above equations (cf. \eqref{sec:intro.eq:oscexplain1}, \eqref{sec:intro.eq:oscexplain2}). We can only close the global existence argument if the semilinear terms have an integrable weight in time. This is true of the equation for $v_2$ if
\begin{equation}
    \frac{p}{1-p} > 1\iff p\in (1/2,1),
\end{equation}
while it is true of the equation for $v_1$ if
\begin{equation}
    \frac{p}{1-p}-1 > 1 \iff p \in (2/3,1).
\end{equation}
Hence global existence for this model system holds only if $p \in (2/3,1)$.

Relating this back to the gauge-fixed ENSF system, the primary obstruction for global existence beyond this range is the equation \eqref{sec:gaugedeqs.eq:wavelapse} for the lapse. On the one hand, this could be an issue of gauge; recall the discussion in Remark~\ref{sec:intro.rmk:proofbreakdown}. On the other, there is some numerical evidence that the range $p \leq 2/3$ of expansion exponents is unstable. As previously mentioned in Section~\ref{sec:intro.subsec:related}, the work \cite{Mar:instabilityFLRW:25} considers the stability of decelerated FLRW solutions to the Einstein-Euler equations. Marshall shows via numerical simulation that the FLRW solution is unstable in the full range of sound speed $c_s \in [0,1]$, which corresponds to expansion exponents $p \in [1/3,2/3]$. It is not well-understood whether this instability is due to the specific matter model (a perfect fluid with linear equation of state), or whether such instability is a product of the Einstein equations irrespective of matter.\footnote{One way to ``test'' such a question is by considering the decoupled problem, that is, the relativistic Euler equations on a fixed FLRW background. This has been investigated in the works \cite{Spe:relEulerstab:13,Fajetal:decelEulerstab:25}. In \cite{Spe:relEulerstab:13}, Speck shows that the dust model ($c_s = 0$) is stable on a fixed FLRW background for all $p > 1/2$. Note that the FLRW solution to Einstein-dust system has expansion exponent $p = 2/3$. In the recent \cite{Fajetal:decelEulerstab:25}, Fajman et. al prove stability of the relativistic Euler equations on a region of the parameter space $(c_s,p)$, which for all $c_s > 0$ excludes the parameter values $p = \frac{2}{3}(1+c_s^2)^{-1}$ attained by FLRW solutions to Einstein-Euler. Moreover, it is shown that radiative fluids $(c_s = 1/\sqrt{3}$) are unstable for all expansion exponents $p \leq 1$.}
\subsection{Acknowledgments}
I would like to thank Volker Schlue for his feedback during the writing of this paper, as well as Elliot Marshall for several helpful discussions. I acknowledge financial support provided by the University of Melbourne.\newpage
\section{Preliminaries}\label{prelims}
This section contains various preliminaries for the global stability theorem. First we introduce notation that we use throughout this paper. Then we state several important Sobolev space embeddings and inequalities. Finally we show that the FLRW spacetime $(\wt{g}_b,\phi_b)$ is indeed a solution to the Einstein-nonlinear scalar field system.
\subsection{Notation}\label{sec:prelims.subsec:notation}
\paragraph*{Index notation.}
We employ abstract index notation in this paper. Greek letters $\mu,\nu,\alpha,\beta,\dots $ represent spacetime coordinates that range over values $\{0,1,2,3\}$, and lowercase Latin letters $a,b,,\dots$ represent spatial coordinates that range over $\{1,2,3\}$. Pairs of repeated indices in a product are summed. Unless stated otherwise, we raise and lower indices with respect to the conformal metric $g$ introduced in Section~\ref{sec:gaugedeqs}.

\paragraph*{Coordinates.}
We will work in a standard local coordinate system $(x^1,x^2,x^3)$ on $\TT^3$. We remark that this coordinate system is not globally well-defined on all of $\TT^3$; however the coordinate vector-fields
\begin{equation}
    \pa_i = \frac{\pa}{\pa x^i},\qquad i=1,2,3
\end{equation}
are globally well-defined.\footnote{One could equivalently take $(x^1,x^2,x^3)$ to be periodic coordinates on $\TT^3$.} This coordinate system extends to a coordinate system $(x^0,x^1,x^2,x^3)$ on an $(n+1)$-dimensional manifold with boundary $\mc{M} = [\tau_0,T) \times \TT^3$. For the time coordinate, we often write $x^0 = \tau$.

\paragraph*{Differential operators.}
We denote the coordinate vector-fields
\begin{equation}
    \pa_\mu = \frac{\pa}{\pa x^\mu},
\end{equation}
and we often write $\pa_0 = \pa_\tau$ for the time derivative. We perform all computations with respect to the frame $\{\pa_\mu\}_{\mu=01,2,3}$.

Given a multi-index $I = (I_1,I_2,I_3) \in \NN^3$, the operator $\ol{\pa}{}^I$ denotes the spatial derivative
\begin{equation}
    \ol{\pa}{}^I = \pa_1^{I_1}\cdot \pa_2^{I_2}\cdot \pa_3^{I_3}.
\end{equation}
For a spacetime function $u$, we will use the schematic notation $\pa u$ to denote any/all spacetime derivatives of $u$, and $\ol{\pa} u$ to denote any/all purely spatial derivatives of $u$.
\paragraph{Norms and spatial averages.}
Let $u$ be a function on $\TT^3$. We often write the integral of $u$ over $\TT^3$ as
\begin{equation}
    \int_{\TT^3} u := \int_{\TT^3}u(x)\,\de^3 x,
\end{equation}
omitting the volume form $\de^3 x = \de x^1 \wedge \de x^2 \wedge \de x^3$.
We define the standard inhomogeneous Sobolev norm of order $K$ on $\TT^3$ by
\begin{equation}
    \|u\|_{H^K}^2 := \sum_{|I|\leq K}\int_{\TT^3}(\ol{\pa}{}^I u)^2\de^3 x,
\end{equation}
and write $L^2(\TT^3) := H^0(\TT^3)$. Similarly, we define the $W^{K,\infty}(\TT^3)$ norms:
\begin{equation}
    \|u\|_{W^{K,\infty}} := \sum_{|I|\leq K} \mathrm{ess}\,\sup_{\TT^3}|\ol{\pa}{}^Iu|,
\end{equation}
and write $L^\infty(\TT^3):= W^{0,\infty}(\TT^3)$. We denote the spatial average of $u$ over $\TT^3$ by
\begin{equation}
    u_{\av} := \frac{1}{(2\pi)^3}\int_{\TT^3}u(x)\de^3 x.
\end{equation}
We denote the zero-average remainders by $u_{\osc} = u-u_{\av}$.
\paragraph*{Running constants.}
We use a running constant in this paper, which we denote by $C$. This constant can vary from line to line, but is generally used in the context where $C$ must be chosen to be sufficiently large for the inequality to hold. The constant $C$ can depend on the regularity $K$ and smallness parameter $\ve_0$ that determine the size of the initial data, but are independent of $(g^{\mu\nu},\phi)$, assuming of course that these quantities are sufficiently close to the background quantities $(m^{\mu\nu},\phi_b)$. Other constants may be used when they have a specific role in the analysis; we will denote these constants by $c_1$, $c_2$, etc.
\subsection{Sobolev space embeddings}\label{sec:prelims.subsec:sobolev}
We give several Sobolev space based embeddings that are used throughout this paper. Aside from Lemma~\ref{sec:prelims.lem:sobolevcomp}, these are all standard estimates, and we point the interested reader to Chapter 2 of \cite{Heb:sob:00}.
\begin{lemma}[Sobolev embedding, \protect{\cite[Theorem~2.7]{Heb:sob:00}}]\label{sec:prelims.lem:sobolev}
Let $K \in \NN$. There exists a constant $C$ such that if $u \in H^K(\TT^3)$, then $u \in W^{K-2,\infty}(\TT^3)$ and
   \begin{equation}\label{sec:prelims.eq:sobolev}
       \|u\|_{W^{K-2,\infty}} \leq  C\|u\|_{H^K}.
   \end{equation} 
\end{lemma}

An immediate consequence of the above Sobolev embedding and the Leibniz product rule is the following Sobolev space estimate for products.
\begin{lemma}[Sobolev product estimate]\label{sec:prelims.lem:sobolevproduct}
    Let $K \in \NN$, $K \geq 3$. Then there exists a constant $C > 0$ such that if $u,v\in H^K(\TT^3)$, the product $u\cdot v \in H^K(\TT^3)$, and
    \begin{equation}\label{sec:prelims.bootbound.eq:sobolevproduct}
        \|u\cdot v\|_{H^K} \leq C\|u\|_{H^K} \|v\|_{H^K}.
    \end{equation}
\end{lemma}

The next Lemma allows us to estimate compositions of functions $F\circ u$, where $F$ is a smooth function and $u \in H^K(\TT^3)$. Importantly, if $F(x) \leq C|x|^l$ for small $x$, then $\|F\circ u\|_{H^K(\TT^3)} \leq C\|u\|_{H^K(\TT^3)}^l$ for $u$ with small $H^K(\TT^3)$ norm. We delay its proof to Appendix~\ref{app:sobolevcomp}.
\begin{lemma}[Sobolev composition estimate]\label{sec:prelims.lem:sobolevcomp}
    Let $K,l \in \NN$, $K \geq 3$. Let $S \subset \RR^N$ be a compact set with $0 \in S$, and suppose $F: S \ra \RR$ is a smooth function such that 
    \begin{equation}
        \ol{\pa}{}^I F(0) = 0
    \end{equation}
    for all $|I| \leq l-1$, where $\pa^I = \pa_{x_1}^{I_1}\cdots \pa_{x_N}^{I_N}$. Then there exists a constant $C > 0$ such that if $u \in H^K(\TT^3)$ and 
    $u(x) \subset S$ for almost all $x \in \TT^3$, then $F \circ u \in H^K(\TT^3)$, moreover
    \begin{equation}\label{sec:prelims.eq:sobolevcomp}
        \|F \circ u\|_{H^K(\TT^3)} \leq C\|u\|_{H^K(\TT^3)}^l
    \end{equation}
\end{lemma}

We end this section with the Poincar\'e inequality on the torus.
\begin{lemma}[Poincar\'e inequality, \protect{\cite[Theorem~2.11]{Heb:sob:00}}]\label{sec:prelims.lem:poincare}
    There exists a constant $C > 0$ such that for all $u \in H^1(\TT^3)$, we have the following inequality: 
    \begin{equation}\label{sec:avgsys.eq:poincare}
        \|u_{\osc}\|_{L^2} \leq C\|\ol{\pa} u\|_{L^2}.
    \end{equation}
\end{lemma}
\subsection{The background FLRW solution}\label{sec:prelims.subsec:FLRW}
The Einstein-nonlinear scalar field system \eqref{sec:intro.eq:ensf} can equivalently be written in trace-reversed form as
\begin{subequations}
    \begin{align}
    \Ric_{\mu\nu}[\wt{g}] &= \pa_\mu \phi \pa_\nu \phi + \wt{g}_{\mu\nu}V(\phi),\label{sec:prelims.eq:einstein}\\
    \square_g \phi - V'(\phi)&=0.\label{sec:prelims.eq:scalarfield}
\end{align}
\end{subequations}
We recall from \eqref{sec:intro.eq:backgroundmetric}, \eqref{sec:intro.eq:backgroundscalarfield} the background solution $((0,\infty)\times\TT^3,\wt{g}_b,\phi_b)$, given by
\begin{equation}
    \wt{g}_b = -\de t^2 + t^{2p} \delta_{ij}\,\de x^i \de x^j,\qquad 
        \phi_b = (2p)^{1/2} \log t + (p/2)^{1/2}\log\Big[\frac{V_0}{p(3p-1)}\Big].
\end{equation}
In the next lemma, we prove that the above spacetime is in fact a solution to the ENSF system. This was also proved in Lemma 1 of \cite{Rin:powerlaw:09}, albeit for expansion exponents restricted to $p > 1$.
\begin{lemma}[Background solution to the Einstein-nonlinear scalar field system]\label{sec:prelims.lem:background}
    The triple $((0,\infty)\times \TT^3,\wt{g}_b,\phi_b)$ is a solution to the Einstein nonlinear scalar field system \eqref{sec:prelims.eq:einstein}, \eqref{sec:prelims.eq:scalarfield}.
\end{lemma}
\begin{proof}
    The nonzero Christoffel symbols of the metric $\wt{g}_b$ are
    \begin{equation}
        \Gamma_{ij}^t(\wt{g}_b) = p t^{2p-1} \delta_{ij},\qquad \Gamma_{ti}^j(\wt{g}_b) = \frac{p}{t}\delta_i^j,\quad i,j=1,2,3.
    \end{equation}
    The components of the Ricci tensor for an arbitrary metric $\wt{g}$ satisfies the identity
    \begin{equation}
        \Ric_{\mu\nu}[\wt{g}]= \pa_\lambda \Gamma_{\mu\nu}^\lambda(\wt{g}) - \pa_{(\mu}\Gamma_{\nu)\lambda}^\lambda(\wt{g}) + \Gamma_{\mu\nu}^\lambda(\wt{g})\Gamma_{\lambda\delta}^\delta(\wt{g}) - \Gamma_{\mu\lambda}^\delta(\wt{g})\Gamma_{\nu\delta}^\lambda(\wt{g}).
    \end{equation}
    The off-diagonal components of $\Ric[\wt{g}_b]$ are identically zero, and the diagonal components are
    \begin{subequations}
        \begin{align}
            R_{tt} &= -\sum_{i=1}^3\pa_t \Gamma_{ta}^a(\wt{g}_b) - \sum_{i=1}^3 \big(\Gamma_{ta}^a(\wt{g}_b)\big)^2 = \frac{3p(1-p)}{t^2},\\
            R_{ij} &= \pa_t \Gamma_{ij}^t(\wt{g}_b) + \Gamma_{ij}^t(\wt{g}_b)\sum_{j=1}^3\Gamma_{ta}^a(\wt{g}_b) - 2\sum_{a=1}^3\Gamma_{ia}^t(\wt{g}_b)\Gamma_{tj}^a(\wt{g}_b) = \frac{p(3p-1)}{t^{2-2p}}\delta_{ij},\quad i,j=1,2,3.
        \end{align}
    \end{subequations}
    For the background scalar field, we compute
    \begin{equation}
        V(\phi_b) = \frac{p(3p-1)}{t^2}, \qquad \pa_\mu \phi_b\pa_\nu\phi_b = \frac{2p}{t^2}\delta_{\mu}^0\delta_{\nu}^0.
    \end{equation}
    The trace-reversed stress-energy-momentum tensor has the formula $T_{\mu\nu} - \frac{1}{2}T\wt{g}_{\mu\nu} = \pa_\mu\phi\pa_\nu\phi + V(\phi)\wt{g}_{\mu\nu}$. Clearly the off-diagonal components are identically zero, and the diagonal components are
    \begin{equation}
        T_{tt} - \frac{1}{2}T (\wt{g}_b)_{tt} = \frac{3p(1-p)}{t^2},\qquad T_{ij} - \frac{1}{2}T(\wt{g}_b)_{ij} = \frac{p(3p-1)}{t^{2-2p}}\delta_{ij}, \quad i,j=1,2,3,
    \end{equation}
    matching those of $\Ric_{\mu\nu}[\wt{g}]$. Hence the Einstein equations \eqref{sec:prelims.eq:einstein} are satisfied. For the scalar field equations \eqref{sec:prelims.eq:scalarfield} we have
    \begin{align}
        \square_{\wt{g}_b}\phi_b - V'(\phi_b) &= -\pa_t^2 \phi_b - \sum_{a,b=1}^3g^{ab}\Gamma_{ab}^t(\wt{g}_b)\pa_t\phi_b+(2/p)^{1/2}V(\phi_b)\nonumber\\
        &=\frac{(2p)^{1/2}}{t^2} - \frac{3p(2p)^{1/2}}{t^2} + \frac{(2p)^{1/2}(3p-1)}{t^2} = 0.
    \end{align}
    This completes the proof.
\end{proof}

\paragraph{The FLRW solution in conformal coordinates.} We recall that the FLRW metric $\wt{g}_b$ is conformally flat. We change time coordinate
\begin{equation}
    \tau = \frac{1}{1-p}\int \frac{\de t}{t^p} = t^{1-p}.
\end{equation}
In the coordinates $(\tau,x)$, the metric and scalar field for the background solution take the form 
\begin{equation}\label{sec:prelims.eq:backgroundconf}
\wt{g}_b = \tau^{2p/(1-p)}\big(-(1-p)^{-2}\de \tau^2 + \delta_{ij}\,\de x^i\de x^j\big),\qquad \phi_b = \frac{(2p)^{1/2}}{1-p}\log \tau + (p/2)^{1/2}\log\Big|\frac{V_0}{p(3p-1)}\Big|.
\end{equation}
So that the conformal metric $m := \tau^{-2p/(1-p)}\wt{g}_b$ is the flat Minkowski metric. We will henceforth refer to the pair $(m,\phi_b)$ as the background solution. We note that the future boundary $\Sigma^+ = \{t = \infty\}$ is the level set $\{\tau = \infty\}$, while the past singularity $\Sigma^- = \{t = 0\}$ is the level set $\{\tau = 0\}$.
\section{The gauge-fixed Einstein-nonlinear scalar field system}\label{sec:gaugedeqs}
This section is dedicated to deriving the gauge-fixed Einstein-nonlinear scalar field system \eqref{sec:gaugedeqs.eq:wave} which is central to our global existence argument. First we restate the Einstein-nonlinear scalar field system as a system of equations on a conformally transformed metric. Then we reformulate these equations as a system for the metric and scalar field perturbations, with the background FLRW spacetime constituting the ``zero solution''. We introduce the particular choice of wave coordinates that we make for this paper, and and derive the resulting gauge-fixed system.
\subsection{The Einstein-nonlinear scalar field system on a conformal metric}
The gauge-fixed system we derive amounts to a system for a conformal metric $g = \Omega^2 \wt{g}$ (and the the original scalar field $\phi$), where $(\wt{g},\phi)$ satisfy the Einstein-nonlinear scalar field system. Due to the particularly simple nature of our choice of conformal factor (see \eqref{sec:gaugedeqs.eq:conformalfactor}), this is completely equivalent to the ENSF system, and also rescales the leading order asymptotics of the metric components, simplifying the analysis. We will solve for the inverse metric components $g^{\mu\nu}$: $\mu,\nu = 0,1,2,3$, rather than the (lowered) metric components $g_{\mu\nu}$.

We give two preliminary lemmas which amount to Ricci curvature identities for a Lorentzian metric. These are fairly standard identities, and so we will delay their proof to Appendix~\ref{app:ricciids}. The first lemma is an identity satisfied by the Ricci tensors of two metrics which are related by a conformal transformation.
\begin{lemma}[Ricci curvatures for conformal metrics]\label{sec:gaugedeqs.lem:ricciconformal}
    Let $\wt{g}$ be a Lorentzian metric, and introduce the conformally transformed metric $g = \Omega^2\wt{g}$, where $\Omega\in C^\infty(\mc{M}\ra \RR)$ is a scalar function that is smooth and positive everywhere. Then the Ricci curvature of the metrics $g$ and $\wt{g}$ are related by the identity
    \begin{align}
        \Ric^{\mu\nu}[g] =\>&g^{\mu\alpha}g^{\nu\beta}\Ric_{\alpha\beta}[\wt{g}] + g^{\delta\lambda}\pa_\lambda g^{\mu\nu} \pa_\delta (\log \Omega) - 2g^{\lambda(\mu}\pa_\lambda g^{\nu)\delta}\pa_\delta (\log \Omega) \nonumber\\
        &- g^{\mu\nu} \Big[g^{\alpha\beta}\pa_\alpha \pa_\beta (\log\Omega) - 2g^{\alpha\beta}\pa_\alpha (\log \Omega) \pa_\beta (\log \Omega) - \Gamma^\lambda \pa_\lambda (\log \Omega)\Big]\nonumber\\
        &-2g^{\mu\alpha}g^{\nu\beta}\Big[\pa_\alpha\pa_\beta(\log \Omega) + \pa_\alpha (\log \Omega)\pa_\beta (\log \Omega)\Big],\label{sec:gaugedeqs.eq:ricciconformal}
    \end{align}
    where $\Gamma^\lambda:= g^{\alpha\beta}\Gamma_{\alpha\beta}^\lambda(g)$.
\end{lemma}

The second lemma is a decomposition of the Ricci curvature into a scalar wave operator, non-hyperbolic principal terms which vanish upon fixing wave coordinates, and semilinear terms.
\begin{lemma}[Ricci curvature wave decomposition]\label{sec:gaugedeqs.lem:ricciwave}
    Let $g$ be a Lorentzian metric, and let $\wh{\square}_g$ denote the reduced wave operator 
    \begin{equation}\label{sec:gaugedeqs.eq:reducedwaveoperator}
    \wh{\square}_g := g^{\alpha\beta}\pa_\alpha\pa_\beta.
    \end{equation}
    Then the Ricci curvature of $g$ satisfies the identity
    \begin{equation}
        2\Ric^{\mu\nu}[g] = \wh{\square}_g g^{\mu\nu} + 2g^{\lambda(\mu}\pa_\lambda \Gamma^{\nu)} - \mc{F}^{\mu\nu}(g,\pa g),
    \end{equation}
    where the $\mc{F}^{\mu\nu}$ are the functions
    \begin{multline}\label{sec:gaugedeqs.eq:semilinearterms}
        \mc{F}^{\mu\nu}(g,\pa g) = g^{\alpha\beta}g^{\lambda\delta}\pa_\lambda g^{\mu\nu}\pa_{(\alpha} g_{\beta)\delta} - \frac{1}{2}g^{\alpha\beta}g^{\lambda\delta}\pa_\lambda g^{\mu\nu}\pa_\delta g_{\alpha\beta} + 2g^{\lambda(\mu}g^{\nu)\delta}\pa_\delta g^{\alpha\beta}\pa_\alpha g_{\beta\lambda}\\- \frac{1}{2}g^{\lambda(\mu}g^{\nu)\delta}\pa_\lambda g^{\alpha\beta}\pa_\delta g_{\alpha\beta}
        -g^{\alpha\beta}g^{\lambda(\mu|}\pa_\alpha g^{|\nu)\delta}\pa_\delta g_{\beta\lambda} - g^{\alpha\beta}g^{\lambda(\mu|}\pa_\alpha g^{|\nu)\delta}\pa_\beta g_{\lambda\delta}.
    \end{multline}
\end{lemma}

We will now fix the conformal factor to be
\begin{equation}\label{sec:gaugedeqs.eq:conformalfactor}
    \Omega = \tau^{-\frac{p}{1-p}},
\end{equation}
which is precisely the conformal factor which transforms the background FLRW metric $\wt{g}_b$ into the Minkowski metric $m$ (cf. the discussion of regarding the conformal metric in Section~\ref{sec:intro.subsec:proof}). This means that instead of solving for a metric $\wt{g}$ which is initially a small perturbation of the FLRW solution $\wt{g}_b$, we will instead derive a system of equations for the (inverse) conformal metric components $g^{\mu\nu}$, which are initially small perturbations of Minkowski.

\smallskip
We combine the previous two lemmas to derive equations for the components of the inverse conformal metric components $g^{\mu\nu}$ and scalar field $\phi$. As we have not yet fixed a gauge, the equations for the $g^{\mu\nu}$ are not wave equations, but the structure of the non-hyperbolic principle terms are made clear.
\begin{lemma}\label{sec:gaugedeqs.prop:confwave}
    Suppose $(\wt{g},\phi)$ satisfy the Einstein-nonlinear scalar field system \eqref{sec:prelims.eq:einstein}, \eqref{sec:prelims.eq:scalarfield}. Let $(\tau = x^0,x^1,x^2,x^3)$ be coordinates on $\mc{M}$, and let $g = \Omega^2\wt{g}$ be the conformally related metric with conformal factor given by \eqref{sec:gaugedeqs.eq:conformalfactor}. Then the quantities $\wh{\square}_g g^{\mu\nu}$, $\wh{\square}_g \phi$ satisfy the identities
    \begin{subequations}
        \begin{align}
            \wh{\square}_g g^{\mu\nu} =\>& -\frac{2p}{1-p}\frac{1}{\tau}g^{0\lambda}\pa_\lambda g^{\mu\nu} + \frac{4p}{1-p}\frac{1}{\tau}g^{\lambda(\mu}\pa_\lambda g^{\nu)0} - 2g^{\lambda(\mu}\pa_\lambda \Gamma^{\nu)}(g)\nonumber\\
            &-\frac{2p}{1-p}\frac{1}{\tau}\Gamma^0(g)g^{\mu\nu} +  \frac{2p(3p-1)}{(1-p)^2}\frac{1}{\tau^2}g^{00}g^{\mu\nu} - \frac{4p}{(1-p)^2}\frac{1}{\tau^2}g^{0\mu}g^{0\nu}\nonumber\\
            &+ \mc{F}^{\mu\nu}(g,\pa g) + 2g^{\mu\alpha}g^{\mu\beta}\pa_\alpha\phi\pa_\beta\phi + 2\tau^{2p/(1-p)}g^{\mu\nu}V(\phi),\label{sec:gaugedeqs.eq:confmet1}\\
            \wh{\square}_g\phi =&-\frac{2p}{1-p}\frac{1}{\tau}g^{0\lambda}\pa_\lambda \phi + \Gamma^\lambda \pa_\lambda\phi+ \tau^{2p/(1-p)}V'(\phi).\label{sec:gaugedeqs.eq:confsf1}
        \end{align}
    \end{subequations}
\end{lemma}
\begin{proof}
    The equation \eqref{sec:gaugedeqs.eq:confsf1} follows from \eqref{sec:prelims.eq:scalarfield} with the computation
    \begin{align}
        \wh{\square}_g \phi &= \tau^{2p/(1-p)}\wt{g}^{\alpha\beta}\pa_\alpha\pa_\beta \phi\nonumber\\
        &= \tau^{2p/(1-p)}\big[\square_{\wt{g}} \phi +\wt{g}^{\alpha\beta}\Gamma_{\alpha\beta}^\lambda(\wt{g}) \pa_\lambda \phi\big].\label{sec:gaugedeqs.eq:confwaveproof1}
    \end{align}
    We expand the contracted Christoffel symbol $\wt{g}^{\alpha\beta}\Gamma_{\alpha\beta}^\lambda(\wt{g})$:
    \begin{align}
        \wt{g}^{\alpha\beta}\Gamma_{\alpha\beta}^{\lambda}(\wt{g}) &= \frac{1}{2}\wt{g}^{\alpha\beta}\wt{g}^{\lambda\delta}(2\pa_{(\alpha}\wt{g}_{\beta)\delta} - \pa_\delta \wt{g}_{\alpha\beta})\nonumber\\
        &= \frac{1}{2}\tau^{-4p/(1-p)}g^{\alpha\beta}g^{\lambda\delta}\big[2\pa_{(\alpha|}(\tau^{2p/(1-p)}g_{|\beta)\delta}) - \pa_\delta (\tau^{2p/(1-p)}g_{\alpha\beta})\big]\nonumber\\
        &=\tau^{-2p/(1-p)}\Big[-\frac{2p}{1-p}\frac{1}{\tau}g^{0\lambda} +\Gamma^\lambda(g)\Big],\label{sec:gaugedeqs.eq:christrel}
    \end{align}
    combining this with \eqref{sec:gaugedeqs.eq:confwaveproof1} yields \eqref{sec:gaugedeqs.eq:confsf1}.
    
    Next we derive \eqref{sec:gaugedeqs.eq:confmet1}. We compute for the terms on the right hand side of \eqref{sec:gaugedeqs.eq:ricciconformal}:
    \begin{subequations}
        \begin{gather}
            g^{\delta\lambda}\pa_\lambda g^{\mu\nu} \pa_\delta (\log \Omega) - 2g^{\lambda(\mu}\pa_\lambda g^{\nu)\delta}\pa_\delta (\log \Omega) = -\frac{p}{1-p}\frac{1}{\tau}g^{0\lambda}\pa_\lambda g^{\mu\nu} + \frac{2p}{1-p}\frac{1}{\tau}g^{\lambda(\mu}\pa_\lambda g^{\nu)0},\\
            - g^{\mu\nu} \Big[g^{\alpha\beta}\pa_\alpha \pa_\beta (\log\Omega) - 2g^{\alpha\beta}\pa_\alpha (\log \Omega) \pa_\beta (\log \Omega) - \Gamma^\lambda \pa_\lambda (\log \Omega)\Big] = g^{\mu\nu}\Big[\frac{p(3p-1)}{(1-p)^2}\frac{1}{\tau^2}g^{00} - \frac{p}{1-p}\frac{1}{\tau}\Gamma^0 \Big],\\
            -2g^{\mu\alpha}g^{\nu\beta}\Big[\pa_\alpha\pa_\beta(\log \Omega) + \pa_\alpha (\log \Omega)\pa_\beta (\log \Omega)\Big] = -\frac{2p}{(1-p)^2}\frac{1}{\tau^2}g^{0\mu}g^{0\nu}.
        \end{gather}
    \end{subequations}
    Therefore by Lemma~\ref{sec:gaugedeqs.lem:ricciconformal} the Ricci curvature of $g$ obeys the identity
    \begin{multline}
        2\Ric^{\mu\nu}[g] = 2g^{\mu\alpha}g^{\nu\beta}\Ric_{\alpha\beta}[\wt{g}] -\frac{2p}{1-p}\frac{1}{\tau}g^{0\lambda}\pa_\lambda g^{\mu\nu} + \frac{4p}{1-p}\frac{1}{\tau}g^{\lambda(\mu}\pa_\lambda g^{\nu)0}\\-\frac{2p}{1-p}\frac{1}{\tau}\Gamma^0(g)g^{\mu\nu} +  \frac{2p(3p-1)}{(1-p)^2}\frac{1}{\tau^2}g^{00}g^{\mu\nu} - \frac{4p}{(1-p)^2}\frac{1}{\tau^2}g^{0\mu}g^{0\nu}.\label{sec:gaugedeqs.eq:confwaveproof2}
    \end{multline}
    The Einstein equations imply that the Ricci curvature of $\wt{g}$ satisfies
    \begin{align}
        g^{\mu\alpha}g^{\nu\beta}\Ric_{\alpha\beta}[\wt{g}] &= g^{\mu\alpha}g^{\nu\beta}\pa_\alpha\phi\pa_\beta\phi + g^{\mu\alpha}g^{\nu\beta}\wt{g}_{\alpha\beta}V(\phi)\nonumber\\
        &= g^{\mu\alpha}g^{\nu\beta}\pa_\alpha\phi\pa_\beta\phi + \tau^{2p/(1-p)}g^{\mu\nu}V(\phi).\label{sec:gaugedeqs.eq:confwaveproof3}
    \end{align}
    Combining \eqref{sec:gaugedeqs.eq:confwaveproof2} and \eqref{sec:gaugedeqs.eq:confwaveproof3} with the decomposition of the Ricci curvature from Lemma~\ref{sec:gaugedeqs.lem:ricciwave}, we obtain \eqref{sec:gaugedeqs.eq:confmet1}.
\end{proof}
\subsection{Equations for the metric and scalar field perturbations}
We wish to solve the system \eqref{sec:prelims.eq:einstein}, \eqref{sec:prelims.eq:scalarfield} for small perturbations of the background FLRW spacetime. With that in mind, we introduce the functions
\begin{equation}
     h^{\mu\nu} := g^{\mu\nu} - m^{\mu\nu},\qquad \psi := \phi - \phi_b,
\end{equation}
where $(m,\phi_b)$ are the conformal metric and scalar field of the background FLRW solution \eqref{sec:intro.eq:background}
Note that the background components $m^{\mu\nu}$ are constant, and so $\pa g^{\mu\nu} = \pa h^{\mu\nu}$. Moreover, the background shift components are identically zero, therefore $g^{0i} = h^{0i}$. We will be use these identities throughout the rest of the paper.

In the following Lemma we derive equations for the perturbations $h^{\mu\nu}$, $\psi$, which are equivalent to the Einstein-nonlinear scalar field system. We fix the scalar field potential to be
\begin{equation}
    V(\phi) = V_0 \exp(-(2/p)^{1/2}\phi).
\end{equation}
\begin{lemma}[Equations for metric and scalar field perturbations]\label{lem:waveequationsperturbations}
    The system of equations for the $(g^{\mu\nu},\phi)$ in Proposition~\ref{sec:gaugedeqs.prop:confwave} are equivalent to the following equations for the metric perturbations $(h^{\mu\nu},\psi)$:
    \begin{subequations}
        \begin{align}
            \wh{\square}_g  h^{00} =\,& \frac{1}{\tau}g^{0\lambda}\pa_\lambda\Big[\frac{2p}{1-p} h^{00} - 2\tau\Gamma^0\Big]\nonumber\\
            &+\frac{1}{\tau^2}\Big[-\frac{4p-2}{1-p}g^{00}\Gamma^0 + \frac{2p(3p-1)}{(1-p)^2}g^{00} h^{00}\Big] +\mc{F}^{00}+\mc{G}^{00},\label{sec:gaugedeqs.eq:wavelapseungauged}\\
            \wh{\square}_g  h^{0i} =\,& \frac{1}{\tau}g^{i\lambda}\pa_\lambda\Big[\frac{2p}{1-p} h^{00} - \tau \Gamma^{0}\Big] - \frac{1}{\tau}g^{0\lambda}\pa_\lambda \Big[\tau \Gamma^{i}\Big]\nonumber\\
            &+ \frac{1}{\tau^2}\Big[\tau g^{00}\Gamma^i-\frac{3p-1}{1-p}\tau g^{0i}\Gamma^{0}   + \frac{2p(3p-1)}{(1-p)^2}g^{0i} h^{00}\Big] +\mc{F}^{0i}+\mc{G}^{0i},\label{sec:gaugedeqs.eq:waveshiftungauged}\\
            \wh{\square}_g  h^{ij} =& -\frac{2p}{1-p}\frac{1}{\tau}g^{0\lambda}\pa_\lambda  h^{ij} + \boxed{\frac{1}{\tau}\Big[\frac{4p}{1-p}g^{\lambda(i}\pa_\lambda h^{j)0} - 2g^{\lambda(i}\pa_\lambda\big(\tau \Gamma^{j)}\big)\Big]}\nonumber\\
            &+\frac{2}{\tau} g^{0(i}\Gamma^{j)} + \boxed{\frac{1}{\tau^2}\Big[-\frac{2p}{1-p}\tau\Gamma^0g^{ij} + \frac{2p(3p-1)}{(1-p)^2} h^{00}g^{ij}\Big]} + \mc{F}^{ij}+\mc{G}^{ij},\label{sec:gaugedeqs.sec:gaugedeqs.eq:wavehoriungauged}\\
            \wh{\square}_g \psi =\,& -\frac{2p}{1-p}\frac{1}{\tau}g^{0\lambda}\pa_\lambda \psi - \frac{2(3p-1)}{(1-p)^2}\frac{1}{\tau^2}g^{00}\psi + \boxed{\frac{(2p)^{1/2}}{1-p}\frac{1}{\tau^2}\Big[\tau\Gamma^0 - \frac{3p-1}{1-p}h^{00}\Big]}\nonumber\\
            &+ \Gamma^\lambda \pa_\lambda \psi + \frac{2(3p-1)}{(1-p)^2}\frac{1}{\tau^2} h^{00}\psi - \frac{(2p)^{1/2}}{\tau^2}\mc{W}(\psi),\label{sec:gaugedeqs.eq:wavescalarungauged}
        \end{align}
    \end{subequations}
    where $\mc{W}(\psi)$ is the function
    \begin{equation}\label{sec:gaugedeqs.eq:potenterror}
        \mc{W}(\psi) :=  (3p-1)\big[\exp(-(2/p)^{1/2}\psi) -1 + (2/p)^{1/2}\psi\big].
    \end{equation}
    and the $\mc{G}^{\mu\nu}$ are the functions
    \begin{equation}\label{sec:gaugedeqs.eq:errorG}
        \mc{G}^{\mu\nu} := \frac{(32p)^{1/2}}{1-p}\frac{1}{\tau}g^{0(\mu}g^{\nu)\lambda}\pa_\lambda \psi + 2g^{\mu\alpha}g^{\nu\beta}\pa_\alpha\psi\pa_\beta\psi -\frac{(8p)^{1/2}(3p-1)}{\tau^2}\psi g^{\mu\nu} + \frac{2p}{\tau^2}\mc{W}(\psi)g^{\mu\nu}.
    \end{equation}
\end{lemma}
\begin{remark}
The boxed terms in the equations \eqref{sec:gaugedeqs.sec:gaugedeqs.eq:wavehoriungauged} and \eqref{sec:gaugedeqs.eq:wavescalarungauged} are terms which will, upon fixing wave coordinates, vanish identically. See Remark~\ref{sec:gaugedeqs.rmk:gaugemotivation} for more details.
\end{remark}
\begin{remark}
The function $\mc{W}(\psi)$ is the nonlinear error part of $V(\phi)$. We note that $\mc{W}(0) = \mc{W}'(0) = 0$, and so $\mc{W}(\psi) = O(\psi^2)$ for small $\psi$.
\end{remark}
\begin{proof}
    First we prove the equation \eqref{sec:gaugedeqs.eq:wavescalarungauged} for the scalar field perturbation. Using the formula \eqref{sec:prelims.eq:backgroundconf} for $\phi_b$, We compute
    \begin{align}
        \wh{\square}_g \psi &= \wh{\square}_g\phi - g^{00}\pa_0^2\phi_b\nonumber\\
        &=\wh{\square}_g\phi - \frac{(2p)^{1/2}(1-p)}{\tau^2} + \frac{(2p)^{1/2}}{1-p}\frac{1}{\tau^2} h^{00}.
    \end{align}
    We expand the terms on the right hand side of \eqref{sec:gaugedeqs.eq:confsf1}:
    \begin{subequations}
        \begin{align}
            -\frac{2p}{1-p}\frac{1}{\tau}g^{0\lambda}\pa_\lambda \phi =&\>\frac{(2p)^{3/2}}{\tau^2} - \frac{(2p)^{3/2}}{(1-p)^2}\frac{1}{\tau^2} h^{00} - \frac{2p}{1-p}\frac{1}{\tau}g^{0\lambda}\pa_\lambda\psi,\\
            \Gamma^\lambda \pa_\lambda \phi =&\> \frac{(2p)^{1/2}}{1-p}\frac{1}{\tau}\Gamma^0 + \Gamma^\lambda \pa_\lambda \psi.
        \end{align}
    \end{subequations}
    For the last term on the RHS of \eqref{sec:gaugedeqs.eq:confsf1}, we compute for the potential:
    \begin{equation}
        V(\phi_b) = p(3p-1)\tau^{-2/(1-p)},
    \end{equation}
    which implies
    \begin{equation}
        \tau^{2p/1-p}V(\phi) = \tau^{2p/1-p}V(\phi_b)\cdot\exp(-(2/p)^{1/2}\psi) = \frac{p(3p-1)}{\tau^2}\exp(-(2/p)^{1/2}\psi).
    \end{equation}
    Similarly, since $V'(\phi) = -(2/p)^{1/2}V(\phi)$, we have
    \begin{equation}
        \tau^{2p/1-p}V'(\phi) = -\frac{(2p)^{1/2}(3p-1)}{\tau^2}\exp(-(2/p)^{1/2}\psi).
    \end{equation}
    We Taylor expand to quadratic order in $\psi$:
    \begin{equation}
        \exp(-(2/p)^{1/2}\psi) = 1 - (2/p)^{1/2}\psi + O(\psi^2) = 1 + (2/p)^{1/2}\Big[\frac{1}{(1-p)^2}g^{00} - h^{00}\Big]\psi + O(\psi^2),
    \end{equation}
    and so we define the function $\mc{W}(\psi)$ like \eqref{sec:gaugedeqs.eq:potenterror}, which yields
    \begin{equation}
        \tau^{2p/(1-p)}V'(\phi) = - \frac{(2p)^{1/2}(3p-1)}{\tau^2}-\frac{2(3p-1)}{(1-p)^2}\frac{1}{\tau^2}g^{00}\psi +\frac{2(3p-1)}{(1-p)^2}\frac{1}{\tau^2} h^{00}\psi -\frac{(2p)^{1/2}}{\tau^2}\mc{W}(\psi).
    \end{equation}
    The terms that depend only on $\tau$ cancel. Combining this with the earlier computation yields \eqref{sec:gaugedeqs.eq:wavescalarungauged}.
    
    Next we derive the equations for the metric perturbations. We expand the terms on the right hand side of \eqref{sec:gaugedeqs.eq:confmet1} that depend on the scalar field:
    \begin{subequations}
        \begin{align}
            2g^{\mu\alpha}g^{\nu\beta}\pa_\alpha \phi \pa_\beta \phi =&\> \frac{4p}{(1-p)^2}\frac{1}{\tau^2}g^{0\mu}g^{0\nu} + \frac{(32p)^{1/2}}{1-p}\frac{1}{\tau}g^{0(\mu}g^{\nu)\lambda}\pa_\lambda \psi +2 g^{\mu\alpha}g^{\nu\alpha}\pa_\alpha \psi\pa_\beta \psi\label{sec:gaugedeqs.eq:metricpertproof1}\\
            2\tau^{2p/(1-p)}V(\phi)g^{\mu\nu}=& -\frac{2p(3p-1)}{(1-p)^2}\frac{1}{\tau^2}g^{00}g^{\mu\nu} + \frac{2p(3p-1)}{(1-p)^2}\frac{1}{\tau^2} h^{00}g^{\mu\nu}\nonumber\\
            &- \frac{(8p)^{1/2}(3p-1)}{\tau^2}\psi g^{\mu\nu} + \frac{2p}{\tau^2}\mc{W}(\psi)g^{\mu\nu}.\label{sec:gaugedeqs.eq:metricpertproof2}
        \end{align}
    \end{subequations}
    Observe that the first term on the right hand side of each of \eqref{sec:gaugedeqs.eq:metricpertproof1} and \eqref{sec:gaugedeqs.eq:metricpertproof2} cancel with the corresponding terms in \eqref{sec:gaugedeqs.eq:confmet1}, leaving us with
    \begin{multline}
        \wh{\square}_g  h^{\mu\nu} = -\frac{2p}{1-p}\frac{1}{\tau}g^{0\lambda}\pa_\lambda  h^{\mu\nu} + \frac{4p}{1-p}\frac{1}{\tau}g^{\lambda(\mu}\pa_\lambda  h^{\nu)0} - 2g^{\lambda(\mu}\pa_\lambda \Gamma^{\nu)}\\
        -\frac{2p}{1-p}\frac{1}{\tau}\Gamma^0g^{\mu\nu} + \frac{2p(3p-1)}{(1-p)^2}\frac{1}{\tau^2} h^{00}g^{\mu\nu}+ \mc{F}^{\mu\nu}+\mc{G}^{\mu\nu},
    \end{multline}
    where the $\mc{G}^{\mu\nu}$ are defined in \eqref{sec:gaugedeqs.eq:errorG}. We group terms appropriately to obtain
    \begin{multline}
        \wh{\square}_g  h^{\mu\nu} = -\frac{2p}{1-p}\frac{1}{\tau}g^{0\lambda}\pa_\lambda  h^{\mu\nu} + \frac{1}{\tau}\Big[\frac{4p}{1-p}g^{\lambda(\mu}\pa_\lambda  h^{\nu)0} - 2 g^{\lambda(\mu}\pa_\lambda(\tau \Gamma^{\nu)})\Big]\\
        +\frac{1}{\tau^2}\Big[2\tau g^{0(\mu}\Gamma^{\nu)}-\frac{2p}{1-p}\tau\Gamma^0g^{\mu\nu} + \frac{2p(3p-1)}{(1-p)^2} h^{00}g^{\mu\nu}\Big]+ \mc{F}^{\mu\nu}+\mc{G}^{\mu\nu}.
    \end{multline}
    The result now follows after considering the $00$-, $0i$-, and $ij$-components individually.
\end{proof}
\begin{remark}
We recall the cancellation of the terms
    \begin{equation}\label{sec:gaugedeqs.eq:cancellingterms}
        \frac{2p(3p-1)}{(1-p)^2}\frac{1}{\tau^2}g^{00}g^{\mu\nu} ,\qquad-\frac{4p}{(1-p)^2}\frac{1}{\tau^2}g^{0\mu}g^{0\nu},
    \end{equation}
    present in the system \eqref{sec:gaugedeqs.eq:confmet1} with the corresponding terms in \eqref{sec:gaugedeqs.eq:metricpertproof1} - \eqref{sec:gaugedeqs.eq:metricpertproof2}. While important for our argument, such a cancellation is expected because the background metric (the Minkowski metric) must also solve the system \eqref{sec:gaugedeqs.eq:confmet1}. The terms \eqref{sec:gaugedeqs.eq:cancellingterms} are the only ones present in \eqref{sec:gaugedeqs.eq:confmet1} that do not depend on any derivatives of the metric. Since all terms containing derivatives vanish for the Minkowski metric, these are the only potentially nonzero terms present in the system, and so their cancellation with the terms arising from $g^{\mu\alpha}g^{\nu\beta}\pa_\alpha \phi \pa_\beta \phi$ and $\tau^{2p/(1-p)}V(\phi) g^{\mu\nu}$ is necessary.
\end{remark}
\subsection{Wave coordinates}
We fix wave coordinates
\begin{equation}\label{sec:gaugedeqs.eq:waveconstraint}
    \Gamma^\mu(g) = \mc{Y}^{\mu}(g),\qquad \mu=0,1,2,3,
\end{equation}
where $\mc{Y}^\mu$ are the gauge source functions
\begin{equation}\label{sec:gaugedeqs.eq:gaugesourcechoice}
    \mc{Y}^0(g) = \frac{3p-1}{1-p}\frac{1}{\tau} h^{00},\qquad \mc{Y}^i(g) = \frac{2p}{1-p}\frac{1}{\tau} h^{0i}, \qquad i=1,2,3.
\end{equation}
\begin{remark}[Motivation for choice of gauge source functions]\label{sec:gaugedeqs.rmk:gaugemotivation}
    Our motivation for this choice of the functions $\mc{Y}^\mu$ is that they simplify the system more than any other choice. We highlight the boxed terms in the equations \eqref{sec:gaugedeqs.sec:gaugedeqs.eq:wavehoriungauged} for the horizontal metric perturbations and \eqref{sec:gaugedeqs.eq:wavescalarungauged} for the scalar field perturbations. These are all mixed linear terms, which complicate the structure of the equations. Fixing wave coordinates \eqref{sec:gaugedeqs.eq:waveconstraint} with source functions \eqref{sec:gaugedeqs.eq:gaugesourcechoice} ensures that these terms vanish completely. Despite this, there are still mixed linear terms present in the system, see the terms $\mc{L}^{\mu\nu}$, $\mc{L}^{(\psi)}$ in the gauge-fixed system \eqref{sec:gaugedeqs.eq:wave}.  
\end{remark}
\begin{remark}[Propagation of wave coordinate condition]
It is a standard property of the Einstein equations that wave coordinate conditions of the form \eqref{sec:gaugedeqs.eq:waveconstraint} are propagated by the reduced Einstein equations. One can prove that if \eqref{sec:gaugedeqs.eq:waveconstraint} are satisfied on some initial Cauchy surface, then they are satisfied for all future time such that the solution to the reduced system exists. See Chapter VII of \cite{Cho:grbook:15} for details.

In other works, wave coordinates are usually fixed with respect to the original metric $\wt{g}$ that solves the Einstein equations, rather than to a conformal metric. In the present setting these two formulations are equivalent due to the choice of conformal factor \eqref{sec:gaugedeqs.eq:conformalfactor}. Indeed, by \eqref{sec:gaugedeqs.eq:christrel}, the contracted Christoffel symbols of $g$ and $\wt{g}$ satisfy the relation 
    \begin{equation}
        \wt{g}^{\alpha\beta}\Gamma_{\alpha\beta}^\mu(\wt{g}) = \tau^{-2p/(1-p)}\Big(g^{\alpha\beta}\Gamma_{\alpha\beta}^\mu(g)-\frac{2p}{1-p}\frac{1}{\tau}g^{0\mu}\Big),
    \end{equation}
    and so fixing $\Gamma^\lambda(g) = \mc{Y}^\lambda(g)$ is equivalent to setting
    \begin{equation}
        \wt{g}^{\alpha\beta}\Gamma_{\alpha\beta}^\mu(\wt{g}) = \tau^{-2p/(1-p)}\mc{Y}^\mu\big(\tau^{2p/(1-p)}\wt{g}\big) - \frac{2p}{1-p}\frac{1}{\tau}\wt{g}^{0\mu}.
    \end{equation}
    Hence fixing wave coordinates for the conformal metric $g$ are equivalent to fixing (different) wave coordinates for $\wt{g}$.
\end{remark}
\subsection{The gauge-fixed system}
We now state the gauge-fixed system, which is equivalent to the ENSF system provided the wave coordinate condition \eqref{sec:gaugedeqs.eq:waveconstraint} holds. We introduce the parameter
\begin{equation}
    q := \frac{p}{1-p}-1,
\end{equation}
and observe that for $p \in (2/3,1)$, $q > 1$.
The gauge-fixed Einstein-nonlinear scalar field system is
\begin{subequations}\label{sec:gaugedeqs.eq:wave}
    \begin{align}
        &\wh{\square}_gh^{00} + \frac{2q}{\tau}g^{00}\pa_\tau h^{00} - \frac{4q+2}{\tau^2}g^{00}h^{00} = \mc{L}^{00} + \mc{N}^{00},\label{sec:gaugedeqs.eq:wavelapse}\\
        &\wh{\square}_g h^{0i} + \frac{2q+2}{\tau}g^{00}\pa_\tau h^{0i} - \frac{2q+2}{\tau^2}g^{00}h^{0i} = \mc{L}^{0i} + \mc{N}^{0i},\label{sec:gaugedeqs.eq:waveshift}\\
        &\wh{\square}_g h^{ij} + \frac{2q+2}{\tau} g^{00}\pa_\tau h^{ij} = \mc{L}^{ij} + \mc{N}^{ij},\label{sec:gaugedeqs.eq:wavehori}\\
        &\wh{\square}_g \psi + \frac{2q+2}{\tau} g^{00}\pa_\tau \psi + \frac{2(2q+1)(q+2)}{\tau^2}g^{00}\psi = \mc{N}^{(\psi)},\label{sec:gaugedeqs.eq:wavescalar}
    \end{align}
\end{subequations}
where the $\mc{L}^{00}$, $\mc{L}^{0i}$, $\mc{L}^{ij}$ are the linear terms
\begin{subequations}
    \begin{align}
        \mc{L}^{00} &:= \frac{(32p)^{1/2}}{1-p}\frac{1}{\tau} (g^{00})^2\pa_\tau \psi - \frac{(8p)^{1/2}(3p-1)}{\tau^2}g^{00}\psi,\\
        \mc{L}^{0i} &:= \frac{1}{\tau}g^{i a}\pa_a h^{00} + \frac{(8p)^{1/2}}{1-p}\frac{1}{\tau}g^{00}g^{ia}\pa_a \psi,\\
        \mc{L}^{ij} &:= -\frac{(8p)^{1/2}(3p-1)}{\tau}g^{ij}\psi,
    \end{align}
\end{subequations}
and the $\mc{N}^{00}$, $\mc{N}^{0i}$, $\mc{N}^{ij}$, $\mc{N}^{(\psi)}$ are the following error terms:
\begin{subequations}
    \begin{align}
        \mc{N}^{\mu\nu} &:= \mc{N}_0^{\mu\nu} + \mc{N}_1^{\mu\nu} + \mc{N}_2^{\mu\nu},\\
        \mc{N}^{(\psi)} &:=\mc{N}_0^{(\psi)} + \mc{N}_1^{(\psi)},
    \end{align}
\end{subequations}
\begin{subequations}
    \begin{align}
        \mc{N}_0^{00} :=&\>\frac{1}{\tau^2}\Big[2pg^{00}\mc{W}(\psi)\Big],\\
        \mc{N}_0^{0i} :=&\>\frac{1}{\tau^2}\Big[\frac{3p-1}{1-p}h^{0i}h^{00} - (8p)^{1/2}(3p-1)h^{0i}\psi + 2ph^{0i}\mc{W}(\psi)\Big],\\
        \mc{N}_0^{ij} := &\>\frac{1}{\tau^2}\Big[\frac{4p}{1-p}h^{0i}h^{0j} + 2pg^{ij}\mc{W}(\psi)\Big],&\\
        \mc{N}_0^{(\psi)} := &\>\frac{1}{\tau^2}\Big[\frac{2(3p-1)}{(1-p)^2}h^{00}\psi-(2p)^{1/2}\mc{W}(\psi)\Big],
        \end{align}
    \end{subequations}
\begin{subequations}
    \begin{align}
        \mc{N}_1^{00} :=&\> \frac{1}{\tau}\Big[-\frac{2(2p-1)}{1-p}h^{0a}\pa_a h^{00} + \frac{(32p)^{1/2}}{1-p}g^{00}h^{0a}\pa_a\psi\Big],\\
        \mc{N}_1^{0i} := &\frac{1}{\tau}\Big[-\frac{2p}{1-p}h^{0a}\pa_a h^{0i} + h^{0i}\pa_\tau h^{00} + \frac{(32p)^{1/2}}{1-p}g^{00}h^{0i}\pa_\tau \psi +\frac{(8p)^{1/2}}{1-p}h^{0a}h^{0i}\pa_a \psi\Big],\\
        \mc{N}_1^{ij} := &\>\frac{1}{\tau}\Big[-\frac{2p}{1-p}\frac{1}{\tau}h^{0a}\pa_a h^{ij} + \frac{(32p)^{1/2}}{1-p}g^{\lambda(i}h^{j)0}\pa_\lambda \psi\Big],\\
        \mc{N}_1^{(\psi)} :=& \frac{1}{\tau} \Big[\frac{3p-1}{1-p}h^{00}\pa_\tau \psi\Big],
    \end{align}
\end{subequations}
\begin{equation}
    \mc{N}_2^{\mu\nu} := \>2g^{\mu\alpha}g^{\nu\alpha}\pa_\alpha\psi\pa_\beta \psi + \mc{F}^{\mu\nu}.
\end{equation}
\begin{remark}
    The error terms are broken up into the terms in $\mc{N}_0$ that are quadratic in $(\tau^{-1}h^{00},\tau^{-1}h^{0i},\tau^{-1}\psi)$, terms in $\mc{N}_2$ that are quadratic in $(\pa h^{\mu\nu},\pa \psi)$, and the $\mc{N}_1$ are comprised of ``mixed'' error terms that are quadratic in one each of $(\tau^{-1}h^{00},\tau^{-1}h^{0i},\tau^{-1}\psi)$ and $(\pa h^{\mu\nu},\pa \psi)$. 
    
    We point out that none of the error terms depend on the undifferentiated horizontal metric perturbations $h^{ij}$. This is important for later estimates of these error terms, as the $h^{ij}$ do not decay at leading order, while the other perturbed quantities do decay.
\end{remark}
\section{The global existence theorem}\label{sec:globex}
In this section we give the full statement of the main theorem of this paper, Theorem~\ref{sec:globex.thm:globalexistence}, and begin its proof.

\paragraph{Norms for global existence.}
There are several norms that are central to the global existence theorem. We fix an integer $K \geq 3$, and fix a parameter $\delta > 0$ sufficiently small that $5\delta \leq q - 1$. Since $q > 1$ for all $p > 2/3$, such a choice can always be achieved.

We define the following norms for the metric and scalar field perturbations:
\begin{subequations}
    \begin{align}
        S_{h^{00},K} &:= \tau^{q-1}\| h^{00}\|_{L^2(\TT^3)} + \tau^{q-\delta}\|\pa  h^{00}\|_{H^K(\TT^3)},\\
        S_{h^{0*},K} &:= \sum_{a=1}^3\Big\{\tau^{q-\delta}\| h^{0a}\|_{L^2(\TT^3)} + \tau^{q-\delta}\|\pa  h^{0a}\|_{H^K(\TT^3)}\Big\},\\
        S_{h^{**},K} &:= \sum_{a,b=1}^3\Big\{\| h^{ab}\|_{L^2(\TT^3)} + \tau^{q-\delta}\|\pa  h^{ab}\|_{H^K(\TT^3)}\Big\},\\
        S_{\psi,K} &:=\tau^{q-1+\delta}\|\psi\|_{L^2(\TT^3)} + \tau^q\|\pa\psi\|_{H^K(\TT^3)},\\
        S_K &:= S_{h^{00},K} + S_{h^{0*},K} + S_{h^{**},K} +S_{\psi,K}.
    \end{align}
\end{subequations}
Central to the global existence argument is proving that the norm $S_K$ is small for all time.
\begin{remark}
    In our argument, we fix $\delta$ for the rest of the paper so that $0 < 5\delta < q-1$. The parameter $\delta$ can in fact be arbitrarily small, and assuming the initial data still satisfies the smallness condition \eqref{sec:gaugedeqs.eq:bounddata}, the global existence part of Theorem~\ref{sec:globex.thm:globalexistence} goes through unchanged. In this way one may obtain (slightly) improved asymptotics for the resulting perturbations.
\end{remark}

We now state the main theorem.
\begin{theorem}[Global existence of the Einstein-nonlinear scalar field system]\label{sec:globex.thm:globalexistence}
    Fix an exponent $p \in (2/3,1)$, an integer $K \geq 3$, and an initial time $\tau_0 > 0$. Consider initial data $\wt{g}_0^{\mu\nu}, \phi_0 \in H^{K+1}(\TT^3)$, $\wt{g}_1^{\mu\nu}, \phi_1 \in H^K(\TT^3)$ for $\mu,\nu=0,1,2,3$ so that
    \begin{equation}\label{sec:globex.eq:initdata}
        \big(\wt{g}^{\mu\nu},\pa_\tau \wt{g}^{\mu\nu},\phi,\pa_\tau \phi \big)\Big|_{\tau = \tau_0} = (\wt{g}_0^{\mu\nu},\wt{g}_1^{\mu\nu},\phi_0,\phi_1).
    \end{equation}
    Assume that $\wt{g}_0^{00} < 0$ and that $(\wt{g}_0^{ij})_{i,j=1,2,3}$ is the inverse of a Riemannian metric on $\TT^3$. Suppose additionally that $(\wt{g}_0^{\mu\nu},\wt{g}_1^{\mu\nu},\phi_0,\phi_1)$ satisfy the Hamiltonian and momentum constraint equations
    \begin{subequations}
        \begin{align}
            \Big(\Ric_{00}[\wt{g}]- \frac{1}{2}R\wt{g}_{00} - T_{00}[\wt{g},\phi]\Big)\Big|_{\tau = \tau_0} &=0,\\
            \Big(\Ric_{0i}[\wt{g}]- \frac{1}{2}R\wt{g}_{0i} - T_{0i}[\wt{g},\phi]\Big)\Big|_{\tau = \tau_0} &=0,\quad i=1,2,3,
        \end{align}
    \end{subequations}
    as well as the initial wave coordinate condition
    \begin{equation}
        \big(\Gamma^\mu - \mc{Y}^\mu\big)\Big|_{\tau = \tau_0} = 0,\quad \mu=0,1,2,3.
    \end{equation}
    \begin{enumerate}
        \item Then for some $T > \tau_0$, there exists a unique classical solution $([\tau_0,T)\times \TT^3,g,\phi)$ to the Einstein nonlinear scalar field system \eqref{sec:gaugedeqs.eq:wave} with initial data \eqref{sec:globex.eq:initdata} and regularity properties
        \begin{gather}
            \wt{g}^{\mu\nu},\phi \in C_{\mathrm{loc}}^0\big([\tau_0,T), H^{K+1}(\TT^3)\big), \qquad
            \pa_\tau \wt{g}^{\mu\nu}, \pa_\tau \phi \in C_{\mathrm{loc}}^0\big([\tau_0,T), H^{K}(\TT^3)\big).
        \end{gather}
        The wave coordinate condition
        \begin{equation}\label{sec:globex.eq:wavecoord}
            \Gamma^\mu- \mc{Y}^\mu= 0,\quad \mu=0,1,2,3
        \end{equation}
        is satisfied for all $(\tau,x) \in [\tau_0,T)\times\TT^3$ and so, letting $g^{\mu\nu} = \tau^{2p/(1-p)}\wt{g}^{\mu\nu}$, the functions
        $h^{\mu\nu} = g^{\mu\nu}-m^{\mu\nu}$, $\psi = \phi-\phi_b$ satisfy the gauge-fixed equations \eqref{sec:gaugedeqs.eq:wave} on $[\tau_0,T)\times\TT^3$.
        
        \item There exists a constant $\ve_0 > 0$ sufficiently small such that if $\ve < \ve_0$ and the solution $(\wt{g},\phi)$ satisfies the initial bound
        \begin{equation}\label{sec:gaugedeqs.eq:bounddata}
        S_K(\tau_0) \leq \ve^2,
        \end{equation}
        then the time interval on which the solution exists is $[\tau_0,\infty)$, i.e. the solution to the ENSF system exists for all time. Additionally, the bound
        \begin{equation}
            S_K(\tau) \leq \ve
        \end{equation}
        holds for all $\tau \geq \tau_0$.
        \item The spacetime $(\mc{M},\wt{g})$ is future-causally geodesically complete, where $\mc{M} = [\tau_0,\infty)\times \TT^3$. That is, all future-causal inextendible geodesics $\gamma: (s_-,s_+) \ra \mc{M}$ are future-complete, i.e. $s_+ = \infty$.
        
        \item The future-global perturbed solution $(\wt{g},\phi)$ from above has the following asymptotics. The quantities $g^{00}$, $g^{0i}$, $\phi$ decay to leading order to their background counterparts, so that
        \begin{subequations}\label{sec:globex.eq:asymp}
            \begin{align}
                \tau^{2p/(1-p)}\|\wt{g}^{00} - \wt{g}_b^{00}\|_{W^{K-1,\infty}} &\leq \frac{C\ve}{\tau^{p/(1-p)-2}},\label{sec:globex.eq:asymplapse}\\
                \tau^{2p/(1-p)}\|\wt{g}^{0i}\|_{W^{K-1,\infty}} &\leq \frac{C\ve}{\tau^{p/(1-p) - 1-\delta}},\qquad i=1,2,3,\label{sec:globex.eq:asympshift}\\
                \frac{1}{\log \tau}\|\phi - \phi_b\|_{W^{K-1,\infty}} &\leq \frac{C\ve}{\tau^{p/(1-p)-2+\delta}\log \tau}.\label{sec:globex.eq:asympscalar}
            \end{align}
        \end{subequations}
        For the horizontal metric components, there exists a symmetric, positive definite $3\times 3$ matrix $(g_\infty^{ij})_{i,j=1,2,3} \in \RR^{3\times 3}$ such that
        \begin{equation}
            \|\tau^{2p/(1-p)}\wt{g}^{ij} - g_{\infty}^{ij}\|_{W^{K-2,\infty}} \leq \frac{C\ve}{\tau^{p/(1-p)-2}}, \qquad i,j=1,2,3,\label{sec:globex.eq:asymphori}
        \end{equation}
        moreover $g_\infty^{ij}$ satisfies $|g_\infty^{ij} - \delta^{ij}| \leq C\ve$ for $i,j=1,2,3$.
    \end{enumerate}
\end{theorem}
\begin{remark}[Local theory]
    Part (I) of Theorem~\ref{sec:globex.thm:globalexistence} is a local existence statement that is implicitly used in our global existence argument. The local theory for the Einstein equations is by now standard, and we refer the reader to the texts \cite{Rin:Cauchyproblembook:09,Cho:grbook:15} for thorough treatments of this topic. The propagation of the wave coordinate condition $\mc{D}^{\mu} = 0$ (which implies the equivalence of the Einstein equations and the gauge-fixed equations) is central to the local theory, and were used in the foundational works \cite{Cho:Einsteinlwp:52} by Choquet-Bruhat and \cite{ChoGer:mghd:69} by Choquet-Bruhat/Geroch to prove local well-posedness of the Einstein equations.
\end{remark}
\begin{remark}[Asymptotics]\label{sec:globex.rmk:asymptotics}
    The asymptotics described in part (IV) of Theorem~\ref{sec:globex.thm:globalexistence} are expressed in terms of the inverse metric $\wt{g}^{\mu\nu}$ in the coordinates $(\tau,x)$ fixed by the wave coordinate condition \eqref{sec:gaugedeqs.eq:waveconstraint}. We note that since $5\delta < p/(1-p)-2$, the right hand side of the inequalities \eqref{sec:globex.eq:asymp} and \eqref{sec:globex.eq:asymphori} all decay in $\tau$.
    
    We can also present the asymptotics in coordinates $(t,x)$, where $t = [(1-p)\tau]^{1/1-p}$ is the time coordinate with respect to which the background metric $\wt{g}_b$ takes the usual form $\wt{g}_b = -\de t^2 +t^{2p}\delta_{ij}\de x^i\de x^j$. The metric components in coordinates $(t,x)$ then satisfy
    \begin{align}
        \wt{g}^{tt}(t,x) = \Big(\frac{\de t}{\de \tau}\Big)^2 \wt{g}^{00}(\tau(t),x),\quad \wt{g}^{it}(t,x) = \frac{\de t}{\de \tau}\wt{g}^{0i}(\tau(t),x),\qquad \wt{g}^{ij}(t,x) = \wt{g}^{ij}(\tau(t),x).
    \end{align}
    Since $t'(\tau) = t^p$, the bounds \eqref{sec:globex.eq:asymp} and \eqref{sec:globex.eq:asymphori} imply the following bounds for the components $\wt{g}^{tt}$, $\wt{g}^{it}$, $\wt{g}^{ij}$, $\phi$:
    \begin{gather}
        \|\wt{g}^{tt} + 1\|_{W^{K-1,\infty}} \leq \frac{C\ve}{t^{3p-2}},\qquad
        t^p\|\wt{g}^{it}\|_{W^{K-1,\infty}} \leq \frac{C\ve}{t^{2p-1-\delta'}},\qquad
        \|t^{2p}\wt{g}^{ij}-g_\infty^{ij}\|_{W^{K-2,\infty}} \leq \frac{C\ve}{t^{3p-2}},\\
        \frac{1}{\log t}\|\phi - \phi_b\|_{W^{K-1,\infty}} \leq \frac{C\ve}{t^{3p-2+\delta'}\log t},
    \end{gather}
    where $\delta' = (1-p)\delta$. Using algebraic estimates like those found in Lemma 2 and Lemma 7 in \cite{Rin:deSitterENSFstab08}, one can also derive the following asymptotics for the lowered metric components in $(t,x)$ coordinates:
    \begin{equation}
        \|\wt{g}_{tt} + 1\|_{W^{K-1,\infty}} \leq \frac{C\ve}{t^{3p-2}},\qquad
        t^{-p}\|\wt{g}_{it}\|_{W^{K-1,\infty}} \leq \frac{C\ve}{t^{2p-1-\delta'}},\qquad \|t^{-2p}\wt{g}_{ij}-(g_\infty)_{ij}\|_{W^{K-2,\infty}} \leq \frac{C\ve}{t^{3p-2}},
    \end{equation}
    where $(g_\infty)_{ij}$ is the inverse matrix of $(g_\infty^{ij})_{i,j=1,2,3}$. This yields the asymptotic description of $\wt{g}$ presented in Theorem~\ref{sec:intro.thm:main} in the introduction.
\end{remark}
\begin{paragraph}{Proof of parts (II), (III), and (IV) of Theorem~\ref{sec:globex.thm:globalexistence}.}
Part (II) of Theorem~\ref{sec:globex.thm:globalexistence} comprises a global existence statement, part (III) states that the perturbed spacetimes are geodesically complete, and part (IV) describes the asymptotics of these spacetimes. Before beginning the proof of global existence (which comprises the majority of the rest of the paper), we mention that the geodesic completeness statement from part (III) is proven in Proposition~\ref{sec:completethm.prop:complete}, while the asymptotics of part (IV) are proven in Proposition~\ref{sec:completethm.prop:asymp}.

We now begin the global existence part of the proof. Due to the propagation of the wave coordinate condition \eqref{sec:globex.eq:wavecoord}, the Einstein-nonlinear scalar field system is equivalent to the gauge-fixed system \eqref{sec:gaugedeqs.eq:wave}. Thus it suffices to prove global existence of the solutions $(h^{\mu\nu},\psi)$ of the gauged-fixed system. We do this by a bootstrap argument. Let $I \subset [\tau_0,\infty)$ denote the connected set where $T \in I$ if and only if the unique solution to the gauge-fixed system \eqref{sec:gaugedeqs.eq:wave} exists for all $\tau_0 \leq \tau \leq T$ and obeys the bound
\begin{equation}\label{sec:globex.eq:bootstrap}
    S_K(\tau) \leq \ve
\end{equation}
for all $\tau_0 \leq \tau \leq T$. We will henceforth call \eqref{sec:globex.eq:bootstrap} the bootstrap assumption. Clearly $I$ is nonempty, as there exists a unique solution in a neighbourhood of $\tau_0$ that obeys \eqref{sec:globex.eq:bootstrap}, this follows by the local existence part of Theorem~\ref{sec:globex.thm:globalexistence}. Moreover, $I$ is closed by standard continuity properties of the equations. If we can show that $I$ is open in $[\tau_0,\infty)$, that would then imply that $I = [\tau_0,\infty)$, and so the solution $(g^{\mu\nu},\phi)$ exists for all time. Suppose $T \in I$, so that the bootstrap assumption holds for all $\tau_0 \leq \tau \leq T$. By part (I) of Theorem~\ref{sec:globex.thm:globalexistence}, the solution can be extended to a small neighbourhood in time of $T$, and so it remains to improve the bootstrap assumption \eqref{sec:globex.eq:bootstrap}. This is proven in Theorem~\ref{sec:completethm.thm:bootimp}, using the various estimates derived in  Sections~\ref{sec:bootbound} -~\ref{sec:completethm}.
\qed
\end{paragraph}
\section{Bounds from the bootstrap assumption}\label{sec:bootbound}
From this section until Theorem~\ref{sec:completethm.thm:bootimp}, we assume the hypotheses of Theorem~\ref{sec:globex.thm:globalexistence}, and we assume that the bootstrap assumption \eqref{sec:globex.eq:bootstrap} holds on the interval $[\tau_0,T]$. In this section we prove an array of estimates directly from the bootstrap assumption. These include bounds for the lowered metric components $g_{\mu\nu}$, and for the nonlinear terms $\mc{N}^{\mu\nu}$, $\mc{N}^{(\psi)}$ present in the gauge-fixed system \eqref{sec:gaugedeqs.eq:wave}. We make extensive use of the Sobolev embedding theorems from Section~\ref{sec:prelims.subsec:sobolev}.

First we derive bounds for several miscellaneous quantities which appear later in the global existence argument.
\begin{lemma}\label{sec:bootbound.lem:miscests}
    Assuming the bootstrap assumption \eqref{sec:globex.eq:bootstrap}, the following estimates hold:
    \begin{align}
        \|(g^{00})^{-1} - (m^{00})^{-1}\|_{H^K} &\leq C\ve\tau^{-(q-1-\delta)},\label{sec:bootbound.eq:miscests1}\\
        \|\mc{W}(\psi)\|_{H^K} &\leq C\ve^2\tau^{-2(q-1-\delta)}.\label{sec:bootbound.eq:miscests2}
    \end{align}
\end{lemma}
\begin{proof}
    Both inequalities follow from the estimate \eqref{sec:prelims.eq:sobolevcomp} from Lemma~\ref{sec:prelims.lem:sobolevcomp}, as well as the bootstrap assumption. We introduce the function $F(x) = (m^{00}+x)^{-1}-(m^{00})^{-1}$. Since $F(0) = 0$ and $h^{00}$ is contained in a small ball around $0$, we apply Lemma~\ref{sec:prelims.lem:sobolevcomp} with $l =1$ to $F \circ h^{00}$:
    \begin{equation}
        \|(g^{00})^{-1} - (m^{00})^{-1}\|_{H^K} = \|F\circ h^{00}\|_{H^K} \leq C\|h^{00}\|_{H^K} \leq C\ve \tau^{-(q-1-\delta)},
    \end{equation}
    where the last inequality follows from the bootstrap assumption \eqref{sec:globex.eq:bootstrap}. Similarly for $\mc{W}(\psi)$, we have $\mc{W}(0) = \mc{W}'(0) = 0$. We apply Lemma~\ref{sec:prelims.lem:sobolevcomp} with $l = 2$ to $\mc{W} \circ \psi$, along with the bootstrap assumption \eqref{sec:globex.eq:bootstrap}:
    \begin{equation}
        \|\mc{W}(\psi)\|_{H^K} \leq C\|\psi\|_{H^K}^2 \leq C\ve^2\tau^{-2(q-1-\delta)},
    \end{equation}
    and we are done.
\end{proof}
\subsection{The lowered metric components}
Our strategy for proving global existence revolves around the gauge-fixed system \eqref{sec:gaugedeqs.eq:wave}, which we have written as a system for the inverse metric perturbations $h^{\mu\nu} = g^{\mu\nu} - m^{\mu\nu}$ (as well as the scalar field perturbation $\psi$). While we are primarily estimating the raised metric perturbations, we do need to derive bounds for the lowered metric components $g_{\mu\nu}$, as they appear in the gauge-fixed system in the nonlinear terms $\mc{N}^{\mu\nu}$.
\begin{lemma}[Estimates for the lowered metric components]\label{sec:bootbound.lem:loweredmetests}
    Suppose the bootstrap assumption \eqref{sec:globex.eq:bootstrap} holds. Let $g_{\mu\nu}$ denote the (lowered) metric components of $g$, with inverse $(g^{-1})^{\mu\nu} = g^{\mu\nu}$, and let $m_{\mu\nu}$ denote the (lowered) components of the Minkowski metric
    \begin{equation}
        m = -(1-p)^{-2}\de \tau^2 + \delta_{ij}\de x^i \de x^j.
    \end{equation}
    Then the components $g_{\mu\nu}$ satisfy the following estimates:
    \begin{align}
        \sum_{\alpha,\beta = 0}^3\|g_{\alpha\beta} - m_{\alpha\beta}\|_{L^2} &\leq C\ve,\label{sec:bootbound.eq:loweredmetl2}\\
        \sum_{\alpha,\beta = 0}^3 \|\pa g_{\alpha\beta}\|_{H^K}&\leq C\ve \tau^{-(q-\delta)}.\label{sec:bootbound.eq:loweredmetderiv}
    \end{align}
\end{lemma}
\begin{proof}
    The components of the lowered metric $g_{\mu\nu}$ are rational functions of the inverse metric components $g^{\mu\nu}$, in the sense that
    \begin{equation}
        g_{\mu\nu} = \frac{P_{\mu\nu}(g^{-1})}{Q(g^{-1})},
    \end{equation}
    where $P$, $Q$ are polynomials over $\RR^{4\times 4}$. Then the functions
    \begin{equation}
        F_{\mu\nu}:\RR^{4\times 4} \ra \RR,\qquad F_{\mu\nu}(A) = \frac{P_{\mu\nu}(A +m^{-1})}{Q(A+m^{-1})} - \frac{P_{\mu\nu}(m^{-1})}{Q(m^{-1})}
    \end{equation}
    satisfy $F_{\mu\nu}(A=0) = 0$. Since the components of $g^{-1}$ lie in a small ball around the components of $m^{-1}$ by the bootstrap assumption we may apply Lemma~\ref{sec:prelims.lem:sobolevcomp} to $F_{\mu\nu} \circ h$ to bound
    \begin{equation}
        \sum_{\alpha,\beta=0}^3 \|g_{\alpha\beta} - m_{\alpha\beta}\|_{L^2} \leq C\sum_{\alpha\beta=0}^3 \|h^{\alpha\beta}\|_{L^2} \leq C\ve,
    \end{equation}
    yielding \eqref{sec:bootbound.eq:loweredmetl2}. To prove the estimate \eqref{sec:bootbound.eq:loweredmetderiv} we use the identity
    \begin{equation}
        \pa_\mu g_{\nu\lambda} = -g_{\nu\alpha}g_{\lambda\beta}\pa_\mu g^{\alpha\beta} = -g_{\nu\alpha}g_{\lambda\beta}\pa_\mu h^{\alpha\beta}.
    \end{equation}
    Applying the Sobolev product estimate from Lemma~\ref{sec:prelims.lem:sobolevproduct}, the previous estimate \eqref{sec:bootbound.eq:loweredmetl2}, and the bootstrap assumption \eqref{sec:globex.eq:bootstrap} then implies the bound
    \begin{align}
        \|\pa g_{\mu\nu}\|_{H^K} &\leq C\Big(\sum_{\alpha,\beta=0}^3\|g_{\alpha\beta}\|_{H^K}^2\Big)\sum_{\lambda,\delta = 0}^3\|\pa h^{\lambda\delta}\|_{H^K},\\
        &\leq C\ve\tau^{-(q-\delta)},
    \end{align}
    hence we have \eqref{sec:bootbound.eq:loweredmetderiv}.
\end{proof}
\begin{remark}
    The above estimates are not particularly sharp. We could, for example, bound the lowered lapse $g_{00}$ in terms of $g^{00}$ like $\|g_{00}-m_{00}\|_{L^2} \leq C\|g^{00}-m^{00}\|_{L^2}$, and exploit the extra decay of $h^{00}$ (this also applies to the lowered shift). However this is not necessary as all lowered metric components in the gauge-fixed system \eqref{sec:gaugedeqs.eq:wave} are differentiated at least once (see Proposition~\ref{sec:bootbound.prop:errorests}). Hence we may use the estimate \eqref{sec:bootbound.eq:loweredmetderiv} instead.
\end{remark}
\subsection{Estimates of the nonlinear terms from the gauge-fixed system}
In the last part of this section, we derive estimates for the nonlinear error terms present in the gauge-fixed system \eqref{sec:gaugedeqs.eq:wave}.
\begin{proposition}[Error term estimates]\label{sec:bootbound.prop:errorests}
    Suppose the bootstrap assumption \eqref{sec:globex.eq:bootstrap} holds. Then we have the following error estimates:
    \begin{subequations}
        \begin{align}
            \|\mc{N}^{\mu\nu}\|_{H^K} \leq&\> C\ve^2 \tau^{-2(q-\delta)}, \quad \mu,\nu = 0,1,2,3,\label{sec:bootbound.eq:errorests1}\\
            \|\mc{N}^{(\psi)}\|_{H^K} \leq&\> C\ve^2 \tau^{-2(q-\delta)}.\label{sec:bootbound.eq:errorests2}
        \end{align}
    \end{subequations}
\end{proposition}
\begin{proof}
    We write schematic forms of the error, letting $\cdot$ denote any possible contraction or regular product of the given factors.  We will explicitly write the components of the undifferentiated components of $h$ which appear, as the $h^{00}$ and $h^{0i}$ terms decay, while the $h^{ij}$ terms do not. We have:
    \begin{subequations}
        \begin{align}
            \mc{N}^{00} &= g^{-1}\cdot g^{-1} \cdot (\pa  h \cdot \pa  g + \pa \psi \cdot \pa \psi) + \frac{1}{\tau^2}g^{-1} \cdot \mc{W}(\psi),\\
            \mc{N}^{0i} &= g^{-1}\cdot g^{-1} \cdot (\pa  h \cdot \pa  g + \pa \psi \cdot \pa \psi) +  \frac{1}{\tau}g^{-1}\cdot  h^{0i}\cdot \pa \psi +  \frac{1}{\tau^2}h^{0i} \cdot ( h^{00} + \psi  +\mc{W}(\psi)),\\
            \mc{N}^{ij}& = g^{-1}\cdot g^{-1}\cdot (\pa  h \cdot \pa  g + \pa \psi \cdot \pa \psi) +  \frac{1}{\tau}g^{-1}\cdot  h^{0i}\cdot \pa \psi +  \frac{1}{\tau^2}g^{-1}\cdot \mc{W}(\psi),\\
            \mc{N}^{(\psi)} &=   \frac{1}{\tau}(h^{0i} \cdot \pa \psi +  h^{00} \cdot \pa \psi)  +  \frac{1}{\tau^2}( h^{00}\cdot \psi + \mc{W}(\psi)).
        \end{align}
    \end{subequations}
    The $\mc{W}(\psi)$ terms can be estimated by the Sobolev product estimate from Lemma~\ref{sec:prelims.lem:sobolevproduct}, Lemma~\ref{sec:bootbound.lem:miscests} and the bootstrap assumption:
    \begin{equation}
        \tau^{-2}\|g^{\mu\nu}\mc{W}(\psi)\|_{H^K} \leq C\tau^{-2}\|g^{\mu\nu}\|_{H^K}\|\mc{W}(\psi)|\|_{H^K} \leq C\ve^2\tau^{-2(q-\delta)}.
    \end{equation}
    There are no undifferentiated $h^{ij}$ components present in any of the error terms, thus all remaining terms in $\mc{N}^{\mu\nu}$ depend quadratically on $(\tau^{-1}h^{00},\tau^{-1}h^{0i},\tau^{-1}\psi,\pa h^{\mu\nu},\pa g_{\mu\nu},\pa \psi)$, which by the bootstrap assumption and Lemma~\ref{sec:bootbound.lem:loweredmetests} obey the bound
    \begin{equation}
        \frac{1}{\tau}\Big(\|h^{00}\|_{H^K} + \sum_{a = 1}^3 \|h^{0a}\|_{H^K} + \|\psi\|_{H^K}\Big) + \sum_{\mu,\nu=0}^3\|\pa h^{\mu\nu}\|_{H^K} + \sum_{\mu,\nu=0}^3 \|\pa g_{\mu\nu}\|_{H^K} + \|\pa \psi\|_{H^K} \leq C \ve \tau^{-(q-\delta)}.
    \end{equation}
    The estimates \eqref{sec:bootbound.eq:errorests1}, \eqref{sec:bootbound.eq:errorests2} now follow after repeated applications of the Sobolev product Lemma~\ref{sec:prelims.lem:sobolevproduct}, the bootstrap assumption, and the estimates for the lowered metric components from Lemma~\ref{sec:bootbound.lem:loweredmetests}.
\end{proof}
\section{Estimates for the spatially averaged perturbations \texorpdfstring{$h_{\av}^{\mu\nu}$, $\psi_{\av}$}{ }}\label{sec:avgsys}
The goal of this section is to derive a system of equations satisfied by the spatial averages of the quantities $(h^{\mu\nu},\psi)$. Recall that we denote the spatial average of a scalar function $u$ on $\TT^3$ by
\begin{equation}
    u_{\av}:=\frac{1}{(2\pi)^3}\int_{\TT^3}u(x)\de^3 x,
\end{equation}
and we denote the oscillatory functions with zero average by $u_{\osc} = u - u_{\av}$.
The two functions $u_{\osc}$, $u_{\av}$ satisfy the bounds
\begin{equation}\label{sec:avgsys.eq:avgbound}
    |u_\av| \leq C\|u\|_{L^2},\quad \|u_{\osc}\|_{L^2}\leq C\|u\|_{L^2},
\end{equation}
which follow from the Cauchy-Schwartz inequality on $\TT^3$.

We analyse the spatial averages of $h_{\av}^{\mu\nu}$, $\psi_{\av}$ and the functions $h_{\osc}^{\mu\nu}$, $\psi_{\osc}$ separately, as the equations for these quantities decouple (up to some decaying error), and their principle systems yield substantially different asymptotics. On the one hand, the equations satisfied by the averages $h_{\av}^{\mu\nu}$, $\psi_{\av}$ do not possess terms which depend linearly on spatial derivatives, which can limit decay. On the other hand, the $L^2$ norm of the oscillating remainders $h_{\osc}^{\mu\nu}$, $\psi_{\osc}$ can be estimated directly by the following Poincar\'e inequality.

For a given scalar function $u$, we have the following trivial commutator identities
\begin{gather}
    \pa_\tau (u_{\av}) = (\pa_\tau u)_{\av},\qquad \pa_\mu (u_{\osc}) = (\pa_\mu u)_{\osc}, \quad \mu=0,1,2,3,\\
    \pa_i (u_{\av}) = (\pa_i u)_{\av} = 0,\quad i=1,2,3.
\end{gather}
Due to the first two equalities, we will unambiguously write $\pa_\tau u_{\av}$, $\pa_\tau u_{\osc}$. Moreover, by the last line we have $\pa_i u = \pa_i u_{\osc}$. We also observe that since the background solution is spatially homogeneous, the remainder functions satisfy $g^{\mu\nu}_{\osc} =  h^{\mu\nu}_{\osc}$ and $\phi_{\osc} = \psi_{\osc}$.

We will derive the equations for the averaged horizontal metric components $h_{\av}^{ij}$ and the averaged scalar field $\psi_{\av}$ from the evolution equations. Then we will derive equations for the averaged lapse and shift perturbations $h_{\av}^{00}$, $h_{\av}^{0i}$ directly from the gauge equations
\begin{equation}
    \Gamma^\mu(g) = \mc{Y}^{\mu}(g).
\end{equation}
We remind the reader that we continue to assume the bootstrap assumption \eqref{sec:globex.eq:bootstrap} throughout this section.
\subsection{ODE for the averaged horizontal metric and scalar field perturbations}
We recall the equations for satisfied by the horizontal metric perturbations and scalar field perturbation:
\begin{align}
    &\wh{\square}_g  h^{ij} + \frac{2q+2}{\tau}g^{00}\pa_\tau  h^{ij} - \mc{L}^{ij} - \mc{N}^{ij}=0,\label{sec:gaugedeqs.eq:wavehorimeansection}\\
    &\wh{\square}_g \psi +\frac{2q+2}{\tau}g^{00}\pa_\tau \psi + \frac{2(2q+1)(q+2)}{\tau^2}g^{00}\psi - \mc{N}^{(\psi)}= 0,\label{sec:gaugedeqs.eq:wavescalarmeansection}
\end{align}
where the functions $\mc{L}^{ij}$ are
\begin{equation}
    \mc{L}^{ij} = -\frac{(8p)^{1/2}(3p-1)}{\tau^2}g^{ij}\psi.
\end{equation}

We derive ordinary differential equations for $h_{\av}^{ij}$, $\psi_{\av}$ using these equations. These will be second order ODE, due to the wave equations from which they are derived. 
\begin{lemma}[ODE for the averaged horizontal metric and scalar field perturbations]\label{sec:avgsys.lem:odehoriscalar}
    The spatial averages of the horizontal metric and scalar field perturbations satisfy the differential inequalities
    \begin{subequations}
        \begin{gather}
            \Big|\pa_\tau^2 h_{\av}^{ij}(\tau) + \frac{2q+2}{\tau}\pa_\tau  h_{\av}^{ij}(\tau)\Big| \leq C\Big(\frac{1}{\tau^2}|\psi_{\av}(\tau)| + \ve^2\tau^{-2(q-\delta)}\Big),\quad i,j=1,2,3,\label{sec:avgsys.eq:odeavghori}\\
            \Big|\pa_\tau^2 \psi_{\av}(\tau) + \frac{2q+2}{\tau}\pa_\tau \psi_{\av}(\tau) + \frac{2(2q+1)(q+2)}{\tau^2}\psi_{\av}(\tau)\Big| \leq C\ve^2\tau^{-2(q-\delta)}.\label{sec:avgsys.eq:odeavgscalar}
        \end{gather}
    \end{subequations}
\end{lemma}
\begin{proof}
    First we derive \eqref{sec:avgsys.eq:odeavgscalar} by multiplying the equation \eqref{sec:gaugedeqs.eq:wavescalarmeansection} by $(g^{00})^{-1}$ and integrating over $\TT^3$, which yields
    \begin{align}
        0 =\>& \frac{1}{(2\pi)^3}\int_{\TT^3}(g^{00})^{-1}\Big[\wh{\square}_g \psi +\frac{2q+2}{\tau}g^{00}\pa_\tau \psi + \frac{2(2q+1)(q+2)}{\tau^2}g^{00}\psi - \mc{N}^{(\psi)}\Big]\nonumber\\
        =\>& \pa_\tau^2 \psi_{\av} + \frac{2q+2}{\tau}\pa_\tau \psi_{\av} + \frac{2(2q+1)(q+2)}{\tau^2}\psi_{\av}\nonumber\\
        &+\frac{1}{(2\pi)^3}\int_{\TT^3}(g^{00})^{-1}\Big[2h^{0a}\pa_a\pa_\tau \psi + g^{ab}\pa_a\pa_b\psi - \mc{N}^{(\psi)}\Big]. \label{sec:avgsys.eq:odeavgscalarproof1}   
    \end{align}
    Then we estimate the integral in the last line. The bound for $\mc{N}^{(\psi)}$ from Proposition~\ref{sec:bootbound.prop:errorests}, the bound for $(g^{00})^{-1}$ from Lemma~\ref{sec:bootbound.lem:miscests}, and Cauchy-Schwartz give
    \begin{equation}
        \Big|\int_{\TT^3}(g^{00})^{-1}\mc{N}^{(\psi)}\Big| \leq \|(g^{00})^{-1}\|_{L^2}\|\mc{N}^{(\psi)}\|_{L^2} \leq C\ve^2 \tau^{-2(q-\delta)}.\label{sec:avgsys.eq:odeavgscalarproof2}
    \end{equation}
    For the other terms in the integral in \eqref{sec:avgsys.eq:odeavgscalarproof1}, we integrate-by parts and estimate via Cauchy-Schwartz, the Sobolev embedding, Lemma~\ref{sec:bootbound.lem:miscests} to bound $(g^{00})^{-1})$, and the bootstrap assumption:
    \begin{subequations}
    \begin{align}
            \Big|\int_{\TT^3}(g^{00})^{-1}h^{0a}\pa_a \pa_\tau \psi\Big| \leq &\> \Big|\int_{\TT^3}(g^{00})^{-2}h^{0a}\pa_a h^{00}\pa_\tau \psi\Big| + \Big|\int_{\TT^3}(g^{00})^{-1}\pa_a h^{0a}\pa_\tau \psi\Big|\nonumber\\
            \leq &\> \sum_{a=1}^3\big(\|(g^{00})^{-1}\|_{L^\infty}^2\|h^{0a}\|_{L^\infty}\|\pa_a h^{00}\|_{L^2}\|\pa_\tau \psi\|_{L^2}+\|(g^{00})^{-1}\|_{L^\infty} \|\pa_a h^{0a}\|_{L^2}\|\pa_\tau \psi\|_{L^2}\big)\nonumber\\ 
            \leq&\> C\ve^2 \tau^{-2(q-\delta)},\label{sec:avgsys.eq:odeavgscalarproof3}\\
            \Big|\int_{\TT^3}(g^{00})^{-1}g^{ab}\pa_a\pa_b \psi\Big| \leq&\>\Big|\int_{\TT^3}(g^{00})^{-2}g^{ab}\pa_a h^{00}\pa_b \psi\Big| + \Big|\int_{\TT^3}(g^{00})^{-1}\pa_a g^{ab}\pa_b \psi\Big|\nonumber\\
            \leq&\>\sum_{a,b=1}^3\big(\|(g^{00})^{-1}\|_{L^\infty}^2\|g^{ab}\|_{L^\infty}\|\pa_a h^{00}\|_{L^2}\|\pa_b \psi\|_{L^2}+\|(g^{00})^{-1}\|_{L^\infty} \|\pa_a h^{ab}\|_{L^2}\|\pa_b \psi\|_{L^2}\big)\nonumber\\ 
        \leq&\> C\ve^2 \tau^{-2(q-\delta)}.\label{sec:avgsys.eq:odeavgscalarproof4}
        \end{align}
    \end{subequations}
    Plugging the estimates  \eqref{sec:avgsys.eq:odeavgscalarproof2} - \eqref{sec:avgsys.eq:odeavgscalarproof4} into the equation \eqref{sec:avgsys.eq:odeavgscalarproof1}, we obtain \eqref{sec:avgsys.eq:odeavgscalar}. We proceed analogously for \eqref{sec:avgsys.eq:odeavghori}, multiplying \eqref{sec:gaugedeqs.eq:wavehorimeansection} by $(g^{00})^{-1}$ and integrating over the torus:
    \begin{align}
        0 = &\>\pa_\tau^2 h_{\av}^{ij} + \frac{2(q+1)}{\tau}\pa_\tau h_{\av}^{ij} + \frac{(8p)^{1/2}(3p-1)}{(2\pi)^3\tau^2}\int_{\TT^3}(g^{00})^{-1}g^{ij}\psi\nonumber\\
        &+\frac{1}{(2\pi)^3}\int_{\TT^3}(g^{00})^{-1}\Big[2h^{0a}\pa_a\pa_\tau h^{ij} + g^{ab}\pa_a\pa_b h^{ij} - \mc{N}^{ij}\Big].\label{sec:avgsys.eq:odeavghoriproof1}
    \end{align}
    We estimate the latter integral in \eqref{sec:avgsys.eq:odeavghoriproof1} the same way as we did the for the last integral in \eqref{sec:avgsys.eq:odeavgscalarproof1} (using Lemmas~\ref{sec:bootbound.prop:errorests} and \ref{sec:bootbound.lem:miscests}) and calculate a similar bound:
    \begin{equation}\label{sec:avgsys.eq:odeavghoriproof2}
        \Big|\int_{\TT^3}(g^{00})^{-1}\Big[2h^{0a}\pa_a\pa_\tau h^{ij} + g^{ab}\pa_a\pa_b h^{ij} - \mc{N}^{ij}\Big]\Big| \leq C\ve^2 \tau^{-2(q-\delta)}.
    \end{equation}
    For the former integral in \eqref{sec:avgsys.eq:odeavghoriproof1}, we decompose the integrand $(g^{00})^{-1}g^{ij}\psi$ into averaged and oscillating factors, beginning with
    \begin{equation}\label{sec:avgsys.eq:odeavghoriproof3}
        \int(g^{00})^{-1}g^{ij}\psi =\> \int_{\TT^3}(g^{00})^{-1}g^{ij} \psi_{\av}+ \int_{\TT^3}(g^{00})^{-1} g^{ij}\psi_{\osc}.
    \end{equation}
    We bound the first integral on the RHS of \eqref{sec:avgsys.eq:odeavghoriproof3} using Cauchy-Schwartz, the Sobolev embedding, Lemma~\ref{sec:bootbound.lem:miscests} and the bootstrap assumption:
    \begin{equation}\label{sec:avgsys.eq:odeavghoriproof4}
        \Big|\int_{\TT^3}(g^{00})^{-1}g^{ij} \psi_{\av}\Big| \leq C\|(g^{00})^{-1}\|_{L^\infty}\|g^{ij}\|_{L^2}|\psi_{\av}| \leq C|\psi_{\av}|.
    \end{equation}
    The second integral on the RHS of \eqref{sec:avgsys.eq:odeavghoriproof3} appears to be linear, but upon splitting the factors further into averaged and oscillating factors, we find that it depends nonlinearly on $(\pa h^{\mu\nu},\pa \psi)$. We have
    \begin{align}
        \int_{\TT^3}(g^{00})^{-1} g^{ij}\psi_{\osc} 
        =&\>\int_{\TT^3}\psi_{\osc}\Big[((g^{00})^{-1})_{\av} g_{\av}^{ij} + ((g^{00})^{-1})_{\av} h_{\osc}^{ij} + ((g^{00})^{-1})_{\osc} g_{\av}^{ij} + ((g^{00})^{-1})_{\osc} h_{\osc}^{ij}\Big]\nonumber\\
        =&\>\int_{\TT^3}\psi_{\osc}\Big[((g^{00})^{-1})_{\av} h_{\osc}^{ij} + ((g^{00})^{-1})_{\osc} g_{\av}^{ij} + ((g^{00})^{-1})_{\osc} h_{\osc}^{ij}\Big].
    \end{align}
    where for the second equality we used the fact that the term $\psi_{\osc}((g^{00})^{-1})_{\av} g_{\av}^{ij}$ has zero spatial average. We estimate the remaining terms using (in rough order) Cauchy-Schwartz, the Sobolev embedding, Lemma~\ref{sec:bootbound.lem:miscests}, the Poincar\'e inequality from Lemma~\ref{sec:prelims.lem:poincare}, the bound \eqref{sec:avgsys.eq:avgbound}, and the bootstrap assumption to obtain
    \begin{align}
        \Big|\int(g^{00})^{-1}g^{ij}\psi_\osc\Big| \leq&\> \Big|\int_{\TT^3}(g^{00})^{-1}h_{\osc}^{ij} \psi_{\osc}\Big| + |g_{\av}^{ij}|\cdot\Big|\int_{\TT^3}((g^{00})^{-1})_{\osc}\psi_{\osc}\Big|\nonumber\\
        \leq &\> C\big(\|(g^{00})^{-1}\|_{L^\infty}\|h_{\osc}^{ij}\|_{L^2}\|\psi_{\osc}\|_{L^2} + |g_{\av}^{ij}|\cdot\|((g^{00})^{-1})_{\osc}\|_{L^2}\|\psi_{\osc}\|_{L^2}\big)\nonumber\\
        \leq &\>C\big(\|(g^{00})^{-1}\|_{H^2}\|\ol{\pa}h^{ij}\|_{L^2}\|\ol{\pa}\psi\|_{L^2} + \|g^{ij}\|_{L^2}\|\ol{\pa}h^{00}\|_{L^2}\|\ol{\pa}\psi\|_{L^2}\big)\nonumber\\
        \leq&\> C\ve^2 \tau^{-2(q-\delta)}\label{sec:avgsys.eq:odeavghoriproof5}.
    \end{align}
    Plugging the estimates \eqref{sec:avgsys.eq:odeavghoriproof2}, \eqref{sec:avgsys.eq:odeavghoriproof4} and \eqref{sec:avgsys.eq:odeavghoriproof5} into \eqref{sec:avgsys.eq:odeavghoriproof1} yields \eqref{sec:avgsys.eq:odeavghori}.
\end{proof}
\begin{remark}[ODE for the averaged lapse and shift from the gauge-fixed system]\label{sec:avgsys.rmk:odelapseshiftextra}
    Second-order differential inequalities can also be derived from the wave equations for the averaged lapse and shift perturbations $h_{\av}^{00}$, $h_{\av}^{0i}$. By the same argument as in Lemma~\ref{sec:avgsys.lem:odehoriscalar}, $h_{\av}^{00}$ and $h_{\av}^{0i}$ satisfy the following inequalities:
    \begin{subequations}
        \begin{gather}
            \Big|\pa_\tau^2 h_{\av}^{00}(\tau) + \frac{2q}{\tau}\pa_\tau h_{\av}^{00} (\tau) - \frac{4q+2}{\tau^2}h_{\av}^{00}(\tau)\Big| \leq\>C\Big(\frac{1}{\tau^2}|\psi_{\av}(\tau)|+\frac{1}{\tau}|\pa_\tau \psi_{\av}(\tau)|+\ve^2 \tau^{-2(q-\delta)}\Big),\label{sec:avgsys.eq:odeavglapseextra}\\
            \Big|\pa_\tau^2 h_{\av}^{0i}(\tau) + \frac{2q+2}{\tau}\pa_\tau h_{\av}^{0i} (\tau) - \frac{2q+2}{\tau^2}h_{\av}^{0i}(\tau)\Big| \leq\>C\ve^2 \tau^{-2(q-\delta)}.\label{sec:avgsys.eq:odeavgshiftextra}
        \end{gather}
    \end{subequations}
    These are derived from the wave equations \eqref{sec:gaugedeqs.eq:wavelapse}, \eqref{sec:gaugedeqs.eq:waveshift}. Note the absence of linear terms on the right hand side of \eqref{sec:avgsys.eq:odeavgshiftextra}, despite the presence of a mixed linear term $\mc{L}^{0i}$ in \eqref{sec:gaugedeqs.eq:waveshift}. To see why this is the case, computing the spatial average of $(g^{00})^{-1}\mc{L}^{0i}$ we get
    \begin{align}
        \int(g^{00})^{-1}\mc{L}^{0i}= &\>\frac{1}{\tau}\int_{\TT^3}\Big[(g^{00})^{-1}g^{ia}\pa_a h^{00} + \frac{(8p)^{1/2}}{1-p}g^{ia}\pa_a \psi\Big]\nonumber\\
        =&-\frac{1}{\tau}\int_{\TT^3}\Big[\pa_a((g^{00})^{-1} g^{ia})h^{00} + \frac{(8p)^{1/2}}{1-p} \psi\pa_a h^{ia}\Big],
    \end{align}
    where we integrated by parts in the second line. In an analogous manner to \eqref{sec:avgsys.eq:odeavgscalarproof3}, \eqref{sec:avgsys.eq:odeavgscalarproof4} in the previous lemma, this reveals that these terms are in fact error terms, and can be estimated as such.

    We will not use \eqref{sec:avgsys.eq:odeavglapseextra}, \eqref{sec:avgsys.eq:odeavgshiftextra} to estimate $h_{\av}^{00}$, $h_{\av}^{0i}$ directly, as the principle system of these equations (the left hand side of \eqref{sec:avgsys.eq:odeavglapseextra}, \eqref{sec:avgsys.eq:odeavgshiftextra}) produces not decay but growth, due to the bad signs of the
    \begin{equation}
        -\frac{4q+2}{\tau^2}h_{\av}^{00},\qquad -\frac{2q+2}{\tau^2}h_{\av}^{0i}
    \end{equation}
    terms present in \eqref{sec:avgsys.eq:odeavglapseextra}, \eqref{sec:avgsys.eq:odeavgshiftextra} respectively. However, we will use them  (along with \eqref{sec:avgsys.eq:odeavghori} and \eqref{sec:avgsys.eq:odeavgscalar}) to derive energy estimates for the oscillating remainders ($h_{\osc}^{\mu\nu},\psi_\osc)$ in Section~\ref{sec:oscsys}.
\end{remark}
\subsection{ODE for the averaged lapse and shift perturbations}
Next we derive a system of first-order ODE for the spatial averages of the lapse and shift metric $h_{\av}^{00}$, $h_{\av}^{0i}$. We derive these from the gauge constraint equations $D^\mu = \Gamma^\mu(g) - \mc{Y}^\mu = 0$, as they have better structure than the evolution equations for $h^{00}$, $h^{0i}$, at the cost of a loss of derivatives. Since we're only using these to estimate the $L^2$ norms of the functions $h^{00}$, $h^{0i}$, $\pa_\tau  h^{00}$, $\pa_\tau h^{0i}$, this is not an issue.

Using the identity
\begin{equation}
    \pa_\lambda g^{\mu\nu} = -g^{\mu\alpha}g^{\nu\beta}\pa_\lambda g_{\alpha\beta},
\end{equation}
The contracted Christoffel symbols satisfy
\begin{align}
    \Gamma^\mu &= g^{\alpha\beta}g^{\mu\lambda}\big(\pa_{(\alpha}g_{\beta)\lambda} - \frac{1}{2}\pa_\lambda g_{\alpha\beta}\big)\nonumber\\
    &= - \pa_\lambda g^{\lambda\mu} + \frac{1}{2}g_{\alpha\beta}g^{\mu\lambda}\pa_\lambda g^{\alpha\beta}.
\end{align}
Keeping in mind that the components of $m$ are constant and $m^{0i} = 0$, this implies that
\begin{subequations}
    \begin{align}
        \Gamma^0 &= -\pa_\tau  h^{00} - \pa_a  h^{0a} + \frac{1}{2}g_{\alpha\beta}g^{00}\pa_\tau  h^{\alpha\beta} + \frac{1}{2}g_{\alpha\beta} h^{0a}\pa_a  h^{\alpha\beta},\\
        \Gamma^i &= -\pa_\tau  h^{0i} - \pa_a  h^{ia} + \frac{1}{2}g_{\alpha\beta} h^{0i}\pa_\tau  h^{\alpha\beta} + \frac{1}{2}g_{\alpha\beta}g^{ia}\pa_a  h^{\alpha\beta},
    \end{align}
\end{subequations}
moreover we expand the term
\begin{align}
    g_{\alpha\beta}g^{00}\pa_\tau  h^{\alpha\beta} &= g_{00}g^{00}\pa_\tau  h^{00} + 2 h_{0i}g^{00}\pa_\tau  h^{0i} + g_{ab}g^{00}\pa_\tau  h^{ab}\nonumber\\
    &= \pa_\tau  h^{00} -  h_{0a} h^{0a}\pa_\tau h^{00} + 2 h_{0i}g^{00}\pa_\tau  h^{0i} + g_{ab}g^{00}\pa_\tau  h^{ab}.
\end{align}
We assume the wave constraint equations \eqref{sec:gaugedeqs.eq:waveconstraint} are satisfied, so that
\begin{equation}
    \Gamma^0(g) = \frac{2q+1}{\tau}h^{00},\qquad \Gamma^i(g) = \frac{2q+2}{\tau} h^{0i},\quad  i = 1,2,3.
\end{equation}
These are therefore equivalent to the following transport equations:
\begin{subequations}
    \begin{align}
        \pa_\tau  h^{00} + \frac{4q+2}{\tau} h^{00} = \>& g_{ab}g^{00}\pa_\tau  h^{ab}-2\pa_a  h^{0a} + h_{0a} h^{0a}\pa_\tau  h^{00} + 2 h_{0a}g^{00}\pa_\tau  h^{0a}+g_{\alpha\beta} h^{0a}\pa_a  h^{\alpha\beta},\label{sec:avgsys.eq:odelapse0}\\
        \pa_\tau  h^{0i} + \frac{2q+2}{\tau} h^{0i} = \>& -\pa_a  h^{ab} +\frac{1}{2}g_{\alpha\beta} h^{0i}\pa_\tau  h^{\alpha\beta} +\frac{1}{2}g_{\alpha\beta}g^{ia}\pa_a  h^{\alpha\beta}.\label{sec:avgsys.eq:odeshift0}
    \end{align}
\end{subequations}
We use equations \eqref{sec:avgsys.eq:odelapse0}, \eqref{sec:avgsys.eq:odeshift0} to derive differential equations satisfied by the averaged lapse and shift.
\begin{lemma}[ODE for the spatial averages of the lapse and shift]\label{sec:avgsys.lem:odelapseshift}
    Suppose $\mc{D}^{\mu}=0$ for $\mu = 0,1,2,3$. The averaged lapse and shift perturbations $h_{\av}^{00}$, $h_{\av}^{0i}$ satisfy the following differential inequalities:
    \begin{subequations}
        \begin{align}
            \Big|\pa_\tau  h_{\av}^{00}(\tau) + \frac{4q+2}{\tau}h_{\av}^{00}(\tau)\Big| \leq&\> C\Big(\sum_{a,b=1}^3|\pa_\tau h_{\av}^{ab}(\tau)| + \ve^2 \tau^{-2(q-\delta)}\Big),\label{sec:avgsys.eq:odeavglapse}\\
            \Big|\pa_\tau   h_{\av}^{0i}(\tau) + \frac{2q+2}{\tau}  h_{\av}^{0i}(\tau)\Big| \leq&\> C\ve^2\tau^{-2(q-\delta)},\quad i=1,2,3.\label{sec:avgsys.eq:odeavgshift}
        \end{align}
    \end{subequations}
\end{lemma}
\begin{proof}
We integrate the equations \eqref{sec:avgsys.eq:odelapse0}, \eqref{sec:avgsys.eq:odeshift0} over the torus. Any terms which are solely spatial derivatives have zero average, giving
\begin{subequations}\label{sec:avgsys.eq:odeavglapseshiftproof}
    \begin{align}
        \pa_\tau h_{\av}^{00} +  \frac{4q+2}{\tau}  h_{\av}^{00} &=  \frac{1}{(2\pi)^3}\int_{\TT^3}g_{ab}g^{00}\pa_\tau  h^{ab} + \frac{1}{(2\pi)^3}\int_{\TT^3}\Big[h_{0a} h^{0a}\pa_\tau  h^{00} + 2 h_{0a}g^{00}\pa_\tau  h^{0a}+g_{\alpha\beta} h^{0a}\pa_a  h^{\alpha\beta}\Big]\label{sec:avgsys.eq:odeavglapseproof},\\
        \pa_\tau   h_{\av}^{0i} + \frac{2q+2}{\tau} h_{\av}^{0i} &= \frac{1}{2(2\pi)^3}\int[g_{\alpha\beta}g^{ia}\pa_a  h^{\alpha\beta} +g_{\alpha\beta} h^{0i}\pa_\tau  h^{\alpha\beta}\Big],\label{sec:avgsys.eq:odeavgshiftproof}
    \end{align}
\end{subequations}
We treat the equation \eqref{sec:avgsys.eq:odeavglapseproof} for $h_{\av}^{00}$ first. Estimating the second integral using Cauchy-Schwartz and the bootstrap assumption:
\begin{align}
    \Big|\int_{\TT^3}\Big[ h_{0a} h^{0a}\pa_\tau  h^{00} +&2 h_{0a}g^{00}\pa_\tau  h^{0a}+g_{\alpha\beta} h^{0a}\pa_a  h^{\alpha\beta}\Big]\Big|\nonumber\\
    \leq&\> C\sum_{\alpha,\beta = 0}^4\sum_{a=1}^3\Big[\| h_{0a}\|_{L^\infty}\| h^{0a}\|_{L^2}\|\pa  h^{00}\|_{L^2}\nonumber\\
    &+\| h_{0a}\|_{L^\infty}\|g^{00}\|_{L^2}\|\pa  h^{0a}\|_{L^2}+  \|g_{\alpha\beta}\|_{L^\infty}\| h^{0a}\|_{L^2}\|\pa g^{\alpha\beta}\|_{L^2}\Big]\nonumber\\
    \leq &\>C\ve^2\tau^{-2(q-\delta)}.\label{sec:avgsys.eq:odeavglapseerror1}
\end{align}
For the first integral term on the right hand side of \eqref{sec:avgsys.eq:odeavglapseproof} we decompose into average and oscillating factors. Keeping in mind that all terms which contain a single oscillating factor have zero average, we have 
\begin{align}
    \int_{\TT^3}g_{ab}g^{00}\pa_\tau  h^{ab} =&\> \sum_{a,b=1}^3\Big(\>\int_{\TT^3}g_{ab}g^{00}\pa_\tau h_{\av}^{ab} + \int_{\TT^3}g_{ab}g^{00}\pa_\tau h_{\osc}^{ab}\Big)\nonumber\\
    =&\> \sum_{a,b=1}^3\Big(\pa_\tau h_{\av}^{ab}\int_{\TT^3}g_{ab}g^{00} + \int_{\TT^3}\Big[g_{ab}h_{\osc}^{00} \pa_\tau h_{\osc}^{ab} + (g_{ab})_{\osc} h_{\av}^{00}\pa_\tau h_{\osc}^{ab}\Big]\Big).
\end{align}
Then we estimate
\begin{align}
    \Big|\int_{\TT^3}g_{ab}g^{00}\pa_\tau  h^{ab}\Big| \leq&\> C\sum_{a,b = 1}^3\Big(|\pa_\tau h_{\av}^{ab}|\cdot \|g_{ab}\|_{L^2}\|g^{00}\|_{L^2}\nonumber\\
    &+ \|g_{ab}\|_{L^\infty}\|h_{\osc}^{00}\|_{L^2}\|\pa_\tau h_{\osc}^{ab}\|_{L^2} + |h_{\av}^{00}|\cdot\|(g_{ab})_{\osc}\|_{L^2}\|\pa_\tau h_{\osc}^{ab}\|\Big)\nonumber\\
    \leq&\>C\sum_{a,b=1}^3\Big(|\pa_\tau h_{\av}^{ab}|\cdot\|g_{ab}\|_{L^2}\|g^{00}\|_{L^2}\nonumber\\
    &+\|g_{ab}\|_{H^2}\|\ol{\pa} h^{00}\|_{L^2}\|\pa_\tau h^{ab}\|_{L^2}+\|h^{00}\|_{L^2}\|\ol{\pa} g_{ab}\|_{L^2}\|\pa_\tau h^{ab}\|_{L^2}\Big)\nonumber\\\leq&\> C\sum_{a,b=1}^3|\pa_\tau h_{\av}^{ab}(\tau)| + C\ve^2 \tau^{-2(q-\delta)},\label{sec:avgsys.eq:odeavglapseerror2}
\end{align}
with the first inequality following from Cauchy-Schwartz, the second from a combination of the Sobolev embedding, the Poincar\'e inequality, \eqref{sec:avgsys.eq:avgbound}, and the final inequality from the bootstrap assumption and Lemma~\ref{sec:bootbound.lem:loweredmetests} to bound the lowered metric components. 
The equation \eqref{sec:avgsys.eq:odeavglapseproof} combined with the error bounds \eqref{sec:avgsys.eq:odeavglapseerror1}, \eqref{sec:avgsys.eq:odeavglapseerror2}, yields the ODE \eqref{sec:avgsys.eq:odeavglapse}.

For the equation \eqref{sec:avgsys.eq:odeavgshiftproof} we estimate one error term with Cauchy-Schwartz and the bootstrap assumption:
\begin{equation}
    \Big|\int_{\TT^3}g_{\alpha\beta} h^{0i}\pa_\tau  h^{\alpha\beta}\Big|\lesa \sum_{\alpha,\beta = 0}^3\|g_{\alpha\beta}\|_{L^\infty}\| h^{0i}\|_{L^2}\|\pa h^{\alpha\beta}\|_{L^2} \leq C \ve^2\tau^{-2(q-\delta)},\label{sec:avgsys.eq:odeavgshifterror1}
\end{equation}
while we decompose the other term into averaged and oscillatory factors and estimate in an analogous manner to \eqref{sec:avgsys.eq:odeavglapseerror2}:
\begin{align}
    \Big|\int_{\TT^3}g_{\alpha\beta} g^{ia}\pa_a  h^{\alpha\beta}\Big| &= \Big|\int_{\TT^3}\Big[(g_{\av})_{\alpha\beta}g_{\osc}^{ia}\pa_a  h^{\alpha\beta} + (g_{\osc})_{\alpha\beta}   h_{\av}^{ia}\pa_a  h^{\alpha\beta} + (g_{\osc})_{\alpha\beta} h_{\osc}^{ia}\pa_a  h^{\alpha\beta}\Big]\de^3 x\Big|\nonumber\\
    &\leq C\sum_{\alpha,\beta=0}^4\sum_{a=1}^3\Big(|(g_{\av})_{\alpha\beta}|\cdot\|g_{\osc}^{ia}\|_{L^2}\|\ol{\pa}  h^{\alpha\beta}\|_{L^2} + \|(g_{\osc})_{\alpha\beta}\|_{L^2}\|g^{ia}\|_{L^\infty}\|\ol{\pa}  h^{\alpha\beta}\|_{L^2}\Big)\nonumber\\
    &\leq C \sum_{\alpha,\beta=0}^4\sum_{a=1}^3\Big(\|g_{\alpha\beta}\|_{L^2}\|\ol{\pa} g^{ia}\|_{L^2}\|\ol{\pa}  h^{\alpha\beta}\|_{L^2} + \|\ol{\pa} g_{\alpha\beta}\|_{L^2}\|g^{ia}\|_{H^2}\|\ol{\pa}  h^{\alpha\beta}\|_{L^2}\Big)\nonumber\\
    &\leq C \ve^2\tau^{-2(q-\delta)}.\label{sec:avgsys.eq:odeavgshifterror2}
\end{align}
Plugging the estimates \eqref{sec:avgsys.eq:odeavgshifterror1}, \eqref{sec:avgsys.eq:odeavgshifterror2} into \eqref{sec:avgsys.eq:odeavgshiftproof} we obtain the desired inequality \eqref{sec:avgsys.eq:odeavgshift}.
\end{proof}
\subsection{Estimates for the averaged quantities}
Combining Lemmas~\ref{sec:avgsys.lem:odehoriscalar} and~\ref{sec:avgsys.lem:odelapseshift} we arrive at the following system of ordinary differential inequalities:
\begin{subequations}\label{sec:avgsys.eq:odeavg2}
    \begin{gather}
        \Big|\pa_\tau  h_{\av}^{00}(\tau) + \frac{4q+2}{\tau}h_{\av}^{00}(\tau)\Big| \leq\> C\Big(\sum_{a,b=1}^3|\pa_\tau h_{\av}^{ab}(\tau)| + \ve^2 \tau^{-2(q-\delta)}\Big),\label{sec:avgsys.eq:odeavglapse2}\\
        \Big|\pa_\tau   h_{\av}^{0i}(\tau) + \frac{2q+2}{\tau}  h_{\av}^{0i}(\tau)\Big| \leq\> C\ve^2\tau^{-2(q-\delta)},\label{sec:avgsys.eq:odeavgshift2}\\
        \Big|\pa_\tau^2 h_{\av}^{ij}(\tau) + \frac{2q+2}{\tau}\pa_\tau  h_{\av}^{ij}(\tau)\Big| \leq C\Big(\frac{1}{\tau^2}|\psi_{\av}(\tau)| + \ve^2\tau^{-2(q-\delta)}\Big),\label{sec:avgsys.eq:odeavghori2}\\
        \Big|\pa_\tau^2 \psi_{\av}(\tau) + \frac{2q+2}{\tau}\pa_\tau \psi_{\av}(\tau) + \frac{2(2q+1)(q+2)}{\tau^2}\psi_{\av}(\tau)\Big| \leq C\ve^2\tau^{-2(q-\delta)}.\label{sec:avgsys.eq:odeavgscalar2}
    \end{gather}
\end{subequations}
We introduce the norm $S_{\av}$ for the averaged quantities $(h_{\av}^{\mu\nu},\pa_\tau h_{\av}^{\mu\nu},\psi_{\av},\pa_\tau \psi_{\av})$:
\begin{multline}\label{sec:avgsys.eq:avgenergy}
    S_{\av}(\tau):= \tau^{q-1}|h_{\av}^{00}(\tau)| + \tau^{q}|\pa_\tau h_{\av}^{00}(\tau)| + \sum_{a=1}^3\Big\{\tau^{q-1}|h_{\av}^{0a}(\tau)| + \tau^q|\pa_\tau h_{\av}^{0a}(\tau)|\Big\}\\
    +\sum_{a,b=1}^3\Big\{|h_{\av}^{ab}(\tau)|+\tau^{q}|\pa_\tau h_{\av}^{ab}(\tau)|\Big\} + \tau^{q+\delta-1}|\psi_{\av}(\tau)|+\tau^{q+\delta}|\pa_\tau \psi_{\av}(\tau)|.
\end{multline}
This norm satisfies the bound 
\begin{equation}\label{sec:avgsys.eq:avgestrough}
    S_{\av}(\tau) \leq C\ve\tau^\delta
\end{equation}
as an immediate consequence of the bound \eqref{sec:avgsys.eq:avgbound} and the bootstrap assumption. In the next proposition we use the system \eqref{sec:avgsys.eq:odeavg2} to improve this bound.

\begin{proposition}[Estimate for the spatial averages]\label{sec:avgsys.prop:avgest}
    Suppose the bootstrap assumption \eqref{sec:globex.eq:bootstrap} holds on the interval $I = [\tau_0,T]$, and suppose the initial data satisfies the bound \eqref{sec:gaugedeqs.eq:bounddata}. Then the norm $S_{\av}$ satisfies the following estimate for all $\tau \in I$:
    \begin{equation}\label{sec:avgsys.eq:avgest}
        S_{\av}(\tau) \leq C\ve^2.
    \end{equation}
\end{proposition}
\begin{remark}[Reduced decay of averages from nonlinear coupling]\label{sec:avgsys.rmk:bottleneck}
    The weights in $\tau$ in the various terms that make up the norm $S_{\av}$ rescale the leading-order behaviour of the quantities $(h_{\av}^{\mu\nu},\pa_\tau h_{\av}^{\mu\nu},\psi_{\av},\pa_\tau \psi_{\av})$. Consequently, the improved bound \eqref{sec:avgsys.eq:avgest} on $S_{\av}$ implies the corresponding amount of decay of these quantities. These decay rates are weaker than the principle ODE (the left hand side of the system \eqref{sec:avgsys.eq:odeavg2}) predict. For example, the principle terms in the ODE \eqref{sec:avgsys.eq:odeavgshift2} for the $h_{\av}^{0i}$ imply that $h_{\av}^{0i} = O(\tau^{-(q+1)})$, not $O(\tau^{-q})$ as implied by \eqref{sec:avgsys.eq:avgest}. Ultimately the weaker decay is due to the $\ve^2 \tau^{-2(q-\delta)}$ inhomogeneous terms present on the right hand side of the system \eqref{sec:avgsys.eq:odeavg2}, which arise from the nonlinear coupling of the system to the oscillating terms $(h_{\osc}^{\mu\nu},\psi_{\osc})$. The decay of the oscillating terms limit or bottleneck the decay available from the principle system for the averaged quantities.
\end{remark}
\begin{proof}
    Let us define the following rescaled functions
    \begin{subequations}
        \begin{gather}
            \mf{u}^{00}:= \tau^{q-1}h_{\av}^{00},\qquad \mf{u}^{0i}:= \tau^{q}h_{\av}^{0i},\qquad
            \mf{u}^{ij}:= h_{\av}^{ij},\qquad \mf{v}^{ij}:= \tau^q\pa_\tau\mf{u}^{ij},\\ \mf{r}:= \tau^{q+\delta-1}\psi_{\av},\qquad \mf{s}:=\tau\pa_\tau \mf{r}.
        \end{gather}
    \end{subequations}
    We claim that the system \eqref{sec:avgsys.eq:odeavg2} implies the following differential equalities/inequalities:
    \begin{subequations}
        \begin{gather}
            \frac{1}{2}\pa_\tau (\mf{u}^{00})^2 \leq -\frac{3q+3}{\tau}(\mf{u}^{00})^2 + \frac{C}{\tau}\sum_{a,b=1}^3|\mf{u}^{00}||\mf{v}^{ab}| + C\ve^2\tau^{-(q+1-2\delta)}|\mf{u}^{00}|,\label{sec:avgsys.eq:avgestproof1}\\
            \frac{1}{2}\pa_\tau(\mf{u}^{0i})^2 \leq-\frac{q+2}{\tau}(\mf{u}^{0i})^2 + C\ve^2\tau^{-(q-2\delta)}|\mf{u}^{0i}|,\label{sec:avgsys.eq:avgestproof2}\\
            \frac{1}{2}\pa_\tau(\mf{u}^{ij})^2 = \frac{1}{\tau^{q}}\mf{u}^{ij}\mf{v}^{ij},\qquad \frac{1}{2}\pa_\tau \mf{r}^2 =\frac{1}{\tau}\mf{r}\mf{s}\label{sec:avgsys.eq:avgestproof3}\\
            \frac{1}{2}\pa_\tau (\mf{v}^{ij})^2 \leq - \frac{q+2}{\tau}(\mf{v}^{ij})^2 + \frac{C}{\tau^{1+\delta}}|\mf{v}^{ij}||\mf{r}| + C\ve^2 \tau^{-(q-2\delta)}|\mf{v}^{ij}|,\label{sec:avgsys.eq:avgestproof4}\\
            \frac{1}{2}\pa_\tau \mf{s}^2 \leq -\frac{3-2\delta}{\tau}\mf{s}^2 - \frac{c_1}{\tau}\mf{s}\mf{r} + C\ve^2\tau^{-(q-3\delta)}|\mf{s}|,\label{sec:avgsys.eq:avgestproof5}
        \end{gather}
    \end{subequations}
    where $c_1$ is the positive constant
    \begin{equation}
        c_1 = 3q^2+9q+6-3\delta+\delta^2.
    \end{equation}
    The identities in \eqref{sec:avgsys.eq:avgestproof3}, are true by the definitions of the $\mf{v}^{ij}$ and $\mf{s}$ respectively. The inequality \eqref{sec:avgsys.eq:avgestproof1} follows from the computation
    \begin{align}
        \frac{1}{2}\pa_\tau (\mf{u}^{00})^2 &= \tau^{q-1}\mf{u}^{00}\big(\pa_\tau h_{\av}^{00} + \frac{q-1}{\tau}h_{\av}^{00}\big)\nonumber\\
        &= - \frac{3q+3}{\tau}(\mf{u}^{00})^2 + \tau^{q-1}\mf{u}^{00}\big(\pa_\tau h_{\av}^{00} + \frac{4q+2}{\tau}h_{\av}^{00}\big)\nonumber\\
        &\leq -\frac{3q+3}{\tau}(\mf{u}^{00})^2 + \frac{C}{\tau}\sum_{a,b=1}^3|\mf{u}^{00}||\pa_\tau h_{\av}^{ab}| + C\ve^2\tau^{-(q+1-2\delta)}|\mf{u}^{00}|,
    \end{align}
    where in the last line we used \eqref{sec:avgsys.eq:odeavglapse2}. The inequality \eqref{sec:avgsys.eq:avgestproof2} follows similarly from \eqref{sec:avgsys.eq:odeavgshift2}. For \eqref{sec:avgsys.eq:avgestproof4}, we have
    \begin{align}
        \frac{1}{2}\pa_\tau (\mf{v}^{ij})^2 =&\>\tau^q\mf{v}^{ij}\big(\pa_\tau^2 h_{\av}^{ij} + \frac{q}{\tau}\pa_\tau h_{\av}^{ij}\big)\nonumber\\
        =&\>\tau^q\mf{v}^{ij}\big(\pa_\tau^2h_{\av}^{ij}+\frac{2q+2}{\tau}\pa_\tau h_{\av}^{ij}\big) - \frac{q+2}{\tau}(\mf{v}^{ij})^2\nonumber\\
        \leq&\> \frac{C}{\tau}|\mf{v}^{ij}||\mf{r}| + C\ve^2 \tau^{-(q-2\delta)}|\mf{v}^{ij}|- \frac{q+2}{\tau}(\mf{v}^{ij})^2,
    \end{align}
    where in the last line we used \eqref{sec:avgsys.eq:odeavghori2}. Finally for \eqref{sec:avgsys.eq:avgestproof5}, we compute that
    \begin{align}
        \frac{1}{2}\pa_\tau \mf{s}^2 = &\>\mf{s}\pa_\tau (\tau^{q+\delta} \pa_\tau \psi_{\av} + (q+\delta-1)\tau^{q+\delta-1}\psi_{\av})\nonumber\\
        =&\>\tau^{q+\delta} \mf{s}\Big(\pa_\tau^2\psi_{\av} + \frac{2q+2\delta-1}{\tau}\pa_\tau \psi_{\av} + \frac{(q+\delta-1)^2}{\tau^2}\psi_{\av}\Big)\nonumber\\
        =&\>\tau^{q+\delta}\mf{s}\Big(\pa_\tau^2\psi_{\av} + \frac{2q+2}{\tau}\pa_\tau \psi_{\av} + \frac{2(2q+1)(q+2)}{\tau^2}\psi_{\av}\Big)\nonumber\\
        &-(3-2\delta)\tau^{q+\delta-1}\mf{s}\pa_\tau \psi_{\av} - \big(3q^2+(12-2\delta)q+(3+2\delta-\delta^2)\big)\tau^{q+\delta-2}\mf{s}\psi_{\av}.
    \end{align}
    We estimate by \eqref{sec:avgsys.eq:odeavgscalar2}:
    \begin{equation}
        \tau^{q+\delta}\mf{s}\Big(\pa_\tau^2\psi_{\av} + \frac{2q+2}{\tau}\pa_\tau \psi_{\av} + \frac{2(2q+1)(q+2)}{\tau^2}\psi_{\av}\Big) \leq C\ve^2\tau^{-(q-3\delta)}|\mf{s}|,
    \end{equation}
    moreover
    \begin{multline}
        -(3-2\delta)\tau^{q+\delta-1}\mf{s}\pa_\tau \psi_{\av} - \Big[3q^2+(12-2\delta)q+(3+2\delta-\delta^2)\Big]\tau^{q+\delta-2}\mf{s}\psi_{\av}\\
        = - \frac{3-2\delta}{\tau}\mf{s}^2 - \frac{3q^2+9q+6-3\delta+\delta^2}{\tau^2}\mf{s}\mf{r}.
    \end{multline}
    The constant $c_1$ is clearly positive for $\delta$ sufficiently small. Combining the above computations, \eqref{sec:avgsys.eq:avgestproof5} follows.
    
    Next we derive estimates for the system $(\mf{u}^{\mu\nu}, \mf{v}^{ij}, \mf{r},\mf{s})$ using the system \eqref{sec:avgsys.eq:avgestproof1} - \eqref{sec:avgsys.eq:avgestproof5}. The estimate
    \begin{equation}\label{sec:avgsys.eq:avgestproof7}
     (\mf{u}^{0i}(\tau))^2 \leq C\big((\mf{u}^{0i}(\tau_0))^2 + \ve^4\tau^{-\delta}\big)
    \end{equation}
    follows from \eqref{sec:avgsys.eq:avgestproof2} by a straightforward application of Cauchy-Schwartz and Gr\"onwall's inequality. For the remaining quantities, we estimate the $\mf{r}$, $\mf{s}$ (which are derived from $\psi_{\av}$) first, since these equations do not contain mixed linear terms, while the equations for the $h_{\av}^{00}$, $h_{\av}^{ij}$ quantities contain linear terms that depend on $\psi_{\av}$. We compute from \eqref{sec:avgsys.eq:avgestproof3} and \eqref{sec:avgsys.eq:avgestproof5} the inequality
    \begin{equation}
        \pa_\tau\big(\frac{c_1}{2}\mf{r}^2 + \frac{1}{2}\mf{s}^2\big) \leq \frac{c_1}{\tau}\mf{r}\mf{s} -\frac{3-2\delta}{\tau}\mf{s}^2 - \frac{c_1}{\tau}\mf{r}\mf{s} + C\ve^2 \tau^{-q + 3\delta}|\mf{s}| \leq C\ve^2\tau^{-q+3\delta}|\mf{s}|.
    \end{equation}
    Note the cancellation of the mixed $c_1\tau^{-1}\mf{r}\mf{s}$ terms. This implies by a Gr\"onwall-type inequality that
    \begin{equation}\label{sec:avgsys.eq:avgestproof8}
        \frac{c_1}{2}\mf{r}^2(\tau) + \frac{1}{2}\mf{s}^2(\tau) \leq C(\mf{r}^2(\tau_0) + \mf{s}^2(\tau_0) + \ve^4 \tau^{-\delta}),
    \end{equation}
    as long as $\delta$ is chosen sufficiently small that $q \geq 1+4\delta$. We must estimate the $\mf{u}^{00}$ and $\mf{v}^{ij}$ jointly, due to the mixed $|\mf{u}^{00}||\mf{v}^{ab}|$ term present in \eqref{sec:avgsys.eq:avgestproof1}. Letting $c_2$ denote a positive constant to be set in a moment, we compute
    \begin{align}
        \pa_\tau\Big[\frac{1}{2}(\mf{u}^{00})^2 + \frac{c_2}{2} \sum_{a,b=1}^3(\mf{v}^{ab})^2\Big] \leq &\> -\frac{3q+3}{\tau}(\mf{u}^{00})^2 + \frac{C}{\tau}\sum_{a,b=1}^3 |\mf{u}^{00}||\mf{v}^{ab}| - \frac{c_2(q+2)}{\tau}\sum_{a,b=1}^3(\mf{v}^{ab})^2\nonumber\\
        &+ C\ve^2\tau^{-(q+1-2\delta)}|\mf{u}^{00}| + C\sum_{a,b=1}^3\Big(\frac{1}{\tau^{1+\delta}}|\mf{r}||\mf{v}^{ab}| + \ve^2\tau^{-(q-2\delta)}|\mf{v}^{ab}|\Big).
    \end{align}
    We fix $c_2$ to be sufficiently large that the quadratic terms on the RHS of the above equation are non-positive, i.e. we have
    \begin{equation}
        -\frac{3q+3}{\tau}(\mf{u}^{00})^2 + \frac{C}{\tau}\sum_{a,b=1}^3 |\mf{u}^{00}||\mf{v}^{ab}| - \frac{c_2(q+2)}{\tau}\sum_{a,b=1}^3(\mf{v}^{ab})^2 \leq 0.
    \end{equation}
    Then we apply the inequality \eqref{sec:avgsys.eq:avgestproof8} to estimate $|\mf{r}|$ and apply Cauchy-Schwartz to obtain
    \begin{equation}
        \pa_\tau\Big(\frac{1}{2}(\mf{u}^{00})^2 + \frac{c_2}{2} \sum_{a,b=1}^3(\mf{v}^{ab})^2\Big) \leq \frac{C}{\tau^{1+\delta}}\Big((\mf{u}^{00})^2+\sum_{a,b=1}^3(\mf{v}^{ab})^2\Big) + \frac{C}{\tau^{1+\delta}}(\mf{r}^2(\tau_0)+\mf{s}^2(\tau_0)+\ve^4),
    \end{equation}
    from which we obtain
    \begin{equation}\label{sec:avgsys.eq:avgestproof9}
        \frac{1}{2}(\mf{u}^{00}(\tau))^2 + \frac{c_2}{2} \sum_{a,b=1}^3(\mf{v}^{ab}(\tau))^2 \leq C\Big((\mf{u}^{00}(\tau_0))^2 + \sum_{a,b=1}^3(\mf{v}^{ab}(\tau_0))^2 + \mf{r}^2(\tau_0) +\mf{s}^2(\tau_0) + \ve^4\tau^{-\delta}\Big).
    \end{equation}
    by Gr\"onwall. We then derive the bound
    \begin{equation}\label{sec:avgsys.eq:avgestproof10}
        (\mf{u}^{ij})^2 \leq C\Big((\mf{u}^{00}(\tau_0))^2 + \sum_{a,b=1}^3\big[(\mf{u}^{ab})^2+(\mf{v}^{ab}(\tau_0))^2\big] + \mf{r}^2(\tau_0) +\mf{s}^2(\tau_0) + \ve^4\tau^{-\delta}\Big)
    \end{equation}
    from \eqref{sec:avgsys.eq:avgestproof3} through a straightforward application of Cauchy-Schwartz and Gr\"onwall. We convert the above estimates for $(\mf{u}^{\mu\nu},\mf{v}^{\mu\nu},\mf{r},\mf{s})$ into a bound on $S_{\av}(\tau)$ using
    \begin{subequations}
        \begin{gather}
            \tau^{q-1-\delta}| h_{\av}^{00}| \sim \tau^{-\delta}|\mf{u}^{00}|,\\
            \tau^{q-\delta}|h_{\av}^{0i}| \sim |\mf{u}^{0i}|,\\
            |h_{\av}^{ij}| + \tau^q|\pa_\tau h_{\av}^{ij}| \sim |\mf{u}^{ij}| + |\mf{v}^{ij}|,\\
            \tau^{q+\delta-1}|\psi_{\av}| + \tau^{q+\delta}|\pa_\tau \psi_{\av}|  \sim  |\mf{r}| + |\mf{s}|,
        \end{gather}
    \end{subequations}
    and applying the estimates \eqref{sec:avgsys.eq:avgestproof7} - \eqref{sec:avgsys.eq:avgestproof10}. The quantities $\tau^q |\pa_\tau h_{\av}^{00}|$ and $\tau^{q+1}|\pa_\tau h_{\av}^{0i}|$ can be estimated by taking the inequalities \eqref{sec:avgsys.eq:odeavglapse2}, \eqref{sec:avgsys.eq:odeavgshift2} and applying the estimates \eqref{sec:avgsys.eq:avgestproof7} - \eqref{sec:avgsys.eq:avgestproof10} to the right hand side. It follows that
    \begin{equation}
        S_{\av}(\tau) \leq C(S_{\av}(\tau_0) + \ve^2\tau^{-\delta/2}).
    \end{equation}
    The average of a function satisfies the bound \eqref{sec:avgsys.eq:avgbound} and hence the initial data bound \eqref{sec:gaugedeqs.eq:bounddata} implies
    \begin{equation}
        S_{\av}(\tau_0) \leq C\sum_{\alpha,\beta=0}^3\Big(\|h^{\alpha\beta}(\tau_0)\|_{L^2} + \|\pa_\tau h^{\alpha\beta}(\tau_0)\|_{L^2}\Big) + C\big(\|\psi(\tau_0)\|_{L^2} + \|\pa_\tau \psi(\tau_0)\|_{L^2}\big) \leq C\ve^2.
    \end{equation}
    This completes the proof of Proposition~\ref{sec:avgsys.prop:avgest}.
\end{proof}
\section{Energy estimates}\label{sec:energyests}
In this section we introduce a family of weighted energies, then we derive differential inequalities for them, provided the bootstrap assumption \eqref{sec:globex.eq:bootstrap} holds. We derive higher-order versions of these energies, and prove analogous differential inequalities for those. These weighted energies will be used to derive estimates for the oscillatory functions $(h_{\osc}^{\mu\nu},\psi_{\osc})$ in the next section. For this section however, we will derive inequalities that hold for any (sufficiently regular) scalar function $u$ with zero spatial average.

We introduce the parameters $\alpha, \beta \in \RR_{\geq 0}$ and define the following energy: 
\begin{equation}\label{sec:energyests.eq:energydef}
    E_{\alpha,\beta}[u](\tau) := \frac{1}{2\tau^{2\beta}}\int_{\TT^3} \Big[-g^{00}(\pa_\tau( \tau^{\alpha} u))^2 + \tau^{2\alpha}g^{ab}\pa_a u \pa_b u\Big] \de^3 x.
\end{equation}
In the following lemma we demonstrate that the family of energies $E_{\alpha,\beta}$ are coercive for functions $u$ with zero spatial average.
\begin{lemma}\label{sec:energyests.lem:energycoercivity}
    Suppose $u$ is a function with zero spatial average, so that $u = u_{\osc}$. Then $E_{\alpha,\beta}[u]$ satisfies the following inequalities:
    \begin{equation}\label{sec:energyests.eq:energycoercivity}
        \frac{1}{C}\tau^{2(\alpha-\beta)}(\|u\|_{L^2}^2 + \|\pa u\|_{L^2}^2) \leq E_{\alpha,\beta}[u] \leq C\tau^{2(\alpha-\beta)}(\|u\|_{L^2}^2 + \|\pa u\|_{L^2}^2).
    \end{equation}
\end{lemma}
\begin{proof}
    This is a straightforward application of the bootstrap assumption and the Poincar\'e inequality. We will prove the first inequality, the second follows by a similar argument. The bootstrap assumption implies
    \begin{equation}
        1 \leq -Cg^{00},\qquad \delta^{ij}\leq Cg^{ij},\quad i,j=1,2,3,
    \end{equation}
    and therefore
    \begin{equation}\label{sec:energyests.eq:energycoercivityproof1}
        \|\pa_\tau (\tau^\alpha u)\|_{L^2}^2 + \tau^{2\alpha}\|\ol{\pa} u\|_{L^2}^2 \leq C\int_{\TT^3}\Big[-g^{00}(\pa_\tau (\tau^\alpha u))^2 + \tau^{2\alpha} g^{ij}\pa_i u \pa_j u].
    \end{equation}
    As $u$ has zero spatial average, we apply the Poincar\'e inequality to bound
    \begin{align}
        \tau^{2(\alpha-\beta)}(\|u\|_{L^2}^2 + \|\pa u\|_{L^2}^2) &\leq \frac{C}{\tau^{2\beta}}\big(\tau^{2\alpha}\|u\|_{L^2}^2 + \|\pa_\tau (\tau^\alpha u)\|_{L^2}^2+ \tau^{2\alpha}\|\ol{\pa} u\|_{L^2}^2\big)\nonumber\\
        &\leq \frac{C}{\tau^{2\beta}} \big(\|\pa_\tau (\tau^\alpha u)\|_{L^2}^2 + \tau^{2\alpha}\|\ol{\pa} u\|_{L^2}^2\big).\label{sec:energyests.eq:energycoercivityproof2}
    \end{align}
    Combining \eqref{sec:energyests.eq:energycoercivityproof1} and \eqref{sec:energyests.eq:energycoercivityproof2} yields \eqref{sec:energyests.eq:energycoercivity}.
\end{proof}
\subsection{Energy differential inequality}
We now derive the basic energy inequality for the family of energies $E_{\alpha,\beta}$.
\begin{lemma}[Energy inequality]\label{sec:energyests.lem:energyineq}
    Suppose $u$ is a function with zero spatial average, and let $\beta, \alpha \geq 0$. The energy $E_{\alpha,\beta}[u]$ satisfies the differential inequality
    \begin{subequations}
        \begin{equation}\label{sec:energyests.eq:energyineq}
            \pa_\tau E_{\alpha,\beta}[u](\tau) \leq -\frac{2\beta}{\tau}E_{\alpha,\beta}[u](\tau) +C\Big(\frac{1}{\tau^{1+\delta}}E_{\alpha,\beta}[u](\tau)+ \tau^{2(\alpha - \beta) +1}\|P_\alpha u\|_{L^2}^2\Big),
        \end{equation}
    \end{subequations}
    where $P_\alpha$ is the wave operator
    \begin{equation}
    P_\alpha := \wh{\square}_g + \frac{2\alpha}{\tau}g^{00}\pa_\tau.
\end{equation}
\end{lemma}
\begin{remark}[Justification of weights in the energy $E_{\alpha,\beta}$]
    The introduction of the weights in time for the energy $E_{\alpha,\beta}[u]$ rescale the leading-order behaviour of the function $u$, so that (using Lemma~\ref{sec:energyests.lem:energycoercivity}) one can prove decay of $u$. The first parameter $\alpha$ introduces the principal damping term $2\alpha\tau^{-1}g^{00}\pa_\tau$ (contained in $P_{\alpha}$) to the differential inequality for $E_{\alpha,\beta}$. Such damping terms are present in the gauge-fixed system \eqref{sec:gaugedeqs.eq:wave} (and therefore in the equations for the oscillatory remainders $(h_{\osc}^{\mu\nu},\psi_{\osc})$). More specifically, $\alpha = q$ for the damping term in the equation \eqref{sec:gaugedeqs.eq:wavelapse} for $h^{00}$, while $\alpha = q+1$ for the damping term in all the other equations. We can fix $\alpha$ in the energies for these quantities to exactly cancel with the damping terms in the gauge-fixed system. The actual decay of the quantities $(h^{\mu\nu},\psi)$ are less than what the damping term $2\alpha\tau^{-1}g^{00}\pa_\tau$ predicts, however. This is due to the linear and nonlinear coupling of the equations for the $h^{00}$, $h^{0i}$, $h^{ij}$, $\psi$. The weight $\tau^{-2\beta}$ in $E_{\alpha,\beta}$ are introduced to compensate for this reduced decay. See \eqref{sec:oscsys.eq:energyhochoice} for the specific choice of energies we make for the quantities $(h_{\osc}^{\mu\nu},\psi_{\osc})$.
\end{remark}
\begin{proof}
    Let $v = \tau^{\alpha}u$, so that
    \begin{equation}
        E_{\alpha,\beta}[u] = \frac{1}{2\tau^{2\beta}}\int_{\TT^3}\Big[-g^{00}(\pa_\tau v)^2 + g^{ij}\pa_i v\pa_j v\Big].
    \end{equation}
    We differentiate $E_{\alpha,\beta}[u]$ in time:
    \begin{multline}
        \pa_\tau E_{\alpha,\beta}[u] = -\frac{2\beta}{\tau}E_{\alpha,\beta} +\frac{1}{\tau^{2\beta}}\int_{\TT^3}\Big[-g^{00}\pa_\tau v\pa_\tau^2 v + g^{ab}\pa_a v \pa_b \pa_\tau v\Big]\\
        - \frac{1}{2\tau^{2\beta}}\int_{\TT^3}\Big[(\pa_\tau h^{00})(\pa_\tau v)^2 - (\pa_\tau h^{ab})\pa_a v \pa_b v \Big].
    \end{multline}
    The product rule implies
    \begin{subequations}
        \begin{align}
            \pa_b (g^{ab}\pa_a v \pa_\tau v) =&\> \pa_b h^{ab} \pa_a v\pa_\tau v + g^{ab} \pa_a \pa_b v \pa_\tau v + g^{ab}\pa_a v \pa_\tau\pa_b v,\\
            \pa_a (g^{0a}(\pa_\tau v)^2) =&\> \pa_b g^{0a} (\pa_\tau v)^2 + 2g^{0a} \pa_\tau \pa_a v \pa_\tau v,
        \end{align}
    \end{subequations}
    and so we integrate by parts and compute
    \begin{align}
        \int_{\TT^3}\Big[ -\pa_\tau vg^{00}\pa_\tau^2 v + g^{ab}\pa_a v \pa_b \pa_\tau v\Big] =& -\int_{\TT^3}(\pa_\tau v)\Big[g^{00}\pa_\tau^2 v + 2g^{0a}\pa_\tau\pa_a v + g^{ab} \pa_a\pa_b v\Big]\de^3 x\nonumber\\
        &-\int_{\TT^3}\Big[\frac{1}{2}\pa_a h^{0a} (\pa_\tau v)^2 + \pa_b h^{ab} \pa_a v\pa_\tau v\Big]\de^3\nonumber\\
        =& -\int_{\TT^3}(\pa_\tau v)(\wh{\square}_g v)\de^3 x-\int_{\TT^3}\Big[\frac{1}{2}\pa_a h^{0a} (\pa_\tau v)^2 + \pa_b h^{ab} \pa_a v\pa_\tau v \Big]\de^3 x.
    \end{align}
    Thus we obtain the following differential equation for $E_{\alpha,\beta}[u]$:
    \begin{multline}\label{sec:energyests.eq:energyineqproof1}
        \pa_\tau E_{\alpha,\beta}[u] = -\frac{2\beta}{\tau}E_{\alpha,\beta}[u] -\frac{1}{\tau^{2\beta}}\int_{\TT^3}(\pa_\tau v)(\wh{\square}_g v)
        \\-\frac{1}{2\tau^{2\beta}}\int_{\TT^3}\Big[(\pa_\tau h^{00})(\pa_\tau v)^2 - (\pa_\tau h^{ab})\pa_a v \pa_b v
        +(\pa_a h^{0a}) (\pa_\tau v)^2 + 2(\pa_b h^{ab}) \pa_a v\pa_\tau v\Big].
    \end{multline}
    The final integral in \eqref{sec:energyests.eq:energyineqproof1} is an error term. We estimate
    \begin{multline}\label{sec:energyests.eq:energyineqproof2}
        \frac{1}{\tau^{2\beta}}\Big|\int_{\TT^3}\Big[(\pa_\tau h^{00})(\pa_\tau v)^2 - (\pa_\tau h^{ab})\pa_a v \pa_b v +(\pa_a h^{0a}) (\pa_\tau v)^2 + 2(\pa_b h^{ab}) \pa_a v\pa_\tau v\Big]\de^3 x\Big|\leq\\
        \leq\frac{C}{\tau^{2\beta}}\sum_{\mu,\nu=0}^3\|\pa h^{\mu\nu}\|_{L^\infty}\Big(\|\pa v\|_{L^2}^2 + \|v\|_{L^2}^2\Big) \leq C\sum_{\mu,\nu=0}^3 \|\pa h^{\mu\nu}\|_{H^2} E_{\alpha,\beta}[u] \leq C\ve \tau^{-(q-\delta)}E_{\alpha,\beta}[u],
    \end{multline}
    where the first inequality is due to Cauchy-Schwartz, the second uses the Sobolev embedding and Lemma~\ref{sec:energyests.lem:energycoercivity}, and the last inequality follows from the bootstrap assumption.
    
    Next we estimate the integral of $(\pa_\tau v)(\wh{\square}_g v)$ in \eqref{sec:energyests.eq:energyineqproof1}. We calculate
    \begin{align}
        \wh{\square}_g v  =&\> \tau^\alpha \Big(\wh{\square}_g u + \frac{2\alpha}{\tau}g^{00}\pa_\tau u + \frac{2\alpha}{\tau}h^{0a}\pa_a u + \frac{\alpha(\alpha-1)}{\tau^2}g^{00} u\Big)\nonumber\\
        =&\> \tau^\alpha P_\alpha u+\frac{2\alpha}{\tau}h^{0a}\pa_a v + \frac{\alpha(\alpha-1)}{\tau^2}g^{00}v,
    \end{align}
    so that
    \begin{equation}\label{sec:energyests.eq:energyineqproof3}
        -\frac{1}{\tau^{2\beta}}\int_{\TT^3}(\pa_\tau v)(\square_g v)
        = -\tau^{\alpha-2\beta}\int_{\TT^3}(\pa_\tau v)(P_\alpha u)\\
         - \frac{2\alpha}{\tau^{2\beta+1}}\int_{\TT^3}h^{0a}\pa_a v\pa_\tau v + \frac{\alpha(\alpha-1)}{\tau^{2\beta+2}}\int_{\TT^3}g^{00}v\pa_\tau v.
    \end{equation}
    We estimate the three integrals on the RHS of the above equation using a combination of Cauchy-Schwartz, Lemma~\ref{sec:energyests.lem:energycoercivity}, and the bootstrap assumption:
    \begin{subequations}
    \begin{gather}
            \tau^{\alpha-2\beta}\Big|\int_{\TT^3}(\pa_\tau v)(P_\alpha u)\Big| \leq C\tau^{\alpha-2\beta}\|\pa_\tau v\|_{L^2}\|P_\beta u\|_{L^2}\leq \frac{\beta}{\tau}E_{\alpha,\beta}[u](\tau) + C\tau^{2(\alpha-\beta)+1}\|P_\alpha u\|_{L^2}^2,\\
            \tau^{-(1+2\beta)} \Big|\int_{\TT^3}h^{0a}\pa_\tau v\pa_a v\>\de^3 x\Big| \leq C\tau^{-(1+2\beta)} \sum_{a=1}^3\|h^{0a}\|_{L^\infty}\|\pa_\tau v\|_{L^2}\|\pa_a v\|_{L^2}\leq C\tau^{-(q+1-\delta)}E_{\alpha,\beta}[u],\\
            \tau^{-(2+2\beta)}\Big|\int_{\TT^3}g^{00}v\pa_\tau v\Big| \leq C\tau^{-(2+2\beta)}\|g^{00}\|_{L^\infty}\|v\|_{L^2}\|\pa_\tau v\|_{L^2} \leq \frac{C}{\tau^2}E_{\alpha,\beta}[u].\label{sec:energyestimates.sec:energyests.eq:energyineqproof4}
        \end{gather}
    \end{subequations}
    Plugging these bounds into \eqref{sec:energyests.eq:energyineqproof3}, then combining with \eqref{sec:energyests.eq:energyineqproof1}, \eqref{sec:energyests.eq:energyineqproof2} implies the final estimate \eqref{sec:energyests.eq:energyineq}.
\end{proof}
\subsection{Higher order energy}
Next we derive a higher-order version of the energy inequality from Lemma~\ref{sec:energyests.lem:energyineq} which controls all higher derivatives in the bootstrap norm $S_K(\tau)$. We introduce the energies
\begin{equation}\label{sec:energyests.eq:energyhodef}
    E_{\alpha,\beta}^{(K)}[u](\tau):= \sum_{|I|\leq K}E_{\alpha,\beta}[\ol{\pa}{}^I u](\tau).
\end{equation}
As an immediate consequence of Lemma~\ref{sec:energyests.lem:energycoercivity}, this energy obeys the inequalities
\begin{equation}\label{sec:energyests.eq:energyhocoercivity}
    \frac{1}{C}\tau^{2(\alpha-\beta)} (\|u\|_{H^K}^2 + \|\pa u\|_{H^K}^2) \leq E_{\alpha,\beta}^{(K)}[u](\tau) \leq C\tau^{2(\alpha-\beta)} (\|u\|_{H^K}^2 + \|\pa u\|_{H^K}^2).
\end{equation}
In order to derive an inequality for the higher-order energy $E_{\alpha,\beta}^{(K)}[u]$ we first derive an estimate for the commutator $[P_\alpha,\ol{\pa}{}^I]$.
\begin{lemma}[Wave operator commutator estimate]\label{sec:energyests.lem:wavecom}
    The wave operators $P_\alpha$ obeys the estimate
    \begin{equation}\label{sec:energyests.eq:wavecom}
        \sum_{|I|\leq K}\|[P_\alpha, \ol{\pa}{}^I]u\|_{L^2} \leq C\tau^{-(q-\delta)}\Big(\|P_\alpha u\|_{H^{K-1}} + \|\pa u\|_{H^K} \Big).
    \end{equation}
\end{lemma}
\begin{proof}
    The operator $P_\alpha$ satisfies the commutator identity
    \begin{equation}
    [P_\alpha,\ol{\pa}{}^I]u = [\wh{\square}_g,\ol{\pa}{}^I]u + \frac{2\alpha}{\tau}[g^{00}\pa_\tau,\ol{\pa}{}^I]u =[\wh{\square}_g,\ol{\pa}{}^I]u + \frac{2\alpha}{\tau}[g^{00},\ol{\pa}{}^I]\pa_\tau u.
    \end{equation}
    We estimate the $[\wh{\square}_g,\ol{\pa}{}^I]u$ term first.
    For each multi-index $I$: $|I| \leq K$, the commutator $[\wh{\square}_g,\ol{\pa}{}^I]u$ satisfies the identity
    \begin{equation}
        [\wh{\square}_g,\ol{\pa}{}^I]u = \sum_{0 < J \leq I}\sum_{\alpha,\beta=0}^3 c_{I,J}(\ol{\pa}{}^J  h^{\alpha\beta})(\pa_\alpha\pa_\beta \ol{\pa}{}^{I-J}u),
    \end{equation}
    for some constants $c_{I,J}$. Observe that all terms on the right hand side of the above equation contain at least one spatial derivative of the metric components $g^{\mu\nu}$, and at most $|I|+1$ spacetime derivatives of $u$, at most $2$ of which are purely time derivatives. We sum over all $|I| \leq K$ and estimate the factors with fewer than $\lfloor K/2\rfloor$ derivatives in $L^\infty$:
    \begin{align}
        \sum_{|I|\leq K}\Big(\sum_{0 < J \leq I}\sum_{\alpha,\beta=0}^3 \|(\ol{\pa}{}^J  h^{\alpha\beta})(\pa_\alpha\pa_\beta& \ol{\pa}{}^{I-J}u)\|_{L^2}\Big)\nonumber\\\leq &\> \sum_{\alpha,\beta = 0}^3\Big\{\sum_{\lfloor K/2\rfloor +1 \leq |I| \leq K} \Big(\sum_{|J| \leq \lfloor K/2\rfloor - 1 }\|\ol{\pa}{}^I  h^{\alpha\beta}\|_{L^2} \|\pa_\alpha\pa_\beta \ol{\pa}{}^{J}u\|_{L^\infty}\Big)\nonumber\\
        &+\sum_{1 \leq |I| \leq \lfloor K/2\rfloor +1}\Big(\sum_{\lfloor K/2 \rfloor-1 \leq |J| \leq K-1}\|\ol{\pa}{}^I  h^{\alpha\beta}\|_{L^\infty} \|\pa_\alpha\pa_\beta \ol{\pa}{}^{J}u\|_{L^2}\Big)\Big\}.
    \end{align}
    Since $K \geq 3$, we have $\lfloor K/2\rfloor \leq K-2$, and so the Sobolev embedding implies the inequalities
    \begin{subequations}
        \begin{align} 
            \sum_{\alpha,\beta = 0}^3\sum_{|J| \leq \lfloor K/2\rfloor-1 } \|\pa_\alpha\pa_\beta \ol{\pa}{}^{J}u\|_{L^\infty} \leq&\> C\sum_{\alpha,\beta = 0}^3\sum_{|J| \leq K-1}\|\pa_\alpha\pa_\beta \ol{\pa}{}^{J}u\|_{L^2},\\
            \sum_{\alpha,\beta = 0}^3\sum_{1 \leq |I| \leq \lfloor K/2\rfloor+1}\|\ol{\pa}{}^I  h^{\alpha\beta}\|_{L^\infty}  \leq&\> C\sum_{\alpha,\beta = 0}^3 \|\ol{\pa} h^{\alpha\beta}\|_{H^K},
        \end{align}
    \end{subequations}
    therefore by the bootstrap assumption we obtain the estimate
    \begin{equation}\label{sec:energyests.eq:wavecomproof1}
        \sum_{|I|\leq K}\|[\wh{\square}_g,\ol{\pa}{}^I]u\|_{L^2} \leq C\sum_{\alpha,\beta=0}^3\sum_{|I|\leq K-1}\|\ol{\pa}h^{\alpha\beta}\|_{H^K}\|\pa_\alpha\pa_\beta \ol{\pa}{}^Iu\|_{L^2} \leq C\ve\tau^{-(q-\delta)}\sum_{\alpha,\beta=0}^3\sum_{|I|\leq K-1}\|\pa_\alpha\pa_\beta \ol{\pa}{}^Iu\|_{L^2}.
    \end{equation}
    By a similar argument, for the commutator $\frac{2\alpha}{\tau}[g^{00},\ol{\pa}{}^I]\pa_\tau u$ we have
    \begin{equation}\label{sec:energyests.eq:wavecomproof2}
        \frac{2\alpha}{\tau}\sum_{|I|\leq K}\|[g^{00},\ol{\pa}{}^I]\pa_\tau u\|_{L^2} \leq \frac{C}{\tau}\|\ol{\pa}h^{00}\|_{H^K}\|\pa_\tau u\|_{H^K}\leq C\ve\tau^{-(q+1-\delta)}\|\pa_\tau u\|_{H^K}.
    \end{equation}
    Combining the estimates \eqref{sec:energyests.eq:wavecomproof1} and \eqref{sec:energyests.eq:wavecomproof2} yields the bound
    \begin{align}
        \sum_{|I|\leq K}\big\|[P_\alpha,\ol{\pa}{}^I]u\|_{L^2} \leq&\>C\ve\tau^{-(q-\delta)}\sum_{\alpha,\beta=0}^3 \sum_{|I|\leq K-1}\|\pa_\alpha\pa_\beta\ol{\pa}^I u\|_{L^2} +C\ve\tau^{-(q+1-\delta)}\|\pa_\tau u\|_{H^K}\nonumber\\
        \leq &\>C\ve\tau^{-(q-\delta)} \Big(\sum_{|I|\leq K-1}\|\pa_\tau^2 \ol{\pa}{}^I u\|_{L^2} +\|\pa u\|_{H^K}\Big).\label{sec:energyests.eq:wavecomproof3}
    \end{align}
    It remains to estimate the $\pa_\tau^2 \ol{\pa}{}^I u$ terms. We compute
    \begin{equation}
        \pa_\tau^2 \ol{\pa}{}^Iu = (g^{00})^{-1}\Big(P_\alpha \ol{\pa}{}^Iu - 2 h^{0a}\pa_a\tau \pa_\tau \ol{\pa}{}^Iu - g^{ab}\pa_a\pa_b \ol{\pa}{}^Iu - \frac{2\alpha}{\tau}g^{00}\pa_\tau \ol{\pa}{}^I u \Big),
    \end{equation}
    and estimate via Cauchy-Schwartz, the Sobolev inequality, and the bootstrap assumption:
    \begin{align}
        \sum_{\alpha,\beta =0}^3 \sum_{|I|\leq K-1}\|\pa_\tau^2 \ol{\pa}{}^I u\|_{L^2} \leq C&\sum_{|I|\leq K-1}\Big\{\|P_\alpha \ol{\pa}{}^I u\|_{L^2} + \|g^{0a}\pa_a \ol{\pa}{}^I \pa_\tau u\|_{L^2}\nonumber\\
        &+\|g^{ab}\pa_a\pa_b\ol{\pa}{}^I u\|_{L^2} +\frac{1}{\tau}\|g^{00}\pa_\tau \ol{\pa}{}^I u\|_{L^2} \Big\}\nonumber\\
        \leq C&\Big(\sum_{|I|\leq K-1}\|P_\alpha \ol{\pa}{}^I u\|_{L^2} + \|\pa u\|_{H^K}\Big).
    \end{align}
    Finally, we bound
    \begin{equation}
        \sum_{|J|\leq K-1}\|P_\alpha \ol{\pa}{}^J u\|_{L^2} \leq \|P_\alpha u\|_{H^{K-1}} + \sum_{|J|\leq K-1}\|[P_\alpha,\ol{\pa}{}^J]u\|_{L^2},
    \end{equation}
    which implies
    \begin{equation}
        \sum_{|I|\leq K}\big\|[P_\alpha ,\ol{\pa}{}^I]u\|_{L^2}\leq \>C\ve\tau^{-(q-\delta)} \Big(\|P_\alpha u\|_{H^{K-1}}+\|\pa u\|_{H^K}\Big)\\+ C\ve \tau^{-(q-\delta)}\sum_{|J|\leq K-1}\|[P_\alpha,\ol{\pa}^J]u\|_{L^2}. 
    \end{equation}
    The last term can be absorbed into the left hand side, hence we obtain \eqref{sec:energyests.eq:wavecom}.
\end{proof}

We now state the differential inequalities satisfied by the higher order energy $E_{\alpha,\beta}^{(K)}$ defined in \eqref{sec:energyests.eq:energyhodef}.
\begin{proposition}[Higher-order energy inequality]\label{sec:energyests.prop:energyhoineq}
    Let $u$ be a sufficiently regular function with zero spatial average. The energy $E_{\alpha,\beta}^{(K)}$ satisfies the following differential inequality:
    \begin{equation}\label{sec:energyests.eq:energyhoineq}
        \pa_\tau E_{\alpha,\beta}^{(K)}[u](\tau)  \leq -\frac{\beta}{\tau} E_{\alpha,\beta}^{(K)}[u](\tau) + C\Big(\frac{1}{\tau^{1+\delta}}E_{\alpha,\beta}^{(K)}[u](\tau) + \tau^{2(\alpha-\beta)+1}\|P_\alpha u\|_{H^K}^2\Big).
    \end{equation}
\end{proposition}
\begin{proof}
    Lemma~\ref{sec:energyests.lem:energyineq} implies the bound
    \begin{equation}\label{sec:energyests.eq:energyhoineqproof1}
        \pa_\tau E_{\alpha,\beta}^{(K)}[u]\leq -\frac{2\beta}{\tau} E_{\alpha,\beta}^{(K)}[u](\tau) + \frac{C}{\tau^{1+\delta}}E_{\alpha,\beta}^{(K)}[u](\tau) + C\tau^{2(\alpha-\beta)+1+\delta}\sum_{|I|\leq K}\|P_\alpha( \ol{\pa}{}^I u)\|_{L^2}^2.
    \end{equation}
    Then we apply the wave commutator Lemma~\ref{sec:energyests.lem:wavecom} to estimate the final term:
    \begin{align}
        C\tau^{2(\alpha-\beta)+1+\delta}\sum_{|I|\leq K}\|P_\alpha( \ol{\pa}{}^I u)\|_{L^2}^2. &\leq C\tau^{2(\alpha-\beta)+1+\delta}\Big(\|P_\alpha u\|_{H^K}^2 + \sum_{|I|\leq K}\|[P_\alpha,\ol{\pa}{}^I]u\|_{L^2}^2\Big)\nonumber\\
        &\leq C\tau^{2(\alpha-\beta)+1+\delta}\big(\|P_\alpha u\|_{H^K}^2 + \ve\tau^{-2(q-\delta)}\|\pa u\|_{H^K}^2\big)\nonumber\\
        &\leq C\tau^{2(\alpha-\beta)+1+\delta}\|P_\alpha u\|_{H^K}^2 + \frac{C}{\tau^{1+\delta}}E_{\alpha,\beta}^K,\label{sec:energyests.eq:energyhoineqproof2}
    \end{align}
    where the last line holds due to \eqref{sec:energyests.eq:energyhocoercivity}. The result now follows.
\end{proof}
\section{Estimates for the oscillatory remainders \texorpdfstring{$h_{\osc}^{\mu\nu}$, $\psi_{\osc}$}{ }}\label{sec:oscsys}
In this section we will derive energy estimates for the oscillatory remainders $h_{\osc}^{\mu\nu}$, $\psi_{\osc}$ using the higher-order energy inequality of Proposition~\ref{sec:energyests.prop:energyhoineq}. We will choose different parameter values $\alpha,\beta$ for the energies $E_{\alpha,\beta}^{(K)}$ for each of $h_{\osc}^{00}$, $h_{\osc}^{0i}$, $h_{\osc}^{ij}$, $\psi_{\osc}$, specifically we consider the following energies:
\begin{equation}\label{sec:oscsys.eq:energyhochoice}
    E^{00}:= E_{q,\delta}^{(K)}[h_{\osc}^{00}],\quad E^{0i}:= E_{q+1,1+\delta}^{(K)}[h_{\osc}^{0i}],\quad E^{ij}:= E_{q+1,1+\delta}^{(K)}[h_{\osc}^{ij}],\quad E^{(\psi)}:= E_{q+1,1}^{(K)}[\psi_{\osc}].
\end{equation}
A short computation using the bound \eqref{sec:avgsys.eq:avgbound} and the energy coercivity inequalities \eqref{sec:energyests.eq:energyhocoercivity} shows that these energies are controlled by the bootstrap norm $S_K^2(\tau)$, from which one immediately obtains the bound
\begin{equation}\label{sec:oscsys.eq:oscroughest}
    E^{00} + \sum_{a=1}^3 E^{0a} + \sum_{a,b=1}^3 E^{ab} + E^{(\psi)} \leq C\ve^2.
\end{equation}
The goal of this section is to improve this bound. First we derive $H^K$-Sobolev norm bounds on the quantities 
\begin{equation}\label{sec:oscsys.eq:oscwaveop}
P_q h_{\osc}^{00}, \quad P_{q+1} h_{\osc}^{0i}, \quad P_{q+1} h_{\osc}^{ij},\quad  P_{q+1} \psi_{\osc},
\end{equation}
which appear in the differential inequalities for the energies \eqref{sec:oscsys.eq:energyhochoice}.
\begin{lemma}\label{sec:oscsys.lem:oscoperatorest}
    Suppose the gauge-fixed equations \eqref{sec:gaugedeqs.eq:wave} hold for $(h^{\mu\nu},\psi)$. Then the quantities \eqref{sec:oscsys.eq:oscwaveop} satisfy the following norm bounds:
    \begin{subequations}
        \begin{align}
            \|P_q h_{\osc}^{00}\|_{H^K} &\leq C\Big(\frac{1}{\tau}\|\pa \psi_{\osc}\|_{H^K} + \frac{1}{\tau^2}\|\pa h_{\osc}^{00}\|_{H^K} + \ve^2 \tau^{-(q+1+\delta)}\Big),\label{sec:oscsys.eq:oscwaveestlapse}\\
            \|P_{q+1} h_{\osc}^{0i}\|_{H^K} &\leq C\Big(\frac{1}{\tau}\|\pa h_{\osc}^{00}\|_{H^K} + \frac{1}{\tau}\|\pa \psi_{\osc}\|_{H^K}+ \frac{1}{\tau^2}\|\pa h_{\osc}^{0i}\|_{H^K}+\ve^2\tau^{-2(q-\delta)}\Big),\label{sec:oscsys.eq:oscwaveestshift}\\
            \|P_{q+1} h_{\osc}^{ij}\|_{H^K} &\leq C\Big(\frac{1}{\tau^2}\|\pa \psi_{\osc}\|_{H^K} + \ve^2\tau^{-(q+1+\delta)}\Big),\label{sec:oscsys.eq:oscwaveesthori}\\
            \|P_{q+1} \psi_{\osc}\|_{H^K} &\leq C\Big(\frac{1}{\tau^2}\|\pa \psi_{\osc}\|_{H^K} + \ve^2\tau^{-2(q-\delta)}\Big).\label{sec:oscsys.eq:oscwaveestscalar}
        \end{align}
    \end{subequations}
\end{lemma}
\begin{proof}
    We will prove \eqref{sec:oscsys.eq:oscwaveestlapse} first. Since $P_q h_{\osc}^{00} = P_q h^{00} - P_q h_{\av}^{00}$, we see from the wave equation \eqref{sec:gaugedeqs.eq:wavelapse} for $h^{00}$ that
    \begin{equation}
        P_q h^{00} = \frac{4q+2}{\tau}g^{00}h^{00} + \mc{L}^{00} + \mc{N}^{00},
    \end{equation}
    moreover we have
    \begin{align}
        P_q h_{\av}^{00} &= g^{00}\pa_\tau^2 h_{\av}^{00} + \frac{2q}{\tau}g^{00}\pa_\tau h^{00}\nonumber\\
        &= g^{00}\Big(\pa_\tau^2 h_{\av}^{00} + \frac{2q}{\tau}\pa_\tau h_{\av}^{00} - \frac{4q+2}{\tau^2}h_{\av}^{00}\Big) + \frac{4q+2}{\tau^2}g^{00}h_{\av}^{00}.
    \end{align}
    Combining these yields the identity
    \begin{equation}
        P_q h_{\osc}^{00} =  \frac{4q+2}{\tau^2}g^{00}h_{\osc}^{00} - g^{00}\Big(\pa_\tau^2 h_{\av}^{00} + \frac{2q}{\tau}\pa_\tau h_{\av}^{00} - \frac{4q+2}{\tau^2}h_{\av}^{00}\Big) + \mc{L}^{00} + \mc{N}^{00}.
    \end{equation}
    It follows from the Sobolev product lemma and the bootstrap assumption that
    \begin{align}
        \|P_q h_{\osc}^{00}\|_{H^K}
        \leq&\> \frac{4q+2}{\tau^2}\|g^{00}h_{\osc}^{00}\|_{H^K} + \Big|\pa_\tau^2 h_{\av}^{00} + \frac{2q}{\tau}\pa_\tau h_{\av}^{00} - \frac{4q+2}{\tau^2}h_{\av}^{00}\Big|\cdot\|g^{00}\|_{H^K}+\|\mc{L}^{00}\|_{H^K} + \|\mc{N}^{00}\|_{H^K}\nonumber\\
        \leq&\> C\Big(\frac{1}{\tau^2}\|h_{\osc}^{00}\|_{H^K} + \Big|\pa_\tau^2 h_{\av}^{00} + \frac{2q}{\tau}\pa_\tau h_{\av}^{00} - \frac{4q+2}{\tau^2}h_{\av}^{00}\Big| + \|\mc{L}^{00}\|_{H^K} +  \|\mc{N}^{00}\|_{H^K}\Big).
    \end{align}
    Now we estimate the right hand side terms. The differential inequality \eqref{sec:avgsys.eq:odeavglapseextra} for $h_{\av}^{00}$ (see  Remark~\ref{sec:avgsys.rmk:odelapseshiftextra} after Lemma~\ref{sec:avgsys.lem:odehoriscalar}) and the estimate from Proposition~\ref{sec:avgsys.prop:avgest} implies
    \begin{equation}
        \Big|\pa_\tau^2 h_{\av}^{00} + \frac{2q}{\tau}\pa_\tau h_{\av}^{00} - \frac{4q+2}{\tau^2}h_{\av}^{00}\Big| \leq C\Big(\frac{1}{\tau^2}|\psi_{\av}| + \frac{1}{\tau}|\pa_\tau \psi_{\av}|\Big) \leq C\ve^2\tau^{-(q+1+\delta)},
    \end{equation}
    while the bootstrap assumption, Proposition~\ref{sec:avgsys.prop:avgest} and the Poincar\'e inequality gives
    \begin{align}
        \|\mc{L}^{00}\|_{H^K} \leq& \frac{C}{\tau}\|g^{00}\|_{H^K}^2\|\pa_\tau \psi\|_{H^K} + \frac{C}{\tau^2}\|g^{00}\|_{H^K} \|\psi\|_{H^K}\nonumber\\
        \leq&\frac{C}{\tau}(\|\pa_\tau \psi_{\osc}\|_{H^K} + |\pa_\tau \psi_{\av}|) + \frac{C}{\tau^2}(\|\psi_{\osc}\|_{H^K} + |\psi_{\av}|)\nonumber\\
        \leq&\frac{C}{\tau}\|\pa \psi_{\osc}\|_{H^K} + C\ve^2\tau^{-(q+1+\delta)}.
    \end{align}
    The error estimate from Proposition~\ref{sec:bootbound.prop:errorests} reads
    \begin{equation}
        \|\mc{N}^{00}\|_{H^K} \leq C \ve^2\tau^{-2(q-\delta)},
    \end{equation}
    and the Poincar\'e estimate implies
    \begin{equation}
        \frac{1}{\tau^2}\|h_{\osc}^{00}\|_{H^K} \leq \frac{1}{\tau^2}\|\ol{\pa}h_{\osc}^{00}\|_{H^K}.
    \end{equation}
    We note that $\tau^{-(q+1+\delta)}$ has the slowest decay rate out of all the decaying errors, so combining these estimates yields \eqref{sec:oscsys.eq:oscwaveestlapse}.
    For the other estimates, a computation analogous to the one for $h_{\osc}^{00}$ reveals that $h_{\osc}^{0i}$, $h_{\osc}^{ij}$, $\psi_{\osc}$ obey the identities
    \begin{subequations}
        \begin{align}
        P_{q+1}h_{\osc}^{0i} &= \frac{2q+2}{\tau^2}g^{00}h_{\osc}^{0i} - g^{00}\Big(\pa_\tau^2 h_{\av}^{0i} + \frac{2q+2}{\tau}\pa_\tau h_{\av}^{0i} - \frac{2q+2}{\tau^2}h_{\av}^{0i}\Big)+\mc{L}^{0i}+\mc{N}^{0i},\\
        P_{q+1}h_{\osc}^{ij} &=  - g^{00}\Big(\pa_\tau^2 h_{\av}^{ij} + \frac{2q+2}{\tau}\pa_\tau h_{\av}^{ij}\Big)+\mc{L}^{ij}+\mc{N}^{ij},\\
        P_{q+1}\psi_{\osc} &= -\frac{2(2q+1)(q+2)}{\tau^2}g^{00}\psi_{\osc} - g^{00}\Big(\pa_\tau^2 \psi_{\av} + \frac{2q+2}{\tau}\pa_\tau \psi_{\av} + \frac{2(2q+1)(q+2)}{\tau^2}\psi_{\av}\Big)+\mc{N}^{(\psi)},
        \end{align}
    \end{subequations}
    and therefore $P_{q+1} h_{\osc}^{0i}$, $P_{q+1} h_{\osc}^{ij}$ $P_{q+1} \psi_{\osc}$ satisfy the Sobolev estimates
    \begin{subequations}
    \begin{align}
        \|P_{q+1}h_{\osc}^{0i}\|_{H^K} &\leq C \Big(\frac{1}{\tau^2}\|h_{\osc}^{0i}\|_{H^K} + \Big|\pa_\tau^2 h_{\av}^{0i} + \frac{2q+2}{\tau}\pa_\tau h_{\av}^{0i} - \frac{2q+2}{\tau^2}h_{\av}^{0i}\Big|+\|\mc{L}^{0i}\|_{H^K}+\|\mc{N}^{0i}\|_{H^K}\Big),\\
        \|P_{q+1}h_{\osc}^{ij}\|_{H^K} &\leq C\Big(\Big|\pa_\tau^2 h_{\av}^{ij} + \frac{2q+2}{\tau}\pa_\tau h_{\av}^{ij}\Big|+\|\mc{L}^{ij}\|_{H^K}+\|\mc{N}^{ij}\|_{H^K}\Big),\\
        \|P_{q+1}\psi_{\osc}\|_{H^K} &\leq C\Big( \frac{1}{\tau^2}\|\psi_{\osc}\|_{H^K} + \Big|\pa_\tau^2 \psi_{\av} + \frac{2q+2}{\tau}\pa_\tau \psi_{\av} + \frac{2(2q+1)(q+2)}{\tau^2}\psi_{\av}\Big|+\|\mc{N}^{(\psi)}\|_{H^K}\Big).
        \end{align}
    \end{subequations}
    The terms $\mc{N}^{0i}$, $\mc{N}^{ij}$, $\mc{N}^{(\psi)}$ are all estimated via Proposition~\ref{sec:bootbound.prop:errorests}. By Lemma~\ref{sec:avgsys.lem:odehoriscalar}, \eqref{sec:avgsys.eq:odeavgshiftextra}, and Proposition~\ref{sec:avgsys.prop:avgest}, we have
    \begin{subequations}
        \begin{align}
            \Big|\pa_\tau^2 h_{\av}^{0i} + \frac{2q+2}{\tau}\pa_\tau h_{\av}^{0i} - \frac{2q+2}{\tau^2}h_{\av}^{0i}\Big| &\leq C\ve^2\tau^{-2(q-\delta)},\\
            \Big|\pa_\tau^2 h_{\av}^{ij} + \frac{2q+2}{\tau}\pa_\tau h_{\av}^{ij}\Big| &\leq C\Big(\frac{1}{\tau^2}|\psi_{\av}| + \ve^2\tau^{-2(q-\delta)}\Big) \leq C\ve^2\tau^{-(q+1+\delta)},\\
            \Big|\pa_\tau^2 \psi_{\av} + \frac{2q+2}{\tau}\pa_\tau \psi_{\av} + \frac{2(2q+1)(q+2)}{\tau^2}\psi_{\av}\Big| &\leq C\ve^2\tau^{-2(q-\delta)}.
        \end{align}
    \end{subequations}
    For the $\mc{L}^{0i}$, $\mc{L}^{ij}$ terms we have
    \begin{subequations}
        \begin{align}
            \|\mc{L}^{0i}\|_{H^K} &\leq \frac{C}{\tau}(\|\ol{\pa} h_{\osc}^{00}\|_{H^K} + \|\ol{\pa}\psi_{\osc}\|_{H^K}),\\
            \|\mc{L}^{ij}\|_{H^K} &\leq \frac{C}{\tau^2}(\|\psi_{\osc}\|_{H^K} + |\psi_{\av}|) \leq \frac{C}{\tau^2}\|\psi_{\osc}\|_{H^K} + C\ve^2\tau^{-(q+1+\delta)}.
        \end{align}
    \end{subequations}
    The result now follows.
\end{proof}
We now derive energy estimates for the energies $E^{\mu\nu}$, $E^{(\psi)}$ defined in \eqref{sec:oscsys.eq:energyhochoice}. We note that by the energy coercivity inequality \eqref{sec:energyests.eq:energyhocoercivity} we have the estimates
\begin{subequations}\label{sec:oscsys.eq:energyequiv}
    \begin{align}
        \tau^{2(q-\delta)}(\|h_{\osc}^{00}\|_{H^K}^2 + \|\pa h_{\osc}^{00}\|_{H^K}^2) &\leq CE^{00},\\
        \tau^{2(q-\delta)}(\|h_{\osc}^{0i}\|_{H^K}^2 + \|\pa h_{\osc}^{0i}\|_{H^K}^2) &\leq CE^{0i},\\
        \tau^{2(q-\delta)}(\|h_{\osc}^{ij}\|_{H^K}^2 + \|\pa h_{\osc}^{ij}\|_{H^K}^2) &\leq CE^{ij},\\
        \tau^{2q}(\|\psi_{\osc}\|_{H^K}^2 + \|\pa \psi_{\osc}\|_{H^K}^2) &\leq CE^{(\psi)}.
    \end{align}
\end{subequations}
In the following proposition we improve the bound \eqref{sec:oscsys.eq:oscroughest} for these energies. We will use this improved bound, in combination with the estimate for the averages $(h_{\av}^{\mu\nu},\psi_{\av})$ from Proposition~\ref{sec:avgsys.prop:avgest}, to improve the bound on $S_N(\tau)$ from that of the bootstrap assumption \eqref{sec:globex.eq:bootstrap}.
\begin{proposition}[Energy estimate for the oscillatory remainders]\label{sec:oscsys.prop:oscest}
    The energies $E^{00}$, $E^{0i}$, $E^{ij}$, $E^{(\psi)}$ satisfy the following estimate for all $\tau \in [\tau_0,T]$:
    \begin{equation}\label{sec:oscsys.eq:oscest}
        E^{00}(\tau) + \sum_{a=1}^3 E^{0a}[\tau] + \sum_{a,b=1}^3 E^{ab}[\tau] + E^{(\psi)} \leq C\ve^4.
    \end{equation}
\end{proposition}
\begin{proof}
    Combining the energy estimate from Proposition~\ref{sec:energyests.prop:energyhoineq} and the estimates from Lemma~\ref{sec:oscsys.lem:oscoperatorest} the energies $E^{00}$, $E^{0i}$, $E^{ij}$, $E^{(\psi)}$ satisfy the differential inequalities
    \begin{subequations}
        \begin{align}
            \pa_\tau E^{00} \leq &-\frac{\delta}{\tau}E^{00} + \frac{C}{\tau^{1+\delta}}E^{00} + C\tau^{2(q-\delta)+1}\Big(\frac{1}{\tau^2}\|\pa \psi_{\osc}\|_{H^K}^2 + \frac{1}{\tau^4}\|\pa h_{\osc}^{00}\|_{H^K}^2 + \ve^4\tau^{-2(q+1+\delta)}\Big),\\
            \pa_\tau E^{0i}\leq &-\frac{1+\delta}{\tau}E^{0i} + \frac{C}{\tau^{1+\delta}}E^{0i}\\
            &+ C\tau^{2(q-\delta)+1}\Big(\frac{1}{\tau^2}\|\pa h_{\osc}^{00}\|_{H^K}^2+\frac{1}{\tau^2}\|\pa \psi_{\osc}\|_{H^K}^2 + \frac{1}{\tau^4}\|\pa h_{\osc}^{0i}\|_{H^K}^2+\ve^4 \tau^{-4(q-\delta)}\Big),\\
            \pa_\tau E^{ij}\leq &-\frac{1+\delta}{\tau}E^{ij} + \frac{C}{\tau^{1+\delta}}E^{ij}+ C\tau^{2(q-\delta)+1}\Big(\frac{1}{\tau^4}\|\pa \psi_{\osc}\|_{H^K}^2+\ve^4 \tau^{-2(q+1+\delta)}\Big),\\
            \pa_\tau E^{(\psi)}\leq &-\frac{1}{\tau}E^{(\psi)} + \frac{C}{\tau^{1+\delta}}E^{(\psi)}+ C\tau^{2q+1}\Big(\frac{1}{\tau^4}\|\pa \psi_{\osc}\|_{H^K}^2+\ve^4 \tau^{-4(q-\delta)}\Big),
        \end{align}
    \end{subequations}
    which implies by the bounds \eqref{sec:gaugedeqs.eq:wave} that
    \begin{subequations}
        \begin{align}
            \pa_\tau E^{00} \leq &-\frac{\delta}{\tau}E^{00}+\boxed{C\Big(\frac{1}{\tau^{1+\delta}}E^{00} + \frac{1}{\tau^{1+\delta}}E^{(\psi)} + \ve^4 \tau^{-(1+4\delta)}\Big)},
            \label{sec:oscsys.eq:oscestproof1}\\
            \pa_\tau E^{0i}\leq &-\frac{1+\delta}{\tau}E^{0i}+ \frac{C}{\tau}E^{00} + \boxed{C\Big(\frac{1}{\tau^{1+\delta}}E^{(\psi)}+\frac{1}{\tau^{1+\delta}}E^{0i}+\ve^4\tau^{-(2q-1-4\delta)}\Big)},\label{sec:oscsys.eq:oscestproof2}\\
            \pa_\tau E^{ij}\leq &-\frac{1+\delta}{\tau}E^{ij} + \boxed{C\Big(\frac{1}{\tau^{1+\delta}}E^{ij}+\frac{1}{\tau^{3+2\delta}}E^{(\psi)}+\ve^4 \tau^{-(1+4\delta)}\Big)},\\
            \pa_\tau E^{(\psi)}\leq &-\frac{1}{\tau}E^{(\psi)} + \boxed{C\Big(\frac{1}{\tau^{1+\delta}}E^{(\psi)} + \ve^4\tau^{-(2q-1-4\delta)}\Big)}.
        \end{align}
    \end{subequations}
    The boxed terms in each equation are integrable error terms (provided $\delta$ satisfies $q \geq 1+5\delta$), which can be estimated in a straightforward manner via a Gr\"onwall-type inequality. All but one remaining term on the right hand side of the above equations are negative. To deal with this, we combine the inequalities \eqref{sec:oscsys.eq:oscestproof1}, \eqref{sec:oscsys.eq:oscestproof2}. For any constant $c_1 > 0$, this implies the bound
    \begin{align}
         \sum_{a=1}^3\pa_\tau E^{0a} + c_1\pa_\tau E^{00} \leq&\> \boxed{\frac{3C-c_1\delta}{\tau}E^{00}} -\frac{1+\delta}{\tau}\sum_{a=1}^3E^{0a}\nonumber\\
        &+\frac{C}{\tau^{1+\delta}}\Big(E^{00}+\sum_{a=1}^3 E^{0a} + E^{(\psi)} + \ve^4\Big).\label{sec:oscsys.eq:oscestproof3}
    \end{align}
    We fix $c_1$ to be sufficiently large that the boxed term in \eqref{sec:oscsys.eq:oscestproof3} is non-positive. Thus we obtain the following estimate:
    \begin{equation}
        \pa_\tau \Big(c_1E^{00} + \sum_{a=1}^3 E^{0a} + \sum_{a,b=1}^3E^{ab} + E^{(\psi)}\Big)\\
        \leq \frac{C}{\tau^{1+\delta}}\Big(c_1 E^{00} + \sum_{a=1}^3 E^{0a} + \sum_{a,b=1}^3E^{ab} + E^{(\psi)} + \ve^4\Big),
    \end{equation}
    which by a Gr\"onwall-type estimate implies the bound
    \begin{equation}
        c_1E^{00}(\tau) + \sum_{a=1}^3 E^{0a}(\tau) + \sum_{a,b=1}^3E^{ab}(\tau) + E^{(\psi)}(\tau)
        \leq C\Big(E^{00}(\tau_0) + \sum_{a=1}^3 E^{0a}(\tau_0) + \sum_{a,b=1}^3E^{ab}(\tau_0) + E^{(\psi)}(\tau_0) + \ve^4\Big).
    \end{equation}
    Using the coercivity inequality \eqref{sec:energyests.eq:energyhocoercivity} for the energies $E^{\mu\nu}$, $E^{(\psi)}$ the bound \eqref{sec:avgsys.eq:avgbound} on functions with zero spatial average, and the bound \eqref{sec:gaugedeqs.eq:bounddata} on the initial data, we estimate
    \begin{align}
        E^{00}(\tau_0) + \sum_{a=1}^3 E^{0a}(\tau_0) + \sum_{a,b=1}^3E^{ab}(\tau_0) + E^{(\psi)}(\tau_0) &\leq C\sum_{\mu,\nu=1}^3\|\pa h_{\osc}^{\mu\nu}(\tau_0)\|_{H^K}^2  + \|\pa \psi_\osc(\tau_0)\|_{H^K}^2 \nonumber\\&\leq CS_K^2(\tau_0)\leq C\ve^4,
    \end{align}
    and we are done.
\end{proof}
\begin{remark}
    Th
\end{remark}
\section{Completing the proof of \texorpdfstring{Theorem~\ref{sec:globex.thm:globalexistence}}{the main theorem}}\label{sec:completethm}
In this section we complete the proof of the Theorem~\ref{sec:globex.thm:globalexistence} by (i)~improving the bootstrap assumption \eqref{sec:globex.eq:bootstrap}, (ii)~proving the perturbed spacetime $\wt{g}$ is future-causally geodesically complete, and (iii)~deriving the asymptotics \eqref{sec:globex.eq:asymp} and \eqref{sec:globex.eq:asymphori}.
\subsection{Improving the bootstrap assumption}
We improve the bootstrap assumption from Theorem~\ref{sec:globex.thm:globalexistence}, which we restate here:
\begin{equation}\label{sec:completethm.eq:bootstrap}
        S_K(\tau) \leq \ve,\qquad \tau \in [\tau_0,T].
\end{equation}
Recall from the continuity argument following Theorem~\ref{sec:globex.thm:globalexistence} that by improving this bound, we are in fact showing that the time interval of existence for the solutions $(h^{\mu\nu},\psi)$ to the gauge-fixed ENSF system \eqref{sec:gaugedeqs.eq:wave} is $[\tau_0,\infty)$, i.e. the solutions are future-global.
\begin{theorem}\label{sec:completethm.thm:bootimp}
    Let $(h^{\mu\nu},\psi)$ be solutions to the gauge-fixed Einstein-nonlinear scalar field system \eqref{sec:gaugedeqs.eq:wave} on $[\tau_0,T]\times \TT^3$. Assume that the norm $S_K$ of $(h^{\mu\nu},\psi)$ satisfies the initial bound
    \begin{equation}\label{sec:completethm.eq:initialbound}
        S_K(\tau_0) \leq \ve^2.
    \end{equation}
    Moreover, assume that $S_K$ obeys the bound from the bootstrap assumption \eqref{sec:completethm.eq:bootstrap}, where $\ve$ is fixed to be sufficiently small that the results from Sections~\ref{sec:bootbound}-\ref{sec:oscsys} hold. Then $S_K$ obeys the improved bound
    \begin{equation}\label{sec:completethm.eq:bootimp}
        S_K(\tau) \leq C\ve^2,
    \end{equation}
    for all $\tau \in [\tau_0,T]$.
\end{theorem}
\begin{proof}
    Using its definition, we bound $S_K$ by the averages $(h_{\av}^{\mu\nu},\psi_{\av})$ and the oscillating terms $(h_{\osc}^{\mu\nu},\psi_{\osc})$:
    \begin{align}
        S_K(\tau) \leq & \tau^{q-1}|h_{\av}^{00}| + \tau^{q-\delta}|\pa_\tau h_{\av}^{00}| + \sum_{a=1}^3\Big\{\tau^{q-\delta}|h_{\av}^{0a}| + \tau^{q-\delta}|\pa_\tau h_{\av}^{0a}|\Big\}\nonumber\\
        &+ \sum_{a,b=1}^3\Big\{|h_{\av}^{ab}| + \tau^{q-\delta}|\pa_\tau h_{\av}^{ab}|\Big\} + \tau^{q-1+\delta}|\psi_{\av}| + \tau^q|\pa_\tau \psi_{\av}|\nonumber\\
        &\tau^{q-1}\|h_{\osc}^{00}\|_{L^2} + \tau^{q-\delta}\|\pa h_{\osc}^{00}\|_{H^K} + \sum_{a=1}^3\Big\{\tau^{q-\delta}\|h_{\osc}^{0a}\|_{L^2} + \tau^{q-\delta}\|\pa h_{\osc}^{0a}\|_{H^K}\Big\}\nonumber\\
        &+ \sum_{a,b=1}^3\Big\{\|h_{\osc}^{ab}\|_{L^2} + \tau^{q-\delta}\|\pa h_{\osc}^{ab}\|_{H^K}\Big\} + \tau^{q-1+\delta}\|\psi_{\osc}\|_{L^2} + \tau^q\|\pa \psi_{\osc}\|_{H^K}.\label{sec:completethm.eq:bootimpproof1}
    \end{align}
    We estimate the averages $(h_{\av}^{\mu\nu},\psi_{\av})$ using the bound on $S_{\av}$ from Proposition~\ref{sec:avgsys.prop:avgest}:
    \begin{align}
        \tau^{q-1}|h_{\av}^{00}| + \tau^{q-\delta}|\pa_\tau h_{\av}^{00}| + \sum_{a=1}^3\Big\{\tau^{q-\delta}|h_{\av}^{0a}| + \tau^{q-\delta}|\pa_\tau h_{\av}^{0a}|\Big\}&\nonumber\\
        +\sum_{a,b=1}^3\Big\{|h_{\av}^{ab}| + \tau^{q-\delta}|\pa_\tau h_{\av}^{ab}|\Big\} + \tau^{q-1+\delta}|\psi_{\av}| + \tau^q|\pa_\tau \psi_{\av}| &\leq CS_{\av}(\tau)\nonumber\\
        &\leq C\ve^2.\label{sec:completethm.eq:bootimpproof2}
    \end{align}
    To estimate the oscillating remainders $(h_{\osc}^{\mu\nu},\psi_{\osc})$, recall the energies $E^{\mu\nu}$, $E^{(\psi)}$ defined in \eqref{sec:oscsys.eq:energyhochoice}, as well as the bounds \eqref{sec:oscsys.eq:energyequiv} satisfied by them. The estimate \eqref{sec:oscsys.eq:oscest} from Proposition~\ref{sec:oscsys.prop:oscest} for these energies then implies
    \begin{align}
        \tau^{q-1}\|h_{\osc}^{00}\|_{L^2} + \tau^{q-\delta}\|\pa h_{\osc}^{00}\|_{H^K} + \sum_{a=1}^3&\Big\{\tau^{q-\delta}\|h_{\osc}^{0a}\|_{L^2} + \tau^{q-\delta}\|\pa h_{\osc}^{0a}\|_{H^K}\Big\}\nonumber\\
        + \sum_{a,b=1}^3\Big\{\|h_{\osc}^{ab}\|_{L^2} + \tau^{q-\delta}\|\pa h_{\osc}^{ab}\|_{H^K}\Big\}& + \tau^{q-1+\delta}\|\psi_{\osc}\|_{L^2} + \tau^q\|\pa \psi_{\osc}\|_{H^K}\nonumber\\
        &\leq C\Big((E^{00})^{1/2} + \sum_{a=1}^3 (E^{0a})^{1/2} + \sum_{a,b=1}^3 (E^{ab})^{1/2} + (E^{(\psi)})^{1/2}\Big)\nonumber\\
        &\leq C\ve^2.\label{sec:completethm.eq:bootimpproof3}
    \end{align}
    The inequality \eqref{sec:completethm.eq:bootimp} now follows after plugging the bounds \eqref{sec:completethm.eq:bootimpproof2} and \eqref{sec:completethm.eq:bootimpproof3} into \eqref{sec:completethm.eq:bootimpproof1},
\end{proof}
\subsection{Future geodesic completeness}
We have demonstrated that the perturbed solution exists globally. Next we demonstrate that the spacetime $([\tau_0,\infty)\times\TT^3,\wt{g})$ is future-causally geodesically complete, proving part (III) of Theorem~\ref{sec:globex.thm:globalexistence}. Recall that a curve $\gamma:[s_0,s_1) \ra \mc{M}$ is \emph{future-causal} if for all $s \in [s_0,s_1)$, $\gamma$ obeys the inequalities
\begin{subequations}\label{sec:completethm.eq:causalineq} 
    \begin{align}
        \wt{g}[\gamma(s)](\dot{\gamma}(s),\dot{\gamma}(s)) &\leq 0,\label{sec:completethm.eq:causalineq1}\\
        \wt{g}[\gamma(s)](\dot{\gamma}(s),\pa_\tau) &< 0,\label{sec:completethm.eq:causalineq2}
    \end{align}
\end{subequations}
where $\wt{g}[\gamma(s)]$ is the metric $\wt{g}$ evaluated at the point $\gamma(s)$. A curve $\gamma$ is a geodesic of $(\mc{M},\wt{g})$ if it satisfies the geodesic equations
\begin{equation}\label{sec:completethm.eq:geodesic}
    \ddot{\gamma}^\mu(s) + \Gamma_{\alpha\beta}^\mu\big(\wt{g}[\gamma(s)]\big)\dot{\gamma}^{\alpha}(s)\dot{\gamma}^{\beta}(s) = 0,\qquad \mu=0,1,2,3.
\end{equation}
In the case that $\gamma$ is a geodesic, the causal inequalities \eqref{sec:completethm.eq:causalineq} are propagated, so that if they are satisfied at a single point (say $\gamma(s_0)$) then $\gamma$ is future-causal for all $s \in [s_0,s_1)$.

Recall also that a geodesic $\gamma:[s_0,s_{\max}) \ra \mc{M}$ is \emph{future-inextendible} if the limit $\lim_{s \nearrow s_{\max}}\gamma(s)$ does not exist. An inextendible geodesic is \emph{future-complete} if $s_{\max} = \infty$. One can equivalently define the past versions of causality, inextendibility and completeness. We are only interested in the future versions of these notions, so we will often drop the ``future'' specifier.

Finally, a Lorentzian manifold is \emph{future-causal geodesically complete} if all future-causal inextendible geodesics are complete. The goal of this section is to prove this property holds for the class of spacetime for which we just proved global existence. 

\begin{remark}[Geodesic completeness: accelerated expansion and decelerated expansion]
We base our argument for proving geodesic completeness around the one found in \cite{Rin:powerlaw:09}. In \cite{Rin:powerlaw:09}, however, the spacetimes under consideration are undergoing accelerated expansion. A key property of such spacetimes are that the spatial component of any future-causal geodesic $\gamma$ has a limit. That is, the limit
\begin{equation}
    \gamma_\infty^i = \lim_{s \nearrow s_{\max}}\gamma^i(s)
\end{equation}
exists for $i = 1,2,3$. This is used in \cite{Rin:powerlaw:09} to prove that future-causal inextendible geodesics are complete. In contrast, the functions $\gamma^i(s)$ do not have a limit for most future-causal geodesics on slowly expanding spacetimes.
While their spatial velocity decreases, future-causal geodesics can cycle around the spatial sections indefinitely.\footnote{One can see this by explicitly calculating the geodesics on the background FLRW solution. With initial position and velocity $\gamma(0) = (1,x_0)$, $\dot{\gamma}(0) = \xi$ where $\xi$ is future-pointing and null, the resulting geodesic $\gamma(s) = (\tau(s),x(s))$ is
\begin{equation*}
    \gamma^\mu(s) = \frac{1}{\xi^0}\Big(1+\xi^0\big[\frac{1+p}{1-p}\big]s\Big)^{\frac{1-p}{1+p}}\xi^\mu + \delta_{a}^\mu (x_0^a-(\xi^0)^{-1}\xi^a),
\end{equation*}
Since $\xi$ is null, $\xi^i \neq 0$ for one of $i \in \{1,2,3\}$, and thus from the above formula we see that $\gamma^i(s)$ does not have a limit as $s \ra \infty$.} This requires a modified argument to show that $\lim_{s \nearrow s_{\max}} \gamma^0(s) = \infty$ in Proposition~\ref{sec:completethm.prop:complete}, which uses an additional lower bound for $\dot{\gamma}^0$ which we prove in Lemma~\ref{sec:completethm.lem:geodesicbound}.
\end{remark}

First we prove some standard inequalities satisfied by future-causal curves on manifolds with Lorentzian metrics which are ``close'' to a zero shift metric.
\begin{lemma}\label{sec:completethm.lem:causalbounds}
    Let $(\mc{M},\wt{g})$ be a solution to the Einstein-nonlinear scalar field equation constructed in the first part of Theorem~\ref{sec:globex.thm:globalexistence} and let $\gamma:[s_0,s_{\max}) \ra \mc{M}$ be a future-causal curve on that spacetime. Then the time component of $\dot{\gamma}$ controls the spatial components, in the sense that
    \begin{equation}\label{sec:completethm.eq:causalbounds}
        \frac{2}{(1-p)^2}(\dot{\gamma}^0(s))^2 \geq \delta_{ab} \dot{\gamma}^a(s) \dot{\gamma}^b(s).
    \end{equation}
    Moreover, $\dot{\gamma}^0(s) > 0$ for all $s \in [s_0,s_{\max})$.
\end{lemma}
\begin{proof}
    The inequalities \eqref{sec:completethm.eq:causalineq} are preserved under conformal transformation, so we may replace the metric $\wt{g}$ with $g = \tau^{-2p/(1-p)}\wt{g}$ in these equations, which can  then be written as
    \begin{subequations}
        \begin{align}
            g_{00}(\dot{\gamma}^0)^2 + g_{0a}\dot{\gamma}^0\dot{\gamma}^a + g_{ab} \dot{\gamma}^a\dot{\gamma}^b &\leq 0,\label{sec:completethm.eq:causalboundsproof1}\\
            g_{00}\dot{\gamma}^0 + g_{0a} \dot{\gamma}^a &< 0.\label{sec:completethm.eq:causalboundsproof2}
        \end{align}
    \end{subequations}
    The global bound \eqref{sec:globex.eq:bootstrap} implies the following bounds:
    \begin{equation}
        g_{00}(\dot{\gamma}^0)^2 \geq - \Big[\frac{1}{(1-p)^2} + C\ve\Big](\dot{\gamma}^0)^2,\quad
        |2g_{0a}\dot{\gamma}^0\dot{\gamma}^a| \leq C\ve\big((\dot{\gamma}^0)^2 + \delta_{ab}\dot{\gamma}^a\dot{\gamma}^b\big),\quad
        g_{ab}\dot{\gamma}^a\dot{\gamma}^b \geq (1-C\ve)\delta_{ab}\dot{\gamma}^a\dot{\gamma}^b.
    \end{equation}
    Plugging the above inequalities into \eqref{sec:completethm.eq:causalboundsproof1} then yields
    \begin{equation}
         \Big[\frac{1}{(1-p)^2} + C\ve\Big](\dot{\gamma}^0)^2 \geq (1-C\ve)\delta_{ab}\dot{\gamma}^a\dot{\gamma}^b - C\ve\big((\dot{\gamma}^0)^2 + \delta_{ab}\dot{\gamma}^a\dot{\gamma}^b\big),
    \end{equation}
    which implies \eqref{sec:completethm.eq:causalbounds} after setting $\ve$ to be sufficiently small.

    The inequality \eqref{sec:completethm.eq:causalbounds} we just derived implies that $|\dot{\gamma}^0| > 0$, since if $\dot{\gamma} = 0$, then $\dot{\gamma}^i = 0$ for $i = 1,2,3$, therefore $\dot{\gamma} = 0$, contradicting \eqref{sec:completethm.eq:causalboundsproof2}. By the global bound \eqref{sec:globex.eq:bootstrap} we have
    \begin{subequations}
        \begin{gather}
            g_{00}\dot{\gamma}^0 \geq -(1-C\ve)|\dot{\gamma}^0|,\\
            |g_{0a}\dot{\gamma}^a| \geq -C\ve\sqrt{\delta_{ab}\dot{\gamma}^a\dot{\gamma}^b}, \geq -C\ve|\dot{\gamma}^0|,
        \end{gather}
    \end{subequations}
    where the last inequality follows from \eqref{sec:completethm.eq:causalbounds}. Then \eqref{sec:completethm.eq:causalboundsproof2} implies
    \begin{equation}
        -(1-C\ve)|\dot{\gamma}^0| <  |g_{0a} \dot{\gamma}^a| \leq C\ve|\dot{\gamma}^0|,
    \end{equation}
    and so $\dot{\gamma}^0 > 0$ (provided $\ve$ is sufficiently small).
\end{proof}

Next we use the geodesic equation \eqref{sec:completethm.eq:geodesic} to compute quantitative bounds on the tangent vector $\dot{\gamma}$ for a future-causal geodesic $\gamma:[s_0,s_{\max}) \ra \mc{M}$. We will use this to prove that if $\gamma$ is inextendible, then $s_{\max} = \infty$, implying that $\gamma$ is complete. 
\begin{lemma}[Causal geodesic differential inequality]\label{sec:completethm.lem:geodesicbound}
    Let $\gamma:[s_0,s_{\max}) \ra \mc{M}$ be a future-causal geodesic. Then there exist constants $c_2 > c_1 > 0$ such that for all $s \in [s_0,s_{\max})$, the zeroth component $\gamma^0$ obeys the differential inequality
    \begin{equation}
        \frac{c_1}{[\gamma^0(s)]^{4p/(1-p)}} \leq \dot{\gamma}^0(s) \leq \frac{c_2}{[\gamma^0(s)]^{p/(2(1-p))}}.
    \end{equation}
\end{lemma}
\begin{proof}
    Consider the $0$-component of the geodesic equations
    \begin{equation}\label{sec:completethm.eq:geodesicproof1}
        \ddot{\gamma}^0(s) + \Gamma_{\alpha\beta}^0\big(\wt{g}(\gamma(s))\big)\dot{\gamma}^\alpha(s)\dot{\gamma}^\beta(s) = 0.
    \end{equation}
    For convenience we write $\Gamma_{\alpha\beta}^0\big(\wt{g}(\gamma(s))\big) = \Gamma_{\alpha\beta}^0(\wt{g})$, and compute
    \begin{equation}
        \Gamma_{\alpha\beta}^0(\wt{g}) = \Gamma_{\alpha\beta}^0(g) + \frac{2p}{1-p}\frac{1}{\tau}\delta_\alpha^0\delta_\beta^0 - \frac{p}{1-p}\frac{1}{\tau}g^{00}g_{\alpha\beta}.
    \end{equation}
    As a consequence of the global bound \eqref{sec:globex.eq:bootstrap}, the Christoffel symbols all obey the bound $|\Gamma_{\mu\nu}^\lambda(g)| \leq C\ve \tau^{-(q-\delta)}$. Moreover, we have $|g_{\mu\nu} - m_{\mu\nu}| \leq C\ve \tau^{-(q-1-\delta)}$, and $|g^{\mu\nu} - m^{\mu\nu}| \leq C\ve \tau^{-(q-1-\delta)}$. Since $g^{00}g_{\alpha\beta} =\delta_\alpha^0 \delta_\beta^0 - (1-p)^2 \delta_\alpha^a\delta_\beta^b \delta_{ab} + O(g-m)$, we have the bounds
    \begin{equation}
        |\Gamma_{00}^0(\wt{g})- \frac{p}{1-p}\frac{1}{\tau}\big| \leq C\ve\tau^{-(q-\delta)},\qquad
        \big|\Gamma_{0i}^0(\wt{g})| \leq C\ve\tau^{-(q-\delta)},\qquad
    \big|\Gamma_{ij}^0(\wt{g}) - \frac{p(1-p)}{\tau}\delta_{ij}\big| \leq C\ve\tau^{-(q-\delta)}\delta_{ij}.
    \end{equation}
    The curve $\gamma$ is future-causal, so Lemma~\ref{sec:completethm.lem:causalbounds} implies the bounds
    \begin{equation}
        0 < \delta_{ab}\dot{\gamma}^a\dot{\gamma}^b \leq \frac{2}{(1-p)^2}(\dot{\gamma}^0)^2.
    \end{equation}
    Since $q-1-\delta > 0$ and $\tau = \gamma^0$, we may bound the error term like
    \begin{equation}
        \tau^{-(q-\delta)}\big((\dot{\gamma}^0)^2 + \delta_{ab}\dot{\gamma}^a\dot{\gamma}^b\big) \leq \frac{3(\dot{\gamma}^0)^2}{(\gamma^0)^{q-\delta}} \leq \frac{3}{(\gamma^0(s_0))^{q-1-\delta}}\frac{(\dot{\gamma}^0(s))^2}{\gamma^0(s)},
    \end{equation}
    and therefore (as long as $\ve$ is sufficiently small) we may safely absorb the error term into the principle terms, which implies the upper and lower bound
    \begin{subequations}
        \begin{gather}
            \Gamma_{\alpha\beta}^0(\wt{g})\dot{\gamma}^\alpha\dot{\gamma}^\beta \leq \frac{p}{1-p}\frac{1}{\tau}(\dot{\gamma}^0)^2 + \frac{p(1-p)}{\tau}\delta_{ab}\dot{\gamma}^a\dot{\gamma}^b  + C\ve\tau^{-(q-\delta)}\big((\dot{\gamma}^0)^2+\delta_{ab}\dot{\gamma}^a\dot{\gamma}^b\big)\leq \frac{4p}{1-p}\frac{1}{\tau}(\dot{\gamma}^0)^2,\\
            \Gamma_{\alpha\beta}^0(\wt{g})\dot{\gamma}^\alpha\dot{\gamma}^\beta \geq \frac{p}{1-p}\frac{1}{\tau}(\dot{\gamma}^0)^2 + \frac{p(1-p)}{\tau}\delta_{ab}\dot{\gamma}^a\dot{\gamma}^b  - C\ve\tau^{-(q-\delta)}\big((\dot{\gamma}^0)^2+\delta_{ab}\dot{\gamma}^a\dot{\gamma}^b\big)\geq \frac{p}{2(1-p)}\frac{1}{\tau}(\dot{\gamma}^0)^2
        \end{gather}
    \end{subequations}
    Plugging this into the geodesic equation \eqref{sec:completethm.eq:geodesicproof1} we obtain the bounds
    \begin{equation}\label{sec:completethm.eq:geodesicproof2}
        -\frac{4p}{1-p}\frac{1}{\tau}(\dot{\gamma}^0)^2\leq \ddot{\gamma}^0 \leq  - \frac{p}{2(1-p)}\frac{1}{\tau}(\dot{\gamma}^0)^2.
    \end{equation}
    We know that $\dot{\gamma}^0(s) > 0$ for all $s \in [0,s_{\max})$, and so we divide through by $\dot{\gamma}^0$ and integrate \eqref{sec:completethm.eq:geodesicproof2} over $[s_0,s]$. Keeping in mind that $\tau = \gamma^0(s)$, this yields
    \begin{equation}
        -\frac{4p}{1-p}\log\Big[\frac{\gamma^0(s)}{\gamma^0(s_0)}\Big] \leq \log \Big[\frac{\dot{\gamma}^0(s)}{\dot{\gamma}^0(s_0)}\Big] \leq -\frac{p}{2(1-p)}\log \Big[\frac{\gamma^0(s)}{\gamma^0(s_0)}\Big].
    \end{equation}
    Thus we have
    \begin{equation}
        \frac{c_1}{[\gamma^0(s)]^{4p/(1-p)}} \leq \dot{\gamma}^0(s) \leq \frac{c_2}{[\gamma^0(s)]^{p/(2(1-p))}},
    \end{equation}
    where $c_1,c_2$ are the positive constants 
    \begin{equation}
        c_1 = \dot{\gamma}^0(s_0)[\gamma_0(s_0)]^{4p/(1-p)},\qquad c_2 = \dot{\gamma}^0(s_0)[\gamma_0(s_0)]^{p/(2(1-p))}.
    \end{equation}.
\end{proof}
Finally we prove completeness of future-causal inextendible geodesics. This completes the proof of Theorem~\ref{sec:globex.thm:globalexistence}.
\begin{proposition}[Completeness of future-causal inextendible geodesics]\label{sec:completethm.prop:complete}
    Let $(\mc{M},\wt{g})$ be a Lorentzian manifold constructed in Theorem~\ref{sec:globex.thm:globalexistence}, and let $\gamma:[s_0,s_{\max}) \ra \mc{M}$ be a future-causal inextendible geodesic. Then $\lim_{s\nearrow s_{\max}}\gamma^0(s) = \infty$. Moreover, $s_{\max} = \infty$, hence $\gamma$ is complete.
\end{proposition}
\begin{proof}
    First we show that $\lim_{s \nearrow s_{\max}} \gamma^0(s) = \infty$. As $\dot{\gamma}^0(s) > 0$ for all $s \in [s_0,s_{\max})$, $\gamma^0$ is monotonically increasing, and so the only other possibility is that $\gamma^0$ has a finite limit which is approached from below. Suppose $\lim_{s \nearrow s_{\max}}\gamma^0(s) = L < \infty$. Then we have the two cases $s_{\max} = \infty$ or $s_{\max} < \infty$. If $s_{\max} = \infty$, let $s_n$ be a sequence $s_n \ra \infty$. For all $s \in [s_0,s_{\max})$, $\gamma^0(s) \leq L$, and so the bound from Lemma~\ref{sec:completethm.lem:geodesicbound} implies
    \begin{equation}
        \dot{\gamma}^0(s_n) \geq  \frac{c_1}{L^{4p/(1-p)}}.
    \end{equation}
    Integrating this we obtain the lower bound
    \begin{equation}
        \gamma^0(s_n) \geq \gamma^0(s_0) + c(s_n-s_0),
    \end{equation}
    which is unbounded. For the second case $s_{\max} < \infty$, we note that $\gamma^0$ has a limit, and therefore $\lim_{s \nearrow s_{\max}} \gamma^i(s)$ does not exist for one of $i = 1,2,3$. $\gamma$ is future causal, so $\dot{\gamma}^0$ controls all other components of $\dot{\gamma}$, and therefore by Lemma~\ref{sec:completethm.lem:geodesicbound} we have
    \begin{equation}
        |\dot{\gamma}^i(s)|\leq \frac{c_2}{(\gamma^0(s))^{p/(2p(1-p))}},\qquad i=1,2,3.
    \end{equation}
    Let $s_n$ be a sequence such that $s_n \nearrow s_{\max} <\infty$. Then there exists $N \in \NN$ such that $\gamma^0(s_n) \geq L/2$ for all $n \geq N$. The above bound then implies that
    \begin{equation}
        |\dot{\gamma}^i(s_n)| \leq C
    \end{equation}
    for all $n \geq N$, and so integrating over $[s_m,s_n]$ implies
    \begin{equation}
        |\gamma^i(s_n) - \gamma^i(s_m)| \leq C|s_m - s_n|,
    \end{equation}
    implying that $\gamma^i(s_n)$ is a Cauchy, therefore convergent sequence. But this is true for $i=1,2,3$, contradicting the fact that $\lim_{s\nearrow s_{\max}} \gamma(s)$ does not exist. Hence $\lim_{s\nearrow s_{\max}}\gamma^0(s) = \infty$.

    Next we prove that $s_{\max} = \infty$. Once again we use Lemma~\ref{sec:completethm.lem:geodesicbound}, which by the monotonicity of $\gamma^0$ implies
    \begin{equation}
        \dot{\gamma}^0(s) \leq \frac{c_2}{(\gamma^0(s))^{p/(2(1-p))}} \leq \frac{c_2}{(\gamma^0(s_0))^{p/(2(1-p))}}.
    \end{equation}
    Thus we have
    \begin{equation}
        \gamma^0(s) -\gamma^0(s_0) = \int_{s_0}^s \dot{\gamma}^0(\sigma)\de \sigma \leq C(s-s_0),
    \end{equation}
    and hence 
    \begin{equation}
        s_{\max} \geq \frac{1}{C}\big(\lim_{s \nearrow s_{\max}} \gamma^0(s) - 1\big) = \infty.
    \end{equation}
    This completes the proof of geodesic completeness, moreover of Part (II), of Theorem~\ref{sec:globex.thm:globalexistence}.
\end{proof}
\subsection{Future asymptotics}
In the following proposition we derive the asymptotics contained in part (IV) of Theorem~\ref{sec:globex.thm:globalexistence}, thus completing its proof.
\begin{proposition}[Asymptotics of perturbed solutions to the ENSF system]\label{sec:completethm.prop:asymp}
    Let $([\tau_0,\infty)\times\TT^3,\wt{g},\phi)$ be a solution to the Einstein-nonlinear scalar field system such that $S_K(\tau) \leq \ve$ for all $\tau \in [\tau_0,\infty)$. Then $\wt{g}^{00}$, $\wt{g}^{0i}$, $\phi$ satisfy the inequalities \eqref{sec:globex.eq:asymplapse}, \eqref{sec:globex.eq:asympshift}, and \eqref{sec:globex.eq:asympscalar} respectively. Moreover, there exists a symmetric, positive definite $3\times3$ matrix $(g_\infty^{ij})_{i,j=1,2,3}$ such that $|g_{\infty}^{ij} - \delta^{ij}| \leq C\ve$, and $\wt{g}^{ij}$, $g_\infty^{ij}$ satisfy the inequality \eqref{sec:globex.eq:asymphori}. 
\end{proposition}
\begin{proof}
    First, we note that the inequality $S_K(\tau) \leq \ve$ holds for the global solutions constructed in part (II) of Theorem~\ref{sec:globex.thm:globalexistence}. By this inequality and the Sobolev embedding theorem, we have
    \begin{gather}
        \tau^{q-1}\|g^{00} - m^{00}\|_{W^{K-1,\infty}} \leq C\ve,\qquad 
        \tau^{q-\delta}\sum_{a=1}^3 \|g^{0a}\|_{W^{K-1,\infty}} \leq C\ve,\qquad 
        \tau^{q-1+\delta}\|\phi - \phi_b\|_{W^{K-1,\infty}}\leq C\ve,\\
        \sum_{a,b=1}^3 \Big(\|g^{ab} - \delta^{ab}\|_{W^{K-1,\infty}}  + \tau^{q-\delta}\|\pa g^{ab}\|_{W^{K-1,\infty}} \Big)\leq C\ve.\label{sec:completethm.eq:horibound}
    \end{gather}
    The inequalities \eqref{sec:globex.eq:asymplapse}, \eqref{sec:globex.eq:asympshift}, and \eqref{sec:globex.eq:asympscalar} then follow  from the first three inequalities respectively, since $q = \frac{p}{1-p}-1$, and $\wt{g}^{\mu\nu} = \tau^{-2p/(1-p)}g^{\mu\nu}$. For the horizontal metric components, we note that the Poincar\'e inequality and the Sobolev inequality implies
    \begin{equation}
        \|g_{\osc}^{ij}\|_{W^{K-2,\infty}} \leq C\|g_\osc^{ij}\|_{H^K} \leq C\|\ol{\pa} g^{ij}\|_{H^K} \leq C\ve \tau^{-(q-\delta)}, 
    \end{equation}
    and therefore
    \begin{equation}\label{sec:completethm.eq:horilim1}
        \lim_{\tau \ra \infty}\|g^{ij}(\tau) - g_{\av}^{ij}(\tau)\|_{W^{K-2,\infty}} = 0.
    \end{equation}
    Additionally by the fundamental theorem of calculus, the Sobolev inequality and the bound $S_K(\tau) \leq \ve$, it holds for all $\tau_2 > \tau_1 \geq \tau_0$ that
    \begin{align}
        \big|\|g^{ij}(\tau_2)\|_{W^{K-2,\infty}} - \|g^{ij}(\tau_1)\|_{W^{K-2,\infty}}\big| &\leq C\int_{\tau_1}^{\tau_2}\|\pa_\tau g^{ij}(\tau)\|_{W^{K-2,\infty}}\de \tau\nonumber\\
        &\leq C\ve\int_{\tau_1}^{\tau_2}\tau^{-(q-\delta)} \leq C\ve\tau_1^{-(q-1-\delta)},
    \end{align}
    and so $\|g^{ij}(\tau)\|_{W^{K-2,\infty}}$ converges to some value, which by \eqref{sec:completethm.eq:horilim1} must be 
    \begin{equation}
        \lim_{\tau\ra\infty}\|g^{ij}(\tau)\|_{W^{K-2,\infty}} = \lim_{\tau \ra \infty}|g_\av^{ij}(\tau)| =: g_\infty^{ij}.
    \end{equation}
    By \eqref{sec:completethm.eq:horibound} we have
    \begin{align}
        |g_\infty^{ij}-\delta^{ij}| &\leq \limsup_{\tau\ra\infty}\|g_{\infty}^{ij} - g^{ij}(\tau)\|_{W^{K-2,\infty}}+\limsup_{\tau\ra\infty}\|g^{ij}(\tau) - \delta^{ij}\|_{W^{K-2,\infty}}\nonumber\\
        & = \limsup_{\tau\ra\infty}\|g^{ij}(\tau) - \delta^{ij}\|_{W^{K-2,\infty}} \leq C\ve.
    \end{align}
    This also implies that $(g_\infty^{ij})_{i,j=1,2,3}$ is positive definite (provided $\ve$ is sufficiently small). The bound \eqref{sec:globex.eq:asymphori} follows from a second application of the fundamental theorem of calculus:
    \begin{equation}
        \|\tau^{2p/(1-p)}\wt{g}^{ij}(\tau) - g_{\infty}^{ij}\|_{W^{K-2,\infty}} = \|g^{ij}(\tau) - g_{\infty}^{ij}\|_{W^{K-2,\infty}}\leq \int_{\tau}^\infty\|\pa_\tau g^{ij}\|_{W^{K-2,\infty}} \leq C\ve \tau^{-(p/(1-p)-2-\delta)},
    \end{equation}
    and we are done.
\end{proof}
\begin{remark}[Higher-order asymptotics]
    We derive only leading order asymptotics, which are determined primarily from the boundedness of the norm $S_K(\tau)$ and the weights in $\tau$ contained therein. We do not derive higher-order asymptotics, (such as asymptotic expansions) because of the oscillatory nature of the higher derivatives, which ultimately stems from the setting of decelerated expansion.

    The oscillating remainders $(h_{\osc}^{\mu\nu},\psi_{\osc})$ have a leading-order decay due to spacetime expansion, but once rescaled are purely oscillatory. For example, $|h_{\osc}^{00}| =  O(\tau^{-(q-\delta)})$, so take $v = \tau^{q-\delta} \psi_{\osc}$. By the arguments in Sections~\ref{sec:energyests} and~\ref{sec:oscsys}, the energy of $v$ is bounded (and small), but $v$ is oscillatory in nature, and does not possess an asymptotic expansion. This is analogous to the flat homogeneous wave equation $\square u = 0$ on $\RR\times\TT^3$. The remainder $u_{\osc}$ has finite energy, but is an infinite sum of the oscillating modes $\exp(ik\cdot x\pm i|k|t)$, $k \in \ZZ^3\bs\{0\}$.

    We contrast this with the regime of accelerated expansion, where one can derive asymptotic expansions of solutions to the Einstein equations in various settings (see for example \cite{HinVas:KdSexpandingstab:24}). This is possible to do because of the strong localisation that occurs within the spacetimes. Points in spacetime interact with a smaller and smaller portion of the spatial geometry over time, resulting in perturbations which ``smooth out'' in time (but not necessarily in space).
\end{remark}
\appendix
\section{Proof of Sobolev composition lemma}\label{app:sobolevcomp}
\begin{proof}[Proof of Lemma~\ref{sec:prelims.lem:sobolevcomp}]
    We prove the case where $S \subset \RR$, the general case follows by the same argument. We Taylor expand $F$ about $s = 0$ to order $l$:
    \begin{equation}
        |F(s)| = \Big|F(s) - \sum_{j=0}^{l-1}\frac{F^{(j)}(0)}{j!}s^j\Big| \leq C\sup_{S}|F^{(l)}|\cdot|s|^{l} \leq C|s|^{l}.
    \end{equation}
    Using the Sobolev product inequality from Lemma~\ref{sec:prelims.lem:sobolevproduct}, this implies the $L^2$ bound
    \begin{equation}
        \|F\circ u\|_{L^2} \leq C\|u^{l}\|_{L^2} \leq C\|u\|_{H^K}^{l}.
    \end{equation}
    For the higher derivatives, we note that for all $1 \leq |I| \leq K$, the quantity $\ol{\pa}{}^I(F \circ u)$ is equal to a sum of terms of the form
    \begin{equation}
        (\ol{\pa}{}^{J_1}u) \cdots (\ol{\pa}{}^{J_i}u) (F^{(i)}\circ u),
    \end{equation}
    where $1 \leq i \leq |I|$, and $J_1,\dots,J_i$ are multiindices with $|J_1| + \dots |J_i| = |I|$. If $i \leq l-1$ then we Taylor expand $F^{(i)}$ about $s = 0$ to order $l-i$:
    \begin{equation}
        |F^{(i)}(s)| = \Big|F^{(i)}(s) - \sum_{j=0}^{l-i-1}\frac{F^{(i+j)}(0)}{j!}s^j\Big| \leq C\sup_{S}|F^{(l)}|\cdot|s|^{l-i} \leq C|s|^{l-i},
    \end{equation}
    and so
    \begin{equation}
        \|(\ol{\pa}{}^{J_1}u) \cdots (\ol{\pa}{}^{J_i}u) (F^{(i)}\circ u)\|_{L^2} \leq C\|(\ol{\pa}{}^{J_1}u) \cdots (\ol{\pa}{}^{J_i}u) u^{l-i}\|_{L^2}.
    \end{equation}
    Since $K \geq 3$, only one of the factors $\ol{\pa}{}^{J_1}u,\dots,\ol{\pa}{}^{J_i}u$ have more than $K-2$ derivatives, we call this $J_{i_0}$. It follows from Cauchy-Schwartz and the Sobolev embedding from Lemma~\ref{sec:prelims.lem:sobolev} that
    \begin{align}
        C\|(\ol{\pa}{}^{J_1}u) \cdots (\ol{\pa}{}^{J_i}u) u^{l-i}\|_{L^2} &\leq C\|u\|_{L^\infty}^{l-i}\|\ol{\pa}{}^{J_{i_0}}u\|_{L^{2}}\prod_{1 \leq i' \leq i,i'\neq i_0}\|\ol{\pa}{}^{J_{i'}}u\|_{L^{\infty}}\nonumber\\
        &\leq C\|u\|_{H^2}^{l-i}\cdot\|u\|_{H^1}\cdot\|u\|_{H^K}^{i-1} \nonumber\\
        &\leq C\|u\|_{H^K}^{l}.
    \end{align}
    If $i \geq l$ then we do not Taylor expand $F^{(i)}$, we just bound
    \begin{equation}
        \|(\ol{\pa}{}^{J_1}u) \cdots (\ol{\pa}{}^{J_i}u) (F^{(i)}\circ u)\|_{L^2} \leq C\|(\ol{\pa}{}^{J_1}u) \cdots (\ol{\pa}{}^{J_i}u)\|_{L^2}.
    \end{equation}
    There are at least $l$ factors of $\ol{\pa}{}^{J_{i'}} u$, and again at most one has $|J_{i'}| \geq K-2$. Thus we bound again by Cauchy-Schwartz and the Sobolev embedding
    \begin{equation}
        \|(\ol{\pa}{}^{J_1}u) \cdots (\ol{\pa}{}^{J_i}u)\|_{L^2} \leq C\|u\|_{H^K}^{l},
    \end{equation}
    and we are done.
\end{proof}
\section{Proofs of Ricci curvature identities}\label{app:ricciids}
In this appendix we give proofs of the Ricci curvature identities in Lemmas~\ref{sec:gaugedeqs.lem:ricciconformal} and~\ref{sec:gaugedeqs.lem:ricciwave}.

\begin{proof}[Proof of Lemma~\ref{sec:gaugedeqs.lem:ricciconformal}]
    Let $g = \Omega^2 \wt{g}$, where $\Omega^2$ is a nonzero scalar function on our Lorentzian manifold. We will denote covariant differentiation with respect to $g$ by $\nabla$.  We begin with the identity
    \begin{equation}\label{app:ricciids.eq:ricciconfproof0}
        \Ric_{\mu\nu}[g] = \Ric_{\mu\nu}[\wt{g}] - 2\nabla_\mu\nabla_\nu (\log \Omega) - g_{\mu\nu}\square_g (\log \Omega) - 2\nabla_\mu (\log \Omega) \nabla_\nu(\log \Omega) + 2g_{\mu\nu}g^{\alpha\beta}\nabla_\alpha (\log\Omega)\nabla_\beta (\log \Omega).
    \end{equation}
    This can be found in for example in equation (D.6) in Appendix D of \cite{Wal:GRbook:84},\footnote{We note that in \cite{Wal:GRbook:84} the roles of the metric and conformal metric are reversed compared to Lemma~\ref{sec:gaugedeqs.lem:ricciconformal} here. This means one must swap $g \leftrightarrow \wt{g}$, and therefore change conformal factor from $\Omega^{2}$ to $\Omega^{-2}$.} and is derived from the Ricci curvature identity
    \begin{equation}\label{app:ricciids.eq:riccifundid}
        \Ric_{\mu\nu}[g] = \pa_\lambda \Gamma_{\mu\nu}^\lambda(g) - \pa_{(\mu}\Gamma_{\nu)\lambda}^\lambda(g) + \Gamma_{\mu\nu}^{\lambda}(g)\Gamma_{\lambda \delta}^{\delta}(g) - \Gamma_{\mu\lambda}^{\delta}(g)\Gamma_{\nu\delta}^{\lambda}(g),
    \end{equation}
    along with the identity that relates the Christoffel symbols of $g$ and $\wt{g}$:
    \begin{equation}
        \Gamma_{\mu\nu}^\lambda(g) - \Gamma_{\mu\nu}^\lambda(\wt{g}) = 2\delta_{(\mu}^\lambda \pa_{\nu)}(\log \Omega) - g_{\mu\nu}g^{\lambda\delta}\pa_\delta(\log \Omega).
    \end{equation}
    We take \eqref{app:ricciids.eq:ricciconfproof0} and raise indices with respect to the \emph{conformal metric} $g$:
    \begin{multline}\label{app:ricciids.eq:ricciconfproof1}
        \Ric^{\mu\nu}[g] = g^{\mu\alpha}g^{\nu\beta}\Ric_{\alpha\beta}[\wt{g}]-2\nabla^\mu\nabla^\nu (\log \Omega) - g^{\mu\nu}\square_g (\log\Omega) \\- 2\nabla^\mu(\log\Omega)\nabla^\nu(\log\Omega) + 2g^{\mu\nu}g^{\alpha\beta}\nabla_\alpha (\log \Omega)\nabla_\beta(\log\Omega).
    \end{multline}
    Then we decompose the terms on the right hand side containing second-order derivatives of $\log \Omega$:
    \begin{subequations}
        \begin{align}
            -2\nabla^\mu\nabla^\nu (\log \Omega) &= -2g^{\mu\alpha}g^{\nu\beta}\big(\pa_\alpha\pa_\beta (\log \Omega) - \Gamma_{\alpha\beta}^\delta(g)\pa_\delta (\log \Omega)\big),\label{app:ricciids.eq:ricciconfproof2}\\
            -g^{\mu\nu}\square_g (\log\Omega) &= -g^{\mu\nu}\big(g^{\alpha\beta}\pa_\alpha\pa_\beta (\log \Omega) - \Gamma^\delta(g)\pa_\delta (\log\Omega)\big),\label{app:ricciids.eq:ricciconfproof3}
        \end{align}
    \end{subequations}
    where the contracted Christoffel symbol $\Gamma^\delta(g) = g^{\alpha\beta}\Gamma_{\alpha\beta}^\delta(g)$. Derivatives of raised and lowered metric components are related by the identity
    \begin{equation}\label{app:ricciids.eq:diffmetidentity}
        g^{\mu\alpha}g^{\nu\beta} \pa_\lambda g_{\alpha\beta} = - \pa_\lambda g^{\mu\nu},
    \end{equation}
    and therefore
    \begin{equation}\label{app:ricciids.eq:ricciconfproof4}
        2g^{\mu\alpha}g^{\nu\beta}\Gamma_{\alpha\beta}^\delta(g) = g^{\mu\alpha}g^{\nu\beta}g^{\delta\lambda}\big(2\pa_{(\alpha} g_{\beta)\lambda} - \pa_\lambda g_{\alpha\beta}\big) = g^{\delta\lambda}\pa_\lambda g^{\mu\nu}- 2g^{\lambda(\mu}\pa_\lambda g^{\nu) \delta}.
    \end{equation}
    Combining the equations \eqref{app:ricciids.eq:ricciconfproof2}, \eqref{app:ricciids.eq:ricciconfproof3}, \eqref{app:ricciids.eq:ricciconfproof4} and plugging them into \eqref{app:ricciids.eq:ricciconfproof1} we obtain
    \begin{align}
        \Ric^{\mu\nu}[g] =&\> g^{\mu\alpha}g^{\nu\beta}\Ric_{\alpha\beta}[\wt{g}]-2g^{\mu\alpha}g^{\nu\beta}\pa_\alpha\pa_\beta (\log \Omega) + g^{\delta\lambda}\pa_\lambda g^{\mu\nu}\pa_\delta (\log \Omega)- 2g^{\lambda(\mu}\pa_\lambda g^{\nu)\delta}\pa_\delta (\log \Omega) \nonumber\\
        &- g^{\mu\nu} g^{\alpha\beta}\pa_\alpha\pa_\beta (\log\Omega) + g^{\mu\nu}\Gamma^{\lambda}(g)\pa_\lambda (\log \Omega)\nonumber\\
        &- 2g^{\mu\alpha}g^{\nu\beta}\pa_\alpha(\log\Omega)\pa_\beta(\log\Omega) + 2g^{\mu\nu}g^{\alpha\beta}\pa_\alpha (\log \Omega)\pa_\beta(\log\Omega).
    \end{align}
    The result now follows after rearranging and factoring terms appropriately.
\end{proof}

\begin{proof}[Proof of Lemma~\ref{sec:gaugedeqs.lem:ricciwave}]
    We begin with the following identity for the \emph{covariant} Ricci tensor:
    \begin{multline}\label{app:ricciids.eq:riccicovarid}
        2\Ric_{\mu\nu}[g] = -\wh{\square}_g g_{\mu\nu} +2g_{\lambda(\mu}\pa_{\nu)}\Gamma^\lambda(g) + \Gamma^{\lambda}(g)\pa_\lambda g_{\mu\nu}\\
        + 2g^{\alpha\beta}g^{\lambda\delta}\pa_\alpha g_{\lambda(\mu}\pa_{\nu_)}g_{\beta\delta} 
        - \frac{1}{2}g^{\alpha\beta}g^{\lambda\delta}\pa_{\mu}g_{\alpha\lambda}\pa_{\nu}g_{\beta\delta}
        + g^{\alpha\beta}g^{\lambda\delta}\pa_\alpha g_{\mu\lambda}\pa_\beta    g_{\nu\delta} - g^{\alpha\beta}g^{\lambda\delta}\pa_\alpha g_{\mu\delta}\pa_\lambda g_{\nu\beta}.
    \end{multline}
    This can be also derived from the Ricci curvature identity \eqref{app:ricciids.eq:riccifundid}; see for example Lemma~3.3 in \cite{HunLuk:Burnett:24}. We use this to derive the analogous expression for the contravariant Ricci tensor $R^{\mu\nu}$. Raising both indices, we note that
    \begin{align}
        -g^{\mu\mu'}g^{\nu\nu'}\wh{\square}_g g_{\mu'\nu'} &= -g^{\alpha\beta}\pa_\alpha (g^{\mu\mu'}g^{\nu\nu'}\pa_\beta g_{\mu'\nu'}) + g^{\alpha\beta} \pa_\alpha (g^{\mu\mu'}g^{\nu\nu'} )\pa_\beta g_{\mu'\nu'}\nonumber\\
        &= \wh{\square}_g g^{\mu\nu} + 2g^{\alpha\beta} g^{\lambda (\mu|}\pa_\alpha g^{|\nu)\delta} \pa_\beta g_{\lambda \delta}\label{app:ricciids.eq:ricciidentityraisedindices1},
    \end{align}
    where the last line follows from \eqref{app:ricciids.eq:diffmetidentity} as well as the symmetry of $g$. Applying \eqref{app:ricciids.eq:diffmetidentity} to the terms on the RHS of \eqref{app:ricciids.eq:riccicovarid} (relabelling summed indices where necessary) we have
    \begin{subequations}
        \begin{align}
            g^{\mu\mu'}g^{\nu\nu'} \Gamma^\lambda(g) \pa_\lambda g_{\mu'\nu'} &= -g^{\alpha\beta}g^{\lambda\delta}\pa_\lambda g^{\mu\nu}\pa_{(\alpha} g_{\beta)\delta} + \frac{1}{2}g^{\alpha\beta}g^{\lambda\delta}\pa_\lambda g^{\mu\nu}\pa_\delta g_{\alpha\beta},\\
            g^{\mu\mu'}g^{\nu\nu'}g_{\lambda(\mu'}\pa_{\nu')}\Gamma^\lambda(g) &= g^{\lambda (\mu|}\pa_\lambda \Gamma^{|\nu)}(g),\\
            2g^{\mu\mu'}g^{\nu\nu'}g^{\alpha\beta}g^{\lambda\delta}\pa_\alpha g_{\lambda(\mu'}\pa_{\nu')}g_{\beta\delta} &= -2g^{\lambda(\mu}g^{\nu)\delta}\pa_\delta g^{\alpha\beta}\pa_\alpha g_{\beta\lambda},\\
            - \frac{1}{2}g^{\mu\nu'}g^{\nu\nu'}g^{\alpha\beta}g^{\lambda\delta}\pa_{(\mu'}g_{\alpha\lambda}\pa_{\nu')}g_{\beta\delta} &= \frac{1}{2}g^{\lambda(\mu}g^{\nu)\delta}\pa^\lambda g_{\alpha\beta}\pa_\delta g_{\alpha\beta},\\
            g^{\mu\mu'}g^{\nu\nu'}g^{\alpha\beta}g^{\lambda\delta}\pa_\alpha g_{\mu'\lambda}\pa_\beta g_{\nu'\delta}  &= - g^{\alpha\beta}g^{\lambda(\mu|}\pa_\alpha g^{\nu)\delta}\pa_\beta g_{\lambda\delta},\label{app:ricciids.eq:ricciidentityraisedindices2}\\
            -g^{\mu\mu'}g^{\nu\nu'}g^{\alpha\beta}g^{\lambda\delta}\pa_\alpha g_{\mu'\delta}\pa_\lambda g_{\nu'\beta} &= 
            g^{\alpha\beta}g^{\lambda(\mu|}\pa_\alpha g^{|\nu)\delta}\pa_\delta g_{\beta\lambda}.
    \end{align}
    \end{subequations}
    Plugging these into \eqref{app:ricciids.eq:riccicovarid}, and noting the partial cancellation of terms in \eqref{app:ricciids.eq:ricciidentityraisedindices1} and \eqref{app:ricciids.eq:ricciidentityraisedindices2}, we obtain the desired identity.
\end{proof}
\footnotesize
\bibliographystyle{amsalpha}
\bibliography{main.bib}

\providecommand{\bysame}{\leavevmode\hbox to3em{\hrulefill}\thinspace}
\providecommand{\MR}{\relax\ifhmode\unskip\space\fi MR }
\providecommand{\MRhref}[2]{%
  \href{http://www.ams.org/mathscinet-getitem?mr=#1}{#2}
}
\providecommand{\href}[2]{#2}
\begin{thebibliography}{FOOW25}

\bibitem[AF20]{AndFaj:Milnestab:20}
Lars Andersson and David Fajman, \emph{Nonlinear stability of the {M}ilne model with matter}, Comm. Math. Phys. \textbf{378} (2020), no.~1, 261--298.

\bibitem[AM11]{AndMon:Milnestab:11}
Lars Andersson and Vincent Moncrief, \emph{{E}instein spaces as attractors for the {E}instein flow}, J. Differential. Geom. \textbf{89} (2011), no.~1, 1--47.

\bibitem[And04]{And:deSitteroddstab:04}
Michael~T. Anderson, \emph{{On the structure of asymptotically de Sitter and anti-de Sitter spaces}}, Adv. Theor. Math. Phys. \textbf{8} (2004), no.~5, 861--893.

\bibitem[CB52]{Cho:Einsteinlwp:52}
Yvonne~Foures (Choquet)-Bruhat, \emph{Th{\'e}or{\`e}me d'existence pour certains syst{\`e}mes d'{\'e}quations aux d{\'e}riv{\'e}es partielles non lin{\'e}aires}, Acta Math. \textbf{88} (1952), 141--225.

\bibitem[CB15]{Cho:grbook:15}
Yvonne Choquet-Bruhat, \emph{Introduction to general relativity, black holes, and cosmology}, Oxford University Press, Oxford, 2015.

\bibitem[CBG69]{ChoGer:mghd:69}
Yvonne Choquet-Bruhat and Robert~P. Geroch, \emph{Global aspects of the cauchy problem in general relativity}, Comm. Math. Phys. \textbf{14} (1969), 329--335.

\bibitem[Cic24]{Cic:deSitterscat:24}
Serban Cicortas, \emph{Nonlinear scattering theory for asymptotically de {S}itter vacuum solutions in all even spatial dimensions}, ar{X}iv:2410.01558, preprint (2024).

\bibitem[CK93]{ChrKlai:minkowskistab:93}
Demetrios Christodoulou and Sergiu Klainerman, \emph{The global nonlinear stability of the {M}inkowski space}, Princeton Mathematical Series, vol.~41, Princeton University Press, Princeton, NJ, 1993.

\bibitem[Fan25]{Fan:KdSstationarystab:25}
Allen~Juntao Fang, \emph{Nonlinear stability of the slowly-rotating {K}err-de {S}itter family}, Ann. PDE, to appear (2025).

\bibitem[FMO24]{FouMarOli:tiltedfluids:24}
Grigorios Fournodavlos, Elliot Marshall, and Todd~A. Oliynyk, \emph{Future stability of perfect fluids with extreme tilt and linear equation of state {$p=c_s^2 \rho$} for the {E}instein-{E}uler system with positive cosmological constant: The range {$1/3 < c_s^2 < 3/7$}}, ar{X}iv:2404.06789, preprint (2024).

\bibitem[FOOW25]{Fajetal:decelEulerstab:25}
David Fajman, Maximilian Ofner, Todd~A. Oliynyk, and Zoe Wyatt, \emph{Stability of fluids in spacetimes with decelerated expansion}, ar{X}iv:2501.12798, preprint (2025).

\bibitem[Fou22]{Fou:FLRWscalar:22}
Grigorios Fournodavlos, \emph{Future dynamics of flrw for the massless-scalar field system with positive cosmological constant}, J. Math. Phys. \textbf{63} (2022), no.~3.

\bibitem[FOW24]{FajOfnWya:linearduststab:24}
David Fajman, Maximilian Ofner, and Zoe Wyatt, \emph{Slowly expanding stable dust spacetimes}, Arch. Ration. Mech. Anal. \textbf{248} (2024), no.~83.

\bibitem[Fri22]{Fri:FLRW:22}
Alexander Friedmann, \emph{{\"U}ber die {K}r{\"u}mmung des {R}aumes}, Zeitschrift f{\"u}r Physik \textbf{10} (1922), no.~1, 377 -- 386.

\bibitem[Fri86]{Fri:deSitterstab:86}
Helmut Friedrich, \emph{On the existence of n-geodesically complete or future complete solutions of {E}instein’s field equations with smooth asymptotic structure}, Comm. Math. Phys. (1986), no.~4, 587--609.

\bibitem[Fri91]{Fri:deSitterMaxwellstab:91}
\bysame, \emph{On the global existence and the asymptotic behavior of solutions to the {E}instein-{M}axwell-{Y}ang-{M}ills equations}, J. Differential Geom. \textbf{34} (1991), 275--345.

\bibitem[FS24]{FouSch:KdSexpandingstab:24}
Grigorios Fournodavlos and Volker Schlue, \emph{Stability of the expanding region of {K}err de {S}itter spacetimes}, ar{X}iv:2408.02596, preprint (2024).

\bibitem[GPR23]{GroPetRin:quescientbigban:23}
Hans~Oude Groeniger, Oliver Petersen, and Hans Ringstr{\"o}m, \emph{Formation of quiescent big bang singularities}, ar{X}iv:2309.11370, preprint (2023).

\bibitem[Hag25]{Hag:waveFLRW:25}
Mahdi Haghshenas, \emph{Boundedness and decay of waves on spatially flat decelerated {FLRW} spacetimes}, ar{X}iv:2505.16794, preprint (2025).

\bibitem[Hal87]{Hal:ENSFFLRW:87}
J.~J. Halliwell, \emph{Scalar fields in cosmology with an exponential potential}, Phys. Lett. B \textbf{185} (1987), 341--344.

\bibitem[Heb99]{Heb:sob:00}
Emmanuel Hebey, \emph{Nonlinear analysis on manifolds: {S}obolev spaces and inequalities}, Courant Lecture Notes in Mathematics, vol.~5, Courant Institute of Mathematical Sciences, New York, NY; American Mathematical Society, Providence, RI, 1999.

\bibitem[HL24]{HunLuk:Burnett:24}
C{\'e}cile Huneau and Jonathan Luk, \emph{{B}urnett’s conjecture in generalized wave coordinates}, ar{X}iv:2403.03470, preprint (2024).

\bibitem[HR07]{HeiRen:powerlawstab:07}
J.~Mark Heinzle and Alan~D. Rendall, \emph{Power-law inflation in spacetimes without symmetry}, Comm. Math. Phys. \textbf{269} (2007), 1--15.

\bibitem[HS15]{HadSpe:deSitterduststab:15}
Mahir Had{\v z}i{\'c} and Jared Speck, \emph{The global future stability of the {FLRW} solutions to the dust-{E}instein system with a positive cosmological constant}, J. Hyperbolic Differ. Equ. \textbf{12} (2015), no.~01, 87--188.

\bibitem[HV18]{HinVas:KdSstationarystab:18}
Peter Hintz and Andr{\'a}s Vasy, \emph{The global non-linear stability of the {K}err–de {S}itter family of black holes}, Acta Math. \textbf{220} (2018), no.~1, 1 -- 206.

\bibitem[HV24]{HinVas:KdSexpandingstab:24}
\bysame, \emph{Stability of the expanding region of {K}err-de {S}itter spacetimes and smoothness at the conformal boundary}, ar{X}iv:2409.15460, preprint (2024).

\bibitem[KM92]{KitMae:cosmicnohair:92}
Yuichi Kitada and Kei{-}ichi Maeda, \emph{Cosmic no hair theorem in power law inflation}, Phys. Rev. D \textbf{45} (1992), 1416--1419.

\bibitem[KS81]{KlaSar:waveformulaFLRW:81}
Sergiu Klainerman and Peter Sarnak, \emph{Explicit solutions of {$\square u = 0$} on the {F}riedmann-{R}obertson-{W}alker space-times}, Annales de l’institut Henri Poincar{\'e}. Section A, Physique Th{\'e}orique \textbf{35} (1981), no.~4, 253--257.

\bibitem[LI13]{LuoIse:powerlawmag:13}
Xianghui Luo and James Isenberg, \emph{Power law inflation with electromagnetism}, Ann. Phys. \textbf{334} (2013), 420--454.

\bibitem[LR10]{LinRod:Minkowskiwavestab10}
Hans Lindblad and Igor Rodnianski, \emph{The global stability of {M}inkowski space-time in harmonic gauge}, Ann. Math. \textbf{171} (2010), no.~3, 1401--1477.

\bibitem[LVK13]{LubKro:deSitterradstab11}
Christian L{\"u}bbe and Juan~A. Valiente-Kroon, \emph{A conformal approach for the analysis of the non-linear stability of pure radiation cosmologies}, Ann. Phys. \textbf{328} (2013), 1--25.

\bibitem[Mar25]{Mar:instabilityFLRW:25}
Elliot Marshall, \emph{Instability of slowly expanding flrw spacetimes}, Class. Quantum Grav. \textbf{42} (2025), no.~10.

\bibitem[Mat77]{Mat:dissipativewave:77}
Akitaka Matsumura, \emph{Energy decay of solutions of dissipative wave equations}, Proc. Japan Acad. Ser. A Math. Sci. \textbf{53} (1977), no.~7, 232--236.

\bibitem[NR23]{NatRos:waveformulaFLRW:23}
Jos{\'e} Nat{\'a}rio and Flavio Rossetti, \emph{Explicit formulas and decay rates for the solution of the wave equation in cosmological spacetimes}, J. Math. Phys. \textbf{64} (2023), no.~3.

\bibitem[Oli16]{Oli:EulerdeSitterstab:16}
Todd~A. Oliynyk, \emph{Future stability of the {FLRW} fluid solutions in the presence of a positive cosmological constant}, Comm. Math. Phys. \textbf{346} (2016), no.~1, 293--312.

\bibitem[Rin08]{Rin:deSitterENSFstab08}
Hans Ringstr{\"o}m, \emph{Future stability of the {E}instein-non-linear scalar field system}, Invent. Math. \textbf{173} (2008), no.~1, 123--208.

\bibitem[Rin09a]{Rin:Cauchyproblembook:09}
\bysame, \emph{The {C}auchy problem in general relativity}, The European Mathematical Society ,Z{\"u}rich, 2009.

\bibitem[Rin09b]{Rin:powerlaw:09}
\bysame, \emph{Power law inflation}, Comm. Math. Phys. \textbf{290} (2009), 155--218.

\bibitem[Rin25]{Rin:cosmologyreview:25}
\bysame, \emph{Cosmology, the big bang and the {BKL} conjecture}, Comptes Rendus. M{\'e}canique \textbf{353} (2025), 53--78.

\bibitem[RS13]{RodSpe:irrEulerdeSitterstab:13}
Igor Rodnianski and Jared Speck, \emph{The stability of the irrotational {E}uler-{E}instein system with a positive cosmological constant}, J. Eur. Math. Soc. (JEMS) \textbf{15} (2013), no.~6, 2369--2462.

\bibitem[Spe12]{Spe:EulerdeSitterstab:12}
Jared Speck, \emph{The nonlinear future stability of the {FLRW} family of solutions to the {E}uler-{E}instein system with a positive cosmological constant}, Selecta Math. \textbf{18} (2012), no.~3, 633--715.

\bibitem[Spe13]{Spe:relEulerstab:13}
\bysame, \emph{The stabilizing effect of spacetime expansion on relativistic fluids with sharp results for the radiation equation of state}, Arch. Ration. Mech. Anal. \textbf{210} (2013), no.~2, 535--579.

\bibitem[Tay24]{Tay:decelFLRWstab:24}
Martin Taylor, \emph{Future stability of expanding spatially homogeneous {FLRW} solutions of the spherically symmetric {E}instein--massless {V}lasov system with spatial topology {$\mathbb{R}^3$}}, J. Math. Phys. \textbf{65} (2024), no.~2.

\bibitem[Ues80]{Ues:dissipativewave:80}
Hiroshi Uesaka, \emph{The total energy decay of solutions for the wave equation with a dissipative term}, J. Math. Kyoto Univ. \textbf{20} (1980), no.~1, 57--65.

\bibitem[Wal84]{Wal:GRbook:84}
Robert~M. Wald, \emph{General relativity}, Chicago Univ. Pr., Chicago, USA, 1984.

\end{thebibliography}
\end{document}